\documentclass[12pt]{article}
\usepackage[T1]{fontenc}
\usepackage{amsfonts}
\usepackage{amsmath}
\usepackage{amssymb}
\usepackage{amsthm}
\usepackage{bbm}
\usepackage{dsfont}
\usepackage{bm}
\usepackage{mathrsfs}
\usepackage{verbatim}
\usepackage{setspace}
\usepackage{color}
\usepackage{pdfsync}
\usepackage{svg}

\usepackage{enumitem}
\usepackage{graphicx}
\usepackage{subfigure}
\usepackage{tikz}
\usetikzlibrary{patterns}
\usepackage{placeins}
\usepackage[margin=1.0in]{geometry}

\usepackage[pdfborder={0 0 0}]{hyperref}
\hypersetup{
  urlcolor = black,
  pdfauthor = {Steven Campbell, Marcel Nutz},
  pdfkeywords = {Optimal Execution, Transient Price Impact, N-Player Game, Regularization},
  pdftitle = {Optimal Execution among N Traders with Transient Price Impact},
  pdfsubject = {Optimal Execution among N Traders with Transient Price Impact},
  pdfpagemode = UseNone
}

\usepackage[capitalize]{cleveref} %

\theoremstyle{plain}
\newtheorem{theorem}{Theorem}[section]
\newtheorem{proposition}[theorem]{Proposition}
\newtheorem{lemma}[theorem]{Lemma}
\newtheorem{corollary}[theorem]{Corollary}

\theoremstyle{definition}
\newtheorem{definition}[theorem]{Definition}
\newtheorem{remark}[theorem]{Remark}

\theoremstyle{remark}

\renewenvironment{thebibliography}[1]{%
\begin{oldthebibliography}{#1}%
\setlength{\baselineskip}{1em}
\linespread{.2}
\small
\setlength{\parskip}{0.2ex}%
\setlength{\itemsep}{.15em}%
}%
{%
\end{oldthebibliography}%
}
\newcommand{\eps}{\varepsilon}

\newcommand{\cJ}{\mathcal{J}}

\DeclareMathOperator*{\argmin}{arg\, min}

\newcommand{\mykill}[1]{}
\numberwithin{equation}{section}

\renewcommand{\boxed}{} %

\begin{document}

\title{%
Optimal Execution among $N$ Traders\\with Transient Price Impact
}
\date{\today}
\author{
  Steven Campbell%
  \thanks{
  Dept.\ of Statistics, Columbia University, sc5314@columbia.edu.}
  \and
  Marcel Nutz%
  \thanks{
  Depts.\ of Statistics and Mathematics, Columbia University, mnutz@columbia.edu.}
  }
\maketitle \vspace{-1.2em}
\begin{abstract}
We study $N$-player optimal execution games in an Obizhaeva--Wang model of transient price impact. When the game is regularized by an instantaneous cost on the trading rate, a unique equilibrium exists and we derive its closed form. Whereas without regularization, there is no equilibrium. We prove that existence is restored if (and only if) a very particular, time-dependent cost on block trades is added to the model. In that case, the equilibrium is particularly tractable. We show that this equilibrium is the limit of the regularized equilibria as the instantaneous cost parameter $\varepsilon$ tends to zero. Moreover, we explain the seemingly ad-hoc block cost as the limit of the equilibrium instantaneous costs. Notably, in contrast to the single-player problem, the optimal instantaneous costs do not vanish in the limit $\varepsilon\to0$. We use this tractable equilibrium to study the cost of liquidating in the presence of predators and the cost of anarchy. Our results also give a new interpretation to the erratic behaviors previously observed in discrete-time trading games with transient price impact.

\end{abstract}

\vspace{.3em}

{\small

\noindent \emph{Data Availability:}
No datasets were generated or analyzed during the current study.

\noindent \emph{Funding Information:}
Research partially supported by an NSERC Postdoctoral Fellowship\\ (PDF-599675-2025) and NSF Grants DMS-2407074, DMS-2106056.

\noindent \emph{Conflict of Interest:}
The authors report that there are no competing interests to declare.


}
\vspace{.6em}
\section{Introduction}\label{sec:introduction}

Transaction costs are significant for institutional-size trades; e.g., \cite{Mackintosh2022} reports 35 basis points as a typical cost to trade large cap stocks, and more for less liquid securities. The lion's share, about 30 basis points, is attributed to price impact---the fact that sizable orders push prices. This dislocation of the price is persistent but decays at a time scale relevant for execution problems (see \cite{HeyMastromatteoMuhleKarbeWebster.23} and the references therein). The most tractable model capturing this transient price impact is the Obizhaeva--Wang model \cite{ObizhaevaWang.13}; see also~\cite{FruthSchonebornUrusov.13}. Here each buy pushes the price up proportionally to the size of the trade and the dislocation reverts back exponentially over time (and similarly for sells). Following the reduced-form approach, price impact models represent the reactions of other market participants that are not explicitly accounted for. See \cref{sec:prob.formulation} for details and \cite{CarteaJaimungalPenalva.15,GatheralSchied.13,Webster.23} for more background and references. By contrast, the classic Almgren--Chriss model \cite{almgren.chriss.01} features an instantaneous price impact proportional to the trading rate which disappears immediately when trading stops, as well as permanent impact which does not decay. The instantaneous impact amounts to a quadratic cost $\eps\int_0^T v_t^2 dt$ on the trading rate $v_{t}$, where $\eps>0$. In the Almgren--Chriss model, the optimal execution problem of unwinding $x$ shares over a time interval $[0,T]$ has the TWAP strategy as its solution (assuming risk neutrality), meaning that the trading rate is constant. Whereas in the Obizhaeva--Wang model, in addition to a constant trading rate throughout the interval, block trades are placed at the initial and terminal times. The initial block trade jump-starts the resilience whereas the terminal trade arises because there is no subsequent trading that would suffer from the impact. From a stochastic control perspective, it is sometimes easier to work with absolutely continuous trading strategies. Starting with \cite{garleanu.pedersen.16,GraeweHorst.17}, numerous works have added a quadratic instantaneous cost on the trading rate to the Obizhaeva--Wang impact cost. As illustrated in \cite[Figure~1]{GraeweHorst.17}, this ``regularizes'' the problem and leads to a smoothing of the optimal execution strategy: for small instantaneous cost parameter $\eps$, the block trades are approximated by fast continuous trading. 
On a more applied note, it has been argued that block trades are not realistic in lit venues such as central limit order books. However, execution models are generally used to determine trade schedules---i.e., the approximate schedule of how the order will be worked over time (e.g., \cite{Bacidore.20})---rather than order routing. The original strategy and the smoothed one are comparable in terms of  the child order sizes implied for reasonably-sized time bins. Meanwhile, the Obizhaeva--Wang formulation often yields simpler analytic expressions, as emphasized in \cite{Webster.23}.

The present paper studies optimal execution in a competitive setting, where early works include \cite{CarlinLoboViswanathan.07, PedersenBrunnermeier.05, Schoneborn.08, SchonebornSchied.09}. We consider $N$ risk-neutral agents trading a security over a given time interval $[0,T]$. Endowed with initial inventory $x^{i}$, agent $i$ seeks to maximize their expected profit or loss and end with flat inventory. If $x^{i}$ is positive or negative, the agent has an exogenous reason to sell or buy, whereas if $x^{i}=0$, the agent is in the market only to prey on other traders. Agents interact through the security's price as each agent's actions impact the price according to the Obizhaeva--Wang model; in the absence of their actions, the price would follow a martingale. This is the natural, ``naive'' formulation of an $N$-player game extending the single-player problem of optimal execution in the Obizhaeva--Wang model. In fact, we will see that it admits no equilibrium except in trivial cases. We shall shed light on this fact and show how to restore existence by adding a particular cost to block trades, under which a tractable equilibrium emerges. More precisely, we prove that there are \emph{unique} block cost parameters leading to existence, whereas all other choices lead to non-existence (\cref{thm:equil.block.cost}). Mathematically, this ``correct'' cost can be determined from the first-order condition. While that cost initially appears as an unprincipled ad-hoc fix, the subsequent results will give a deeper meaning to it.

The aforementioned regularized version of the Obizhaeva--Wang model, with an additional quadratic instantaneous cost $\eps\int_0^T v_t^2 dt$  on the trading rate $v_{t}$, has been used successfully in the game literature. Indeed, \cite{Strehle.17} shows existence and uniqueness of a Nash equilibrium. The derivation highlights the mathematical significance of the regularization: the first-order condition for the equilibrium boils down to a Fredholm equation of the \emph{second} kind; a type of equation that is well-posed under general conditions (in contrast to the first kind appearing in~\cite{GatheralSchiedSlynko.12}). This equation also leads to an expression for the equilibrium, which however still requires numerical evaluation. In the present paper, while not our primary objective, we add to this literature by providing the solution in fully closed form (see \cref{thm:equil.liq.constr} for this model with liquidation constraint, and \cref{thm:equil.term.pen} for an additional  model allowing incomplete execution). The derivation is admittedly a tour de force, but it enables a fine downstream analysis (in addition to making numerical implementation trivial). We mention that the work of \cite{Strehle.17} has been generalized in several directions, such as incorporating alpha signals \cite{NeumanVoss.23}, alpha signals and non-exponential decay kernels~\cite{AbiJaberNeumanVoss.24} or self-exciting order flow \cite{FuHorstXia.22}.  All these works make crucial use of the regularization by instantaneous cost. While the fixed parameter $\eps>0$ can be arbitrarily small, it is not obvious what exactly this regularization is approximating. The present work aims to shed light on that. A separate stream of literature restricts trading to a discrete set of dates, which can be seen as a different type of regularization. This literature is closely related to our main results, hence discussed in more detail in \cref{se:furtherLiterature} below.

In the single-player execution problem discussed in the first paragraph, the solution of the regularized version converges to the unregularized one for $\eps\to0$, as one would expect (see~\cite{GraeweHorst.17, HorstKivman.24}). Moreover, one can check that the regularizing instantaneous cost $\eps\int_0^T v_t^2 dt$ (with $v$ depending on $\eps$) converges to zero as $\eps\to0$. The game turns out to be markedly different, as can already be gleaned from the aforementioned non-existence of equilibria for the naive formulation. In fact, the instantaneous cost $\eps\int_0^T v_t^2 dt$ of a typical agent in the regularized equilibrium no longer converges to zero. Financially, this is testament to the competition which causes more aggressive trading in the game case. Mathematically, it suggests that the correct limiting game incorporates an additional cost relative to the naive formulation. Indeed, we show that the instantaneous cost converges exactly to the block cost that uniquely gives existence of equilibria, thus explaining the seemingly ad-hoc cost coefficient (\cref{thm:small.eps.lim}). 

Analogously to the single-player case, the equilibrium of the limiting model with block costs is very tractable. We use its expression for the equilibrium impact cost to compare with the single-player case. First, we study the cost of anarchy; i.e., the increase in cost due to competition relative to the strategy that a central planner would use (\cref{se:cost.anarchy}). Second, we show how the presence of $N-1$ predators (traders with zero initial inventory) increases the cost for an agent who needs to unwind inventory (\cref{se:cost.predation}). In both cases, the cost increases with $N$, but the increase tapers out as $N$ gets large.

\subsection{Further Related Literature}\label{se:furtherLiterature}

As mentioned above, restricting trading to a discrete set of dates can be seen as a kind of regularization of a continuous-time model. The thesis \cite{Schoneborn.08} was the first to consider a game in the Obizhaeva--Wang model, for $N=2$ traders. When trading at discrete dates, it was observed that equilibrium exists but consists of erratic strategies (e.g., an agent might sell the entire initial inventory at the first date and buy it back the next date). Moreover, in the high-frequency limit where the gaps between the trading dates tend to zero, the equilibria oscillate and do not converge. This phenomenon is further studied (mostly numerically) in \cite{SchiedZhang.19} for more general decay kernels. Here it is shown that introducing additional trading costs can dampen the oscillations and reduce equilibrium trading costs, which is interpreted as friction providing protection against predatory trading. In a similar framework but focusing on the Obizhaeva--Wang model (and still with $N=2$ traders), \cite{SchiedStrehleZhang.17} provides a detailed analytic study which is closely related to the present results. It is shown that oscillations are suppressed if sufficiently large additional trading costs are introduced, and in that case the high-frequency limit exists. Intuitively, high frequency is akin to small instantaneous cost. In \cite{SchiedStrehleZhang.17} the authors further consider a model with continuous trading and additional block costs. It is shown that an equilibrium exists when a particular block cost is charged at the initial and terminal time, and in that case, the equilibrium coincides with the aforementioned high-frequency limit. The equilibrium is precisely the one of \cref{thm:equil.block.cost} in the particular case $N=2$. In \cite{SchiedStrehleZhang.17} it is not discussed directly if the two equilibrium block costs arise as the limits of the added costs in the discrete model, but this can be conjectured based on our \cref{thm:small.eps.lim}. As can be seen in \cref{thm:equil.block.cost}, the case $N=2$ is unique in that the \emph{same} block cost is charged at the initial and terminal time, whereas for $N>2$ the initial cost is larger and depends on~$N$. In particular, $N=2$ is the only case that allows for a cost that is not time-dependent.

The follow-up work \cite{LuoSchied.19} studies an $N$-player game related to the one of \cite{SchiedStrehleZhang.17}. Instead of having a liquidation constraint, agents unwind their inventory because they have a (mean-variance or exponential) utility function: holding the martingale asset causes disutility, incentivizing liquidation. The discrete-time equilibrium is obtained analytically whereas for the high-frequency limit, only a numerical study is performed. Based on the numerics, it is conjectured that two different critical values for the additional trading cost suppress oscillations in the particular cases where all initial inventories are equal or add up to zero. On the other hand, there is no value that works for general inventories, suggesting that the limit in this model does not exist. The two values coincide with the ones in our analytic result (\cref{thm:equil.block.cost}). We infer from our result that the key to resolving the non-existence is a \emph{time-dependent} additional cost, which has not been considered before.

After explaining that the proof technique of \cite{SchiedStrehleZhang.17} does not extend to $N>2$ players, the theoretical asymptotic  results in \cite{LuoSchied.19} are stated in a slightly different model. Instead of a high-frequency limit, trading takes place at integer dates with time horizon $T=\infty$. Absence of a (finite) horizon circumvents block trades at $T$ and---from the perspective of our results---the necessity to identify a different block cost at~$T$. The authors show that an equilibrium exists for a particular choice of cost parameter (indeed the same as the additional cost at $t=0$ in our model). Several conjectures about non-existence for different values are stated, which are all confirmed by our sharp results. The connection with the prelimit (i.e., sending $T\to\infty$) is not made in \cite{LuoSchied.19}, but we can readily conjecture the results based on ours. In fact, in retrospect, our results suggest that it should be possible to obtain a full extension of the high-frequency limit of \cite{SchiedStrehleZhang.17} to $N>2$ players without sending $T\to\infty$. Namely, introducing a time-dependent additional cost, we would expect convergence to the equilibrium identified in \cref{thm:equil.block.cost}. 

Finally, we discuss some related works that have a different focus. In \cite{NeumanVoss.23}, the model is similar to the one in \cref{sec:N.inst.cost}, with a fixed instantaneous cost, but traders have additional alpha signals, meaning that the unaffected price is not a martingale. Moreover, traders have an additional quadratic inventory cost. The equilibrium has a similar structure as the one in \cref{sec:N.inst.cost}; however, the strategies are not found in fully closed form (only up to solving a quartic equation) and their arguments require the inventory cost.  %
The focus of~\cite{NeumanVoss.23} is on incorporating the trading signal, with additional results on the mean-field limit where the number of traders tends to infinity, whereas we are interested in vanishing instantaneous costs. 

Also in the mean-field setting, but using temporary and permanent price impact, \cite{CasgrainJaimungal.19} considers traders with partial information: they observe prices, but not other traders' strategies. Moreover, there are heterogeneous subpopulations of traders with different risk parameters, and in \cite{CasgrainJaimungal.20}, they additionally have heterogeneous beliefs. The mean-field equilibrium strategies are determined and have an affine form. These works focus on partial information and heterogeneous beliefs, and use the mean-field setting where each agent is infinitesimally small and hence has no price impact. By contrast, we solve an $N$-player game with transient impact and  study vanishing instantaneous costs.

\vspace{1em}

The remainder of this paper is organized as follows. \Cref{sec:prob.formulation} formulates in detail the stochastic games to be considered: with instantaneous cost on trading rate and either a penalty on leftover inventory ($A$) or full liquidation constraint ($A'$), and the limiting model~($B$) with block cost (but no cost on the trading rate). We also include the uniqueness of Nash equilibria in \cref{subse:uniqueness}; this statement applies to all formulations under consideration. \Cref{sec:N.inst.cost} summarizes the main results for the models with instantaneous cost. First, \cref{subse:terminalPenalty} describes the equilibrium for game $A$ with penalty on leftover inventory. This problem can be approached directly by variational methods. Then, \cref{sec:N.inst.cost.liq.constr} treats game $A'$ where full liquidation is imposed as a constraint. Mathematically, we construct the equilibrium of the constrained problem from the former by letting the penalty tend to infinity. Both sections provide closed-form expressions for the equilibrium trading strategies and costs. \Cref{sec:N.block.cost} describes the model $B$ with block cost, giving a full characterization of existence (or non-existence) of equilibria depending on the block cost parameters. When it exists, the equilibrium and its costs are found in closed form. \Cref{sec:id.limit} shows how the model with block costs arises as the limit of model $A'$ when the instantaneous cost tends to zero. \Cref{se:cost.anarchy.predation} discusses the cost of anarchy and the cost of predation. \Cref{sec:conclusion} concludes and comments on follow-up research. \Cref{se:notation} details some (reasonably standard) notational conventions while \cref{app:opt.theory} summarizes properties of Gateaux derivatives and $\Gamma$-convergence that are used in the proofs. As the main proofs are lengthy, they are not included in the body of the paper. \Cref{app:proofsProblemformulationAndUniqueness} contains the proofs for \cref{sec:prob.formulation},
\cref{app:N.inst.cost} contains the proofs for \cref{sec:N.inst.cost}, \cref{app:N.block.cost} contains the proofs for \cref{sec:N.block.cost}, and 
\cref{app:proofs.for.sec.id.limit} contains the proofs for \cref{sec:id.limit}. (There is no appendix for \cref{se:cost.anarchy.predation} as the calculations are straightforward.) Last but not least, \cref{app:constants} contains \cref{tab:constants}, a collection of constants that we introduced to shorten the otherwise unwieldy formulas in the main results.

\section{$N$-Player Game Formulation}\label{sec:prob.formulation}

We consider a market on a filtered probability space $(\Omega,\mathcal{F},\mathbb{F},\mathbb{P})$ satisfying the usual conditions. In this market there are $N$ traders and a single asset. We index the traders by $i\in\{1,
\dots, N\}$ and denote their inventory processes by $X^{i}=(X^{i}_t)_{t\geq0}$, where $X^i_t$ indicates the number of shares held by trader~$i$ at time~$t$. Each trader $i$ is endowed with initial holdings $X_{0-}^{i}=x^{i}\in\mathbb{R}$ to be unwound by the common terminal time $T>0$. More precisely, we shall consider problems where a liquidation constraint $X_{T}^{i}=0$ is enforced as well as problems where terminal inventory is merely penalized by a cost.

\begin{definition}\label{def:admissible.X}
    We say that $X^{i}=(X^{i}_t)_{t\geq0}$ is an admissible inventory process for trader $i$ if:
    \begin{enumerate}
        \item $X^{i}$ is c\`adl\`ag and predictable.
        \item The paths $t\mapsto X^{i}_t$ have ($\mathbb{P}$-essentially) bounded total variation.
        \item $X^{i}_{0-}=x^{i}$ and $X^{i}_t$ is constant for $t\geq T$.
    \end{enumerate}
\end{definition}

We assume that the unaffected asset price---which is the price that would obtain if the $N$ agents did not trade---evolves according to a c\`adl\`ag square-integrable martingale, $P=(P_t)_{t \geq0}$. %
To describe the actual price, we define the impact process $I=(I_t)_{t\geq0}$ through the 
generalized Obizhaeva--Wang dynamics
\begin{equation*}%
dI_t = -\beta I_t dt + \lambda \sum_{i=1}^NdX^{i}_t, \ \ \ I_{0-} = 0,
\end{equation*}
for push and resilience parameters $\lambda,\beta>0$. Hence,
\begin{equation*}
    I_t =\lambda \int_0^te^{-\beta (t-s)} \sum_{i=1}^NdX_s^{i},
\end{equation*}
and we define the affected price $S=(S_t)_{t\geq0}$ through
\begin{equation*}
S_t = P_t + I_t, \ \ \ t\geq0.
\end{equation*}
Here and throughout the paper, $\int_{a}^{b}:=\int_{[a,b]}$ (cf.\ \cref{se:notation}), and the ``a.s.'' qualifier is suppressed.

Let $\Delta X_t^{i} :=X^{i}_t - X_{t-}^{i}$. 
We define the impact cost associated with the admissible inventory processes $\boldsymbol{X} = (X^{1},\dots,X^{N})$ in the following way. For trader $i$, the net proceeds or outlays from trading are
\begin{equation}\label{eq:OWcost}\int_0^T S_{t-} dX_t^{i}+\frac{1}{2}\sum_{t\in[0,T]} \Delta S_t  \Delta X_t^{i}.
\end{equation}
This says that continuous trading at time $t$ transacts at the price $S_{t-}$ while a block trade of size $\Delta X_t^{i}$ additionally realizes half the price dislocation at $t$ and has a final execution price of
\begin{equation*}S_{t-} +\frac{1}{2}\Delta S_t=\frac{1}{2}( S_{t-}+S_{t}).
\end{equation*}
This can be interpreted as trader $i$ obtaining the average execution price of all trades happening at~$t$ (rather than the marginal price at $t-$). Implicitly, the above definitions also govern what happens when several agents trade at the same time. In some works, including~\cite{SchiedStrehleZhang.17} concerned with the case $N=2$, ties are explicitly broken randomly: when two agents place a block trade at the same time, a coin flip decides which trade is settled first. The present definitions are equivalent (in terms of expected costs) but shorter to write, especially in the $N$-player case. See also \cite{Webster.23} for a related discussion.

In some of the problems below we allow for the possibility that a trader's inventory has not been entirely unwound by time~$T$. As a result, we must account for the change in value of their holdings over $[0,T]$. To this end, we add to the execution costs the change in the marked-to-market value of their holdings,
\begin{equation}\label{eq:markToMarket}
    X_{0-}^iP_{0-}-X_T^iP_T.
\end{equation}
As is standard in the literature (e.g., \cite{NeumanVoss.22,NeumanVoss.23}), we use the unaffected price $P$ for inventory valuation.\footnote{This convention exists, at least in part, to remedy the loss of convexity (see Lemma \ref{lem:strict.convex}) that can arise when the reference price is instead taken to be $S_T$.} 
If the terminal inventory is constrained to be zero this additional accounting amounts to adding a constant to the trader's cost, hence has no effect on their optimal strategy.

On top of the impact cost, we also consider additional trading and terminal costs. Let $\chi_{E}$ be the characteristic function of a set $E\in\mathcal{F}$,
\[\chi_{E}(\omega)=\begin{cases}
    \infty & \omega\in E,\\
    0 & \mathrm{otherwise},
\end{cases}\] 
and write $\{dX^i\ll dt\}$ for the set of $\omega\in\Omega$ on which the measure associated with the bounded variation function $t\mapsto X^i_t(\omega)$ is absolutely continuous with respect to Lebesgue measure. 

\paragraph{Cost $A$.} In our first problem formulation, indexed by the symbol~$A$, we penalize fast (absolutely continuous) trading with an ``instantaneous cost'' and levy  a terminal inventory penalty. The induced cost is
\begin{equation} \label{eq:addl.cost.A} 
C_{A}(X^{i}):=\frac{\varepsilon}{2}\int_0^T (\dot{X}^{i}_t)^2dt+\chi_{\{dX^i\ll dt\}^c}+\frac{\varphi}{2} (X_T^{i})^2, %
\end{equation}
where $\varepsilon, \varphi>0$ and $\dot{X}^{i}_t$ is the derivative\footnote{This is assured to exist $dt$-almost everywhere.} of $t\mapsto X^{i}_t$. The characteristic function means that discontinuous or singular continuous controls\footnote{i.e., controls $t\mapsto X_t^{i}$ whose Lebesgue decomposition on any $[0,t_0]\subset \mathbb{R}_+$ has singular continuous or purely discontinuous components.} incur infinite costs.

\paragraph{Cost $A'$.} In the second formulation, indexed by $A'$, we modify the above to enforce the hard constraint that all inventory be liquidated by $T$,
\begin{equation}C_{A'}(X^{i}):=\frac{\varepsilon}{2}\int_0^T (\dot{X}^{i}_t)^2dt+\chi_{\{dX^i\ll dt\}^c}+\chi_{
\{X_T^{i} \not=0\}}. %
\end{equation}
Formally, this corresponds to setting $\varphi=\infty$ in~\eqref{eq:addl.cost.A}.

\paragraph{Cost $B$.} In our final variant, indexed $B$, trading need not be absolutely continuous. We charge a (deterministic but possibly time-dependent) cost $\vartheta_t/2\geq 0$ on block trades and enforce inventory liquidation,
\begin{equation}\label{eqn:CB}C_B (X^{i}):=\frac{1}{2}\sum_{t\in[0,T]} \vartheta_t(\Delta X_t^{i})^2+\chi_{\{X_T^{i} \not=0\}}.
\end{equation}

To unify the statements below, we use the symbol $\cdot$ as a placeholder for the type of cost ($A,A'$ or $B$). 
Combining the additional cost $C_\cdot$ with the impact cost~\eqref{eq:OWcost} and the value~\eqref{eq:markToMarket} of the terminal inventory, trader $i$ has the following objective function if we fix the actions $\boldsymbol{X}^{-i}=(X^{1},\dots,X^{i-1},X^{i+1},\dots,X^{N})$ of the other players,
\begin{align}\label{eqn:orig.obj} J_\cdot(X^{i};&\boldsymbol{X}^{-i})
=\mathbb{E}\left[\int_0^T S_{t-} dX_t^{i}+\frac{1}{2}\sum_{t\in[0,T]} \Delta S_t \Delta X_t ^{i} +(X_{0-}^iP_{0-}-X_T^iP_T)+C_\cdot(X^{i})\right].
\end{align}
The aim of trader $i$ is to minimize $J_{\cdot}(X^{i};\boldsymbol{X}^{-i})$. The next proposition uses the martingale property of the unaffected price and provides a more explicit formula for this quantity.

\begin{proposition}\label{prop:obj.func.rep}
    The objective function $J_
    \cdot(X^{i};\boldsymbol{X}^{-i})$ can be written
    \begin{align}\label{eqn:obj.impact.rep}
    J_\cdot(X^{i};\boldsymbol{X}^{-i})&=\mathbb{E}\left[\int_0^T I_{t-} dX_t^{i}+\frac{1}{2}\sum_{t\in[0,T]}  \Delta I_t \Delta X_t ^{i}+C_\cdot(X^{i})\right]\\
    &=\lambda\mathbb{E}\Bigg[\frac{1}{2}\int_0^T\int_0^T e^{-\beta |t-s|} dX_s^{i} dX_t^{i} + \int_0^T\int_0^{t-} e^{-\beta (t-s)} \sum_{j\not=i}dX_s^{j} dX_t^{i} \nonumber \\
    &\quad \quad \quad \quad + \frac{1}{2}\sum_{j\not=i}\sum_{t\in[0,T]} \Delta X_t^{j}\Delta X_t ^{i} \Bigg ]+\mathbb{E}\bigg[C_\cdot(X^{i})\bigg].\label{eqn:obj.schied.rep}
\end{align}
\end{proposition}
The proof is given in Appendix \ref{app:pf.obj.func.rep}. Next, we formally define the Nash equilibria to be studied below.

\begin{definition} \label{def:nash.equilibrium}A strategy profile $\boldsymbol{X}^*=(X^{*,1},\dots,X^{*,N})$ is a \emph{Nash equilibrium} for $J_\cdot$ if:
\begin{enumerate}
    \item $X^{*,i}$ is admissible for all $i=1,\dots, N$.
    \item For all $i=1,\dots,N$ and all admissible strategies $Z$ for trader $i$,
    \begin{align}\label{eq:Nash.condition}
      J_\cdot(Z;\boldsymbol{X}^{*,-i})\geq J_\cdot(X^{*,i};\boldsymbol{X}^{*,-i}).
    \end{align} 
\end{enumerate}
We say that $\boldsymbol{X}^*$ is a Nash equilibrium \emph{in the class of deterministic strategies} if $X^{*,i}$ is admissible and deterministic for all $i$, and~\eqref{eq:Nash.condition} holds for all admissible and deterministic~$Z$. Similar terminology applies for other subsets of admissible strategies.
\end{definition}

\subsection{Uniqueness}\label{subse:uniqueness}

The remainder of this section focuses on the uniqueness of equilibria (regardless of existence) and is valid for all three types of cost. 

\begin{proposition}\label{prop:nash.eq.unique}
    There is at most one Nash equilibrium. %
\end{proposition}

The proof is detailed in Appendix \ref{app:pf.nash.eq.unique} for completeness. While technical, it essentially extends the arguments in \cite[Proposition~4.8]{SchiedStrehleZhang.17} to our setting. A key ingredient is the strict convexity of the objective.

\begin{lemma}\label{lem:strict.convex}
    For any admissible $\boldsymbol{X}^{-i}$,  the objective $J_\cdot(\,\cdot\,; \boldsymbol{X}^{-i})$ is strictly\footnote{An extended real-valued convex function $F$ is called strictly convex if it is strictly convex on its domain
$\mathrm{dom}(F)=\{h:F(h)<\infty\}$.} convex in its first argument.
\end{lemma}

The proof of Lemma~\ref{lem:strict.convex} is analogous to \cite[Lemma 4.7]{SchiedStrehleZhang.17} and omitted. 
The last result of this section states that an equilibrium in the class of deterministic strategies is also an equilibrium in the larger class of all admissible strategies. To wit, if~\eqref{eq:Nash.condition} holds for deterministic competitors~$Z$, then it automatically holds for general~$Z$.

\begin{lemma}\label{lem:nash.eq.det}
    A Nash equilibrium in the class of deterministic strategies is a Nash equilibrium (in the class of admissible strategies).
\end{lemma}
The proof is similar to  \cite[Lemma 4.9]{SchiedStrehleZhang.17} and omitted. Lemma~\ref{lem:nash.eq.det} allows us to restrict to deterministic strategies in parts of the subsequent analysis: if we can find a unique deterministic equilibrium, then Lemma~\ref{lem:nash.eq.det} assures that it is also the unique equilibrium in the larger class of admissible strategies. (It does not assure the reverse: novel arguments will be necessary to see that non-existence of deterministic equilibria implies non-existence in general.)

\begin{remark}
    If $P=(P_t)_{t\geq0}$ were a more general semimartingale (as in e.g., \cite{NeumanVoss.23}), we could not completely remove the unaffected price dynamics from the objective as in Proposition \ref{prop:obj.func.rep}. In that case, solutions are typically stochastic, which precludes much of the explicit analysis in the following sections that relies on Lemma \ref{lem:nash.eq.det}.
\end{remark}

\begin{remark}
    In keeping with most of the extant work in this area, our analysis is conducted under \emph{full information}, meaning objectives, parameters, and initial inventories are common knowledge. An alternative is partial information, in which problem features like the trader inventories are hidden and must be inferred from price dynamics. This represents a substantial increase in complexity since traders also face a filtering problem for their competitors' states. While certainly of interest, this is outside of the scope of the current work.
\end{remark}

\section{Equilibria with Instantaneous Costs}\label{sec:N.inst.cost}

We begin by finding the unique equilibria in the cases where the additional costs are given by $C_{A}(\cdot)$ and $C_{A'}(\cdot)$. That is, in addition to price impact, trading is subject to an instantaneous cost with coefficient $\eps>0$, and there is either a cost on terminal inventory or terminal inventory is required to be zero.

As noted above, we can begin our search by focusing on deterministic strategies. From the form of the costs it is clear that we can further focus on inventory processes with $\dot{X}_t^{i}\in L^2[0,T]$ for all $i$. Therefore, it suffices to take
\begin{equation}\label{eqn:inv.ac}
	X_t^{i}=x^{i}+\int_0^t v_s^{i}ds,  \ \ \ t\in[0,T], \ \ \ i=1,\dots,N,
\end{equation}
for $v^{i}\in L^2[0,T]$. This allows us to parametrize the strategy profile through the auxiliary controls $\boldsymbol{v}=(v^{1},\dots,v^{N})$. Then, by \eqref{eqn:obj.impact.rep} we can write $J_\cdot(X^{i};\boldsymbol{X}^{-i})$ in terms of the impact process which now has the simplified form
\begin{equation*}
	I_t = \lambda\int_0^t e^{-\beta(t-s)}\sum_{j=1}^N v^{j}_s ds, \ \ \ t\in [0,T].
\end{equation*}
As a result, the optimization of the objective \eqref{eqn:orig.obj} can be recast as minimizing
\begin{equation}\label{eqn:cJ.A}
    \mathcal{J}_A(v^{i};\boldsymbol{v}^{-i})=\int_0^T I_t v_t^{i} +\frac{\varepsilon}{2}(v_t^{i})^2dt+\frac{\varphi}{2}(X_T^{i})^2,
\end{equation}
or, respectively,
\begin{equation}\label{eqn:cJ.A'}
	\mathcal{J}_{A'}(v^{i};\boldsymbol{v}^{-i})=\int_0^T I_t v_t^{i} +\frac{\varepsilon}{2}(v_t^{i})^2dt+\chi_{\{X_T^{i}\not=0\}},
\end{equation}
over $v^{i}\in L^2[0,T]$. Notably, the transformed objectives remain (strictly) convex on their domains. We will leverage the Hilbert space structure afforded by this reparametrization to arrive at a complete characterization of the equilibria. For the proofs of the statements in the remainder of this section, see Appendix \ref{app:N.inst.cost}.

\subsection{Terminal Inventory Penalty}\label{subse:terminalPenalty}

Variational arguments in $L^2[0,T]$ allow us to characterize the equilibrium for the cost $\mathcal{J}_A$ in terms of the (unique) solution to a $2N+1$ dimensional system of linear homogeneous ordinary differential equations (ODEs).
\begin{lemma}\label{lem:term.pen.ODE}
    The strategy profile $\boldsymbol{v}$ defines a Nash equilibrium for $\mathcal{J}_A$ if and only if it forms, along with $I$ and auxiliary processes $Y^{1},\dots, Y^{N}$, a solution to the ODE system
	\begin{align*}
		\dot{I}_t&=-\beta I_t+\lambda\sum_{i=1}^Nv_t^{i}, \\
		\dot{Y}^{i}_t&=\beta Y^{i}_t-\lambda v^{i}_t, \ \ \ i=1,\dots,N,\\
		\dot{v}_t^{i}&=\varepsilon^{-1}\left[\beta I_t-\beta Y^{i}_t-\lambda\sum_{j\not=i}v_t^{j}\right], \ \ \ i=1,\dots,N,
	\end{align*}
	subject to the initial and terminal conditions
	\begin{align*}
		I_0&=0, \ \ \ Y_T^{i}=0,\ \ \ 
		v_T^{i}= -\varepsilon^{-1}\left[\varphi X_T^{i}+I_T\right],  \ \ \ i=1,\dots, N.
	\end{align*}
\end{lemma}

While admittedly a tour de force (deferred to Section~\ref{se:proof.of.equil.term.pen}), it turns out that this system can be solved in fully closed form. Thus we arrive at an explicit characterization of the Nash equilibrium for $\mathcal{J}_A$ and hence, through \eqref{eqn:inv.ac}, also for the general objective $J_A$. It turns out that the strategies depend linearly on the mean starting inventory,
\begin{equation*}
	\overline{x} = \frac{1}{N} \sum_{i=1}^N x^{i},
\end{equation*}
and the individual deviations from the mean. Since the expression for the solution is lengthy, the following theorem makes use of simplifying constants $z_1,z_2,z_3,\gamma_1,\gamma_2, \rho_0$, and $\rho_{\pm}$ whose definitions can be found in Table \ref{tab:constants}.

\begin{theorem}\label{thm:equil.term.pen}
    There is a unique Nash equilibrium for $J_A$. Its equilibrium strategy profile $\boldsymbol{X}^*=(X^{*,1},\dots,X^{*,N})$ is 
    \begin{equation*}
    	X^{*,i}_t= f_t (x^{i}-\overline{x})+g_t \overline{x}, \ \ \ t\in[0,T], \ \ \ i=1,\dots,N,
    \end{equation*}
    where
     \begin{align*}
    	f_t &=1-\frac{\left[\beta t+\frac{\lambda (e^{z_{3} t}-1)}{\varepsilon z_3 e^{z_3T}} \right]\varphi }{\varepsilon z_3 + \left[\beta T+\frac{\lambda(e^{z_3 T}-1)}{\varepsilon z_3e^{z_3 T}}\right]\varphi},\\
    	g_t &=1-\frac{\left[\beta\rho_{-}t+\frac{e^{z_1 t}-1}{z_1}-\frac{\gamma_1}{\gamma_2}\frac{(e^{z_2 t} -1)}{z_2}\right]\varphi}{\varepsilon(\rho_0+\beta\rho_{-})+\lambda N(\rho_{+}+\rho_{-})+ \left[\beta\rho_{-}T+\frac{e^{z_1 T}-1}{z_1}-\frac{\gamma_1}{\gamma_2}\frac{e^{z_2 T}-1}{z_2}\right]\varphi}.
    \end{align*}
\end{theorem}

We note that the equilibrium strategy profile simplifies in two special cases, a scheme that will be repeated throughout the paper: For symmetric inventories $x^i\equiv x$, the mean-deviation $x^i-\overline{x}$ vanishes and the profile simplifies to $g_t\overline{x}$. The other extreme is the case $\sum_i x^i=0$ of zero net supply, where the mean vanishes and the profile simplifies to $f_t x^i$.

We can use this explicit solution to compute the equilibrium cost. We report the result in terms of additional simplifying constants $\psi,\xi, \mathfrak{p}$, and $\mathfrak{h}_j$ ($j=1,\dots,5$) whose definitions are provided in Table \ref{tab:constants}.

\begin{corollary}\label{cor:eq.cost.term.pen}
The equilibrium cost for the traders can be written as
    \begin{align*}
        J_A(X^{*,i};\boldsymbol{X}^{*,-i})= \int_0^TI_{t-}dX^{*,i}_t+\frac{\varepsilon}{2}\int_0^T(\dot{X}^{*,i}_t)^2dt+\frac{\varphi}{2} (X_T^{*,i})^2, \ \ \ i=1,\dots,N,
    \end{align*}
in terms of the equilibrium impact process\footnote{By using the form of the constants in Table \ref{tab:constants}, it can be verified that $I_0=0$.}
\begin{equation*}I_t=\frac{-N\lambda\left[\rho_{-}+\frac{1}{z_1+\beta}e^{z_1t} -\frac{\gamma_1}{ \gamma_2}\frac{1}{z_2+\beta}e^{z_2t}\right]\varphi}{\varepsilon(\rho_0+\beta\rho_{-})+\lambda N(\rho_{+}+\rho_{-})+ \left[\beta\rho_{-}T+\frac{e^{z_1 T}-1}{z_1}-\frac{\gamma_1}{\gamma_2}\frac{e^{z_2 T}-1}{z_2}\right]\varphi}\overline{x}, \ \ \ t\in[0,T].
\end{equation*}
Stated explicitly, we obtain the equilibrium impact cost
\begin{align*}
    \int_0^TI_{t-}dX^{*,i}_t &=\frac{\lambda N\varphi^2}{\varepsilon^2}\left[\frac{\mathfrak{h}_1}{\psi^2}\overline{x}^2+\frac{\mathfrak{h}_2}{\xi\psi}(x^{i}-\overline{x})\overline{x}\right],
\end{align*}
equilibrium instantaneous cost
\begin{align*}
    \frac{\varepsilon}{2}\int_0^T(\dot{X}^{*,i}_t)^2dt&=\frac{\varphi^2}{\varepsilon}\left[\frac{\mathfrak{h}_3}{2 \psi^2}\overline{x}^2+\frac{\mathfrak{h}_4}{2 \xi^2}(x^{i}-\overline{x})^2+\frac{\mathfrak{h}_5}{ \xi\psi}\overline{x}(x^{i}-\overline{x})\right],
\end{align*}
and equilibrium terminal penalty
\begin{align*}
\frac{\varphi}{2} (X_T^{*,i})^2
&=\frac{\varphi}{2}\left[\frac{\mathfrak{p}^2}{\psi^2}\overline{x}^2+\frac{ z_3^2}{ \xi^2}(x^i-\overline{x})^2+\frac{ 2z_3\mathfrak{p}}{ \xi\psi}(x^i-\overline{x})\overline{x}\right].
\end{align*}
\end{corollary}

The functions $f$ and $g$ appearing in Theorem \ref{thm:equil.term.pen} are positive and decreasing. The interpretation is that an ``average'' trader with initial inventory $\overline{x}$ unwinds their position over time. Traders whose initial endowments differ from the average adjust their strategy by a decreasing function of this discrepancy. For example, in a market where $x^i\geq0$ for all~$i$ (and so, $\overline{x}\geq0$) traders with less initial inventory than the average (i.e., $x^i\leq \overline{x}$) liquidate their inventory at a slower rate (and may even purchase more inventory) since $\dot{f}_t(x^i-\overline{x})\geq 0$. This is reflected in Corollary \ref{cor:eq.cost.term.pen} where the impact cost is larger (or smaller) than that of the ``average'' player based on the discrepancy of the trader's initial inventory. Intuitively, having less inventory to liquidate leads to a lower impact cost.

An illustration of equilibrium strategies and costs when $N=3$ is provided in Figure \ref{fig:equil.conv.phi}. In that example two of the traders need to liquidate a positive position, while one trader starts with zero inventory and acts as a ``predator.'' We see that both the cost of the liquidators and the profit of the predator\footnote{Note the sign convention; here a negative cost is a profit.} are decreasing in the resilience $\beta$. The strategy of the predator is to trade in the same direction as the liquidators before unwinding near the terminal time. At maturity there is a curious phenomenon where the predator actually holds positive inventory. This is due to the role of a \emph{finite} terminal penalty $\varphi$ and the economic interpretation of marking to the \emph{unaffected} price. Namely, since the impact of the liquidators has depressed the market price, there is an incentive to bias towards a \emph{long} position. This is also reflected in the fact that the liquidators do not close out all of their initial positions.

\subsection{Liquidation Constraint}\label{sec:N.inst.cost.liq.constr}

To find the unique equilibrium associated with the objective $J_{A'}$ enforcing full liquidation at~$T$, we argue that it coincides with the limit of the equilibrium in Theorem \ref{thm:equil.term.pen} as the penalty on terminal inventory tends to infinity, $\varphi\uparrow\infty$.

\begin{lemma}\label{lem:conv.equilibria}
If the Nash equilibrium $\boldsymbol{X}^*$ from Theorem \ref{thm:equil.term.pen} converges in $H^{1}[0,T]^{\times N}$ as $\varphi\uparrow\infty$, then the limit is a Nash equilibrium for $J_{A'}$ and the equilibrium costs in Corollary~\ref{cor:eq.cost.term.pen} converge to the equilibrium costs for $J_{A'}$.
\end{lemma}

The proof in \cref{se:proof.of.conv.equilibria} is based on the $\Gamma$-convergence of $\mathcal{J}_{A}(\cdot;\boldsymbol{v}^{*,-i}(\varphi))$ to $\mathcal{J}_{A'}(\cdot;\tilde{\boldsymbol{v}}^{-i})$ as $\varphi\uparrow\infty$, where $\boldsymbol{v}^{*}(\varphi)$ denotes the equilibrium for finite $\varphi$ and $\tilde{\boldsymbol{v}}$ its limit. 

Using \cref{lem:conv.equilibria}, we can deduce a characterization of the equilibrium by passing to the limit in Theorem~\ref{thm:equil.term.pen}. As mentioned in the introduction, our result provides a closed-form solution to the game previously studied in \cite{Strehle.17}. See Table \ref{tab:constants} for the definitions of the simplifying constants.

\begin{theorem}\label{thm:equil.liq.constr}
	There is a unique Nash equilibrium for $J_{A'}$. Its equilibrium strategy profile $\boldsymbol{X}^*=(X^{*,1},\dots,X^{*,N})$ is 
	\begin{equation*}
		X^{*,i}_t= \mathfrak{f}_t (x^{i}-\overline{x})+\mathfrak{g}_t \overline{x}, \ \ \ t\in[0,T], \ \ \ i =1,\dots,N,
	\end{equation*}
	where
	\begin{equation*}
		\mathfrak{f}_t =1-\frac{\beta t+\frac{\lambda (e^{z_{3} t}-1)}{\varepsilon z_3 e^{z_3T}}}{\beta T+ \frac{\lambda(e^{z_3 T}-1)}{\varepsilon z_3e^{z_3 T}}} , \ \ \ 
		\mathfrak{g}_t =1-\frac{\beta \rho_{-}t+\frac{e^{z_1 t}-1}{z_1}-\frac{\gamma_1}{\gamma_2}\frac{e^{z_2 t} -1}{ z_2}}{\beta \rho_{-}T+\frac{e^{z_1 T}-1}{z_1} -\frac{\gamma_1}{\gamma_2}\frac{e^{z_2 T}-1}{ z_2}}.\\ 
	\end{equation*}
        Moreover, this is the $H^{1}[0,T]^{\times N}$ limit of the equilibrium strategy profile in Theorem \ref{thm:equil.term.pen} as $\varphi\uparrow\infty$. 
\end{theorem}

An inspection of the proof in Appendix \ref{app:pf.equil.liq.constr} actually gives the stronger result that the equilibria of Theorem \ref{thm:equil.term.pen} and their derivatives \textit{of all orders} converge uniformly on $[0,T]$ to their counterparts in Theorem \ref{thm:equil.liq.constr} as $\varphi\uparrow\infty$. As in Corollary \ref{cor:eq.cost.term.pen}, we obtain the equilibrium cost in terms of two additional constants $\Psi$ and $\Xi$ whose form is reported in Table \ref{tab:constants}.

\begin{figure}[!h]
  \centering
  \includegraphics[scale=0.21]{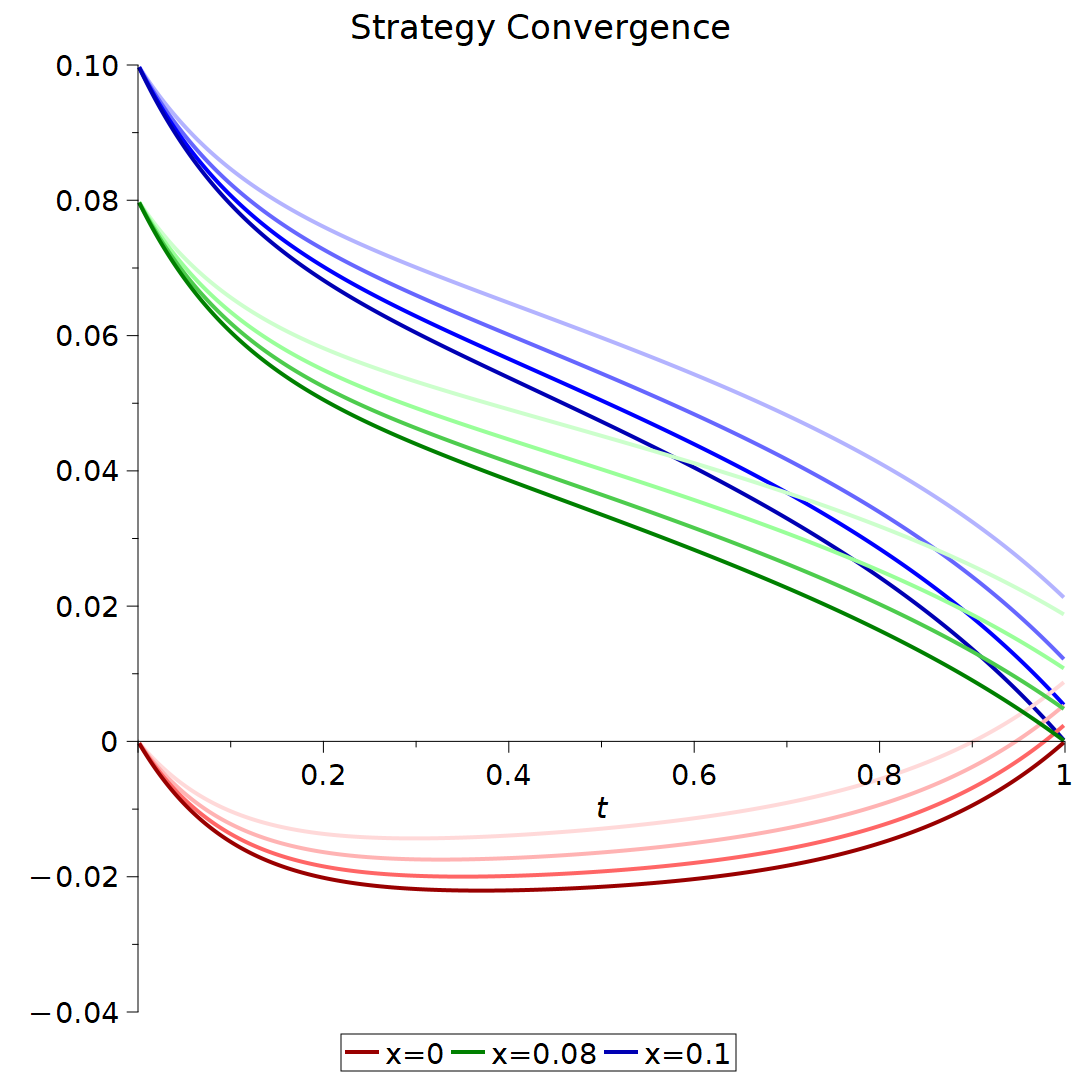}
  \hspace{0.2cm}
  \includegraphics[scale=0.21]{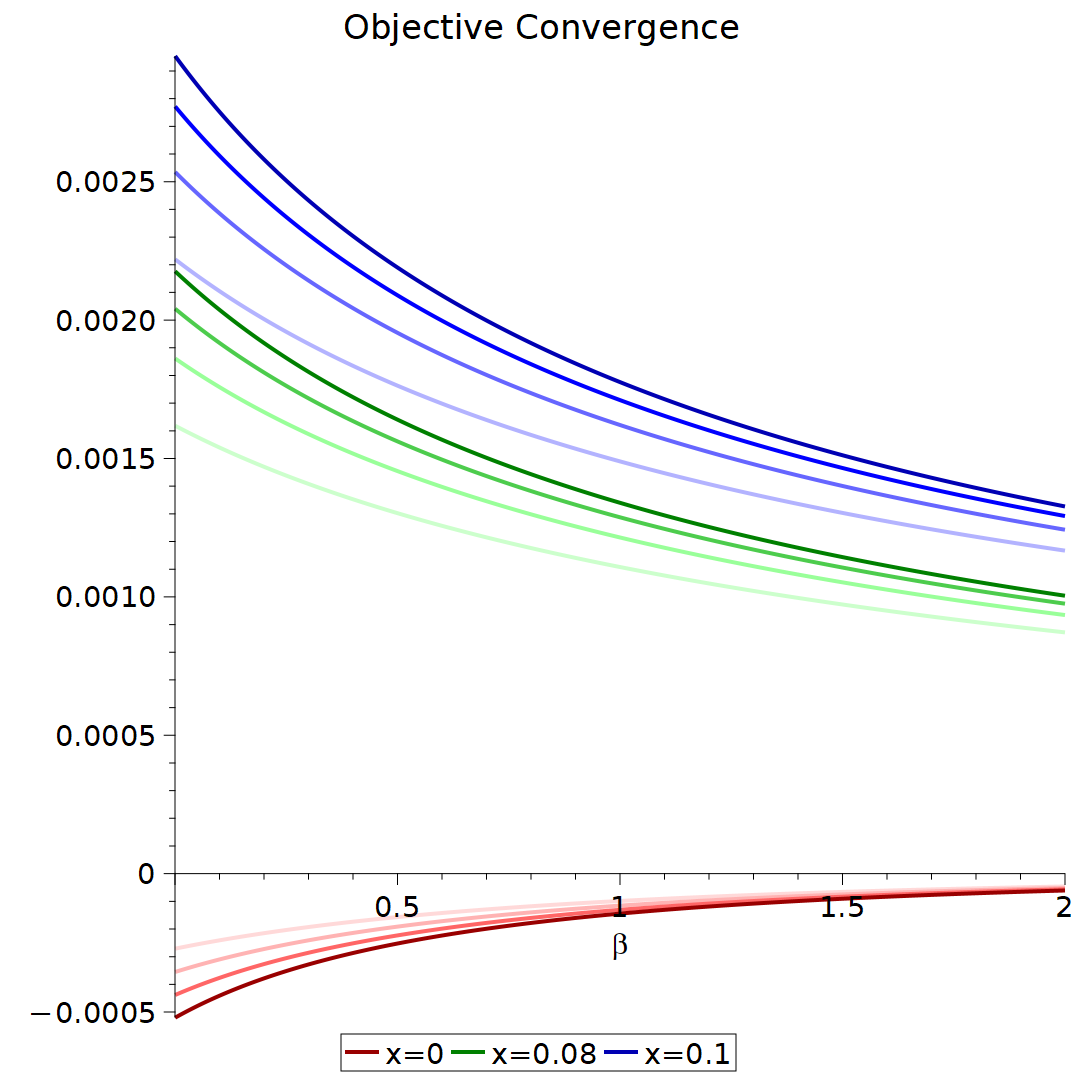}
  \caption{Convergence of the equilibrium in Theorem \ref{thm:equil.term.pen} to that of Theorem \ref{thm:equil.liq.constr} as $\varphi\uparrow \infty$ when $\lambda=0.2$, $\varepsilon=0.05$, $T=1$, and $N=3$. The three colors represent the three agents and range from light to dark as $\varphi$ increases ($\varphi=1,2,5,\infty$). The darkest colors are used for the limiting values corresponding to the liquidation constraint. (Left Panel) Equilibrium strategies $X^{*,i}_t$ as a function of $t$ for $\beta=1$. (Right Panel) Equilibrium costs from Corollaries~\ref{cor:eq.cost.term.pen} and \ref{cor:eq.cost.liq.constr} as a function of $\beta$.}
  \label{fig:equil.conv.phi}
\end{figure}

\begin{corollary}\label{cor:eq.cost.liq.constr}
The equilibrium cost for the traders can be written as
    \begin{align*}
        J_{A'}(X^{*,i};\boldsymbol{X}^{*,-i})= \int_0^TI_{t-}dX^{*,i}_t+\frac{\varepsilon}{2}\int_0^T(\dot{X}^{*,i}_t)^2dt, \ \ \ i=1,\dots,N,
    \end{align*}
in terms of the equilibrium impact process\footnote{As in Corollary \ref{cor:eq.cost.term.pen}, it can be verified using Table \ref{tab:constants} that $I_0=0$.}
\begin{equation*}I_t=\frac{-N\lambda\left[\rho_{-}+\frac{1}{z_1+\beta}e^{z_1t} -\frac{\gamma_1}{ \gamma_2}\frac{1}{z_2+\beta}e^{z_2t}\right]}{\beta\rho_{-}T+\frac{e^{z_1 T}-1}{z_1}-\frac{\gamma_1}{\gamma_2}\frac{e^{z_2 T}-1}{z_2}}\overline{x}, \ \ \ t\in[0,T].
\end{equation*}
Stated explicitly, we obtain the equilibrium impact cost
\begin{align*}
    \int_0^TI_{t-}dX^{*,i}_t &=\lambda N\left[\frac{\mathfrak{h}_1}{\Psi^2}\overline{x}^2+\frac{\mathfrak{h}_2}{\Xi\Psi}(x^{i}-\overline{x})\overline{x}\right],
\end{align*}
and equilibrium instantaneous cost
\begin{align*}
    \frac{\varepsilon}{2}\int_0^T(\dot{X}^{*,i}_t)^2dt&=\varepsilon\left[\frac{\mathfrak{h}_3}{2 \Psi^2}\overline{x}^2+\frac{\mathfrak{h}_4}{2 \Xi^2}(x^{i}-\overline{x})^2+\frac{\mathfrak{h}_5}{ \Xi \Psi}\overline{x}(x^{i}-\overline{x})\right].
\end{align*}
Moreover, this is the limit of the cost in Corollary \ref{cor:eq.cost.term.pen} as $\varphi\uparrow\infty$.
\end{corollary}

Figure \ref{fig:equil.conv.phi} illustrates the convergence of the equilibrium strategies and costs.

\FloatBarrier

\section{Equilibria with Block Costs}\label{sec:N.block.cost}

This section describes the equilibria for the cost $C_B(\cdot)$ in \eqref{eqn:CB}, which has no instantaneous cost on the trading rate but has additional costs on block orders. We shall see that an equilibrium exists only for very particular parameter values. The deeper meaning of this seemingly peculiar cost structure will be addressed in Section~\ref{sec:id.limit}.

\begin{remark}\label{rem:intuition}
Let us begin with some intuition. In the unregularized Obizhaeva--Wang model it is natural to expect that any potential equilibrium strategies would involve block trades, since the classic single-agent solution exhibits block trades at the initial and terminal times. In the $N$-player game, Proposition~\ref{prop:no.interior.jumps} below shows that no interior block trades can occur in a Nash equilibrium: a jump at a time $t\in(0,T)$ can be exploited by other traders with a round-trip perturbation concentrated around~$t$.

As in the single-agent solution, the endpoints $t=0$ and $t=T$ are different. At $t=0$ there is no past, and an initial block trade can jump-start resilience; at $t=T$ there is no future, so a terminal block trade avoids imposing transient impact on subsequent trading. Even for block trades at the initial and terminal times, Nash equilibrium imposes stringent constraints. In the naive Obizhaeva--Wang model, such block trades can still be exploited by a round-trip perturbation (now also involving a block trade), leading to non-existence of equilibrium and suggesting an additional penalty on block trades as a remedy. A more systematic motivation for the penalty will be given in Section~\ref{sec:id.limit}, where we show that the additional block costs are exactly the limits of instantaneous trading costs in the vanishing-regularization limit. 

For simplicity, let us focus on $N=2$ players with symmetric initial conditions. Intuitively, the additional penalty needs to satisfy two competing demands: It should be small enough so that replacing the block trade by fast continuous trading does not improve costs. But it should be large enough so that exploiting by a round-trip perturbation with an even larger block trade is no longer optimal for the opponent, as that would lead to an escalating feedback loop and hence non-existence.

It turns out that adding quadratic costs with the specific coefficients $\vartheta_0=\lambda(N-1)/2$ and $\vartheta_T=\lambda/2$ satisfies these demands and yields existence of equilibrium, whereas other values will not work in general. We observe that $\vartheta_0$ depends on $N$ whereas $\vartheta_T$ does not. In our equilibrium formulas, the initial jump arises purely through the mean component~$\overline{x}$. Stabilizing this common interaction at $t=0$ depends on how many opponents trade against the same instantaneous price formation, hence the $N$-dependence of $\vartheta_0$. By contrast, the terminal block trade clears the idiosyncratic deviation component $x^i-\overline{x}$. Since these deviations sum to zero, the relevant terminal condition is local and yields the single-agent coefficient~$\lambda/2$, independent of $N$.
\end{remark}

As in Section \ref{sec:N.inst.cost} we first restrict our search for an equilibrium to the class of deterministic strategies. However, since strategies with jumps and singular continuous components are no longer ruled out by the cost, we have a much larger space to search over. The following preliminary step narrows down the type of jumps that can arise in equilibrium. First, block trades do not occur on $(0,T)$. Second, if the additional costs $\vartheta_0, \vartheta_T$ are positive, the block trades at $t=0$ and $t=T$ are determined by the block trades of the other players.

\begin{proposition}
    \label{prop:no.interior.jumps}
    If $\boldsymbol{X}^{*}$ is a Nash equilibrium then $\boldsymbol{X}^{*}$ has no interior jumps,
    \[\Delta X_{t}^{*,i} = 0,  \ \ \ i=1,\dots,N, \ \ \ \forall t\in(0,T), \]
    and its initial and terminal jumps satisfy
    \[\vartheta_0\Delta X_{0}^{*,i}=\frac{\lambda}{2}\sum_{j\not=i}\Delta X_{0}^{*,j}, \ \ \ \vartheta_T\Delta X_{T}^{*,i}=-\frac{\lambda}{2}\sum_{j\not=i}\Delta X_{T}^{*,j},  \ \ \ i=1,\dots, N.\]
    In particular, $\vartheta_0 = 0$ implies $\Delta X_0^{*,i}=0$ for all $i$, and $\vartheta_T = 0$ implies $\Delta X_T^{*,i}=0$ for all $i$.
\end{proposition}

The proof in Appendix~\ref{app:N.block.cost} uses necessary conditions for optimality that arise from perturbing a reference strategy by a round trip jump trade. A somewhat technical manipulation of these conditions leads to a relationship between the jumps that must hold when a representative trader acts optimally.

 In view of \cref{prop:no.interior.jumps}, we would like to limit our search to strategies that only jump at the beginning and end of the trading period. In addition, we want to restrict ourselves to absolutely continuous trading on $(0,T)$. The next lemma justifies this reduction.

\begin{lemma}\label{lem:nash.eq.a.c.}
    A Nash equilibrium for $J_B$ in the class of deterministic admissible strategies that are absolutely continuous on $(0,T)$ is a Nash equilibrium (in the class of admissible strategies).
\end{lemma}

The proof in Appendix \ref{app:N.block.cost} proceeds through an approximation argument using Bernstein polynomials. While the proof is fairly technical, the basic idea is that if an agent is incentivized to deviate using a general strategy $Z$, deviating to a smoothed version of $Z$ still reduces the execution cost.

We now lean on Proposition \ref{prop:no.interior.jumps} and Lemma \ref{lem:nash.eq.a.c.} to set up our deterministic problem. To control the trading speed on $(0,T)$, we continue to work with functions $v^{i}\in L^2[0,T]$. We also parametrize the initial and terminal jumps using constants, $\Delta X_0^i  =: a^i\in\mathbb{R}$ and $\Delta X_T^i=:b^i\in \mathbb{R}$. If we let $\theta_a := \vartheta_0$ and $\theta_b :=\vartheta_T$ it is not hard to verify via \eqref{eqn:obj.impact.rep} that the objective for a representative trader $i$ can be recast as minimizing
\begin{align}
  &\cJ_B(a^{i},v^{i};\boldsymbol{v}^{-i}) =  \frac{1}{2} I_0 a^{i} + \int_0^T I_t v^{i}_t \,dt + \frac{1}{2}(I_{T-} + I_T) b^{i}+ \frac{\theta_a}{2}(a^{i})^2+\frac{\theta_b}{2}(b^{i})^2,\nonumber
\end{align} 
where the impact process $I$ is given by
\begin{equation*}
    I_t=\begin{cases}
    \lambda e^{-\beta t}\sum_{j=1}^Na^{j}+\int_0^t \lambda e^{-\beta(t-s)}\sum_{j=1}^Nv^{j}_sds, & t\in[0,T),\\
    I_{T-}+\lambda \sum_{j=1}^Nb^{j}, & t=T
    \end{cases}
\end{equation*}
for $I_{0-}=0$, and the inventories satisfy
\begin{equation*} X^{i}_t=\begin{cases}
    x^{i}+a^{i}+\int_0^tv^{i}_sds, & t\in[0,T),\\
    0, & t=T,
\end{cases}
\end{equation*}
for $X_{0-}^i=x^i$. In particular, this implies that $b^i = -X_{T-}^i$.
In view of the liquidation constraint, $b^i$ is entirely determined by the initial jump $a^i$ and the trading speed $v^i$, hence our minimization is over $\mathbb{R}\times L^2[0,T]$. This retains a Hilbert space structure so, as in Section~\ref{sec:N.inst.cost}, we can apply variational arguments to solve the game. The following lemma shows that the deterministic equilibrium (if it exists) is similarly characterized by a $2N+1$ dimensional system of linear homogeneous ordinary differential equations. However, this time the ODE is written in terms of $\boldsymbol{X}$ (rather than $\boldsymbol{v}$) and the imposition of block costs leads to additional free boundary conditions that must also be pinned down as part of the solution.

\begin{lemma}\label{lem:block.cost.ODE}
The strategy profile $\boldsymbol{X}$ defines a Nash equilibrium for $\mathcal{J}_B$ if and only if it forms, along with $I$ and auxiliary processes $Y^{1},\dots, Y^{N}$, a solution\footnote{More precisely (since $\boldsymbol{X}$ and $I$ may have jumps at $0$ and $T$), $\boldsymbol{X}$ and $I$ satisfy the ODE on $[0,T)$.} to the ODE system 
    \begin{align*}
    \dot{I}_t&=\frac{\beta}{N-1}\left[I_t-\sum_{j=1}^NY^{j}_t\right],\\
    \dot{Y}_t^{i}&=-\frac{\beta}{N-1}\left[I_t-\sum_{j=1}^NY^{j}_t\right], \ \ \ i=1,\dots,N,\\
    \dot{X}_t^{i}&=\frac{\beta}{\lambda(N-1)}\left[I_t+(N-1)Y^{i}_t-\sum_{j=1}^NY^{j}_t\right], \ \ \ i=1,\dots,N,
\end{align*}
subject to the initial and terminal conditions
\begin{align*}
    I_0&=\lambda\sum_{i=1}^Na^{i}, \ \ \ X_0^{i}=x^{i}+a^{i}, \ \ \ Y_T^{i}=\lambda b^{i}, \ \ \ i=1,\dots,N,
\end{align*}
where 
\begin{equation*}\theta_a a^{i} =\frac{\lambda}{2}\sum_{j\not=i}a^{j}, \ \ \ \theta_b b^{i}=-\frac{\lambda}{2}\sum_{j\not=i}b^{j}, \ \ \ \text{and} \ \ \  b^{i}=-X^{i}_{T-}, \  \ \ \ i=1,\dots,N.
\end{equation*}

\end{lemma}

It turns out that the existence of an equilibrium depends crucially on the choice of initial and terminal block costs $\theta_a=\vartheta_0$ and $\theta_b=\vartheta_T$. As the next theorem shows,  there is a single choice yielding existence for general initial inventories. Moreover, that choice consists of different values for the initial and terminal costs, except in the case $N=2$ of two traders. 

\begin{theorem}\mbox{}\label{thm:equil.block.cost}
\begin{enumerate}
    \item[(1)] If $\vartheta_0=\frac{\lambda (N-1)}{2}$ and 
    $\vartheta_T = \frac{\lambda}{2}$, then a Nash equilibrium for $J_B$ exists.

    \item[(2)] If $\vartheta_0\not=\frac{\lambda (N-1)}{2}$ and 
    $\vartheta_T = \frac{\lambda}{2}$, then a Nash equilibrium for $J_B$ 
    exists if and only if $\overline{x}=0$.

    \item[(3)] If $\vartheta_0=\frac{\lambda (N-1)}{2}$ and 
    $\vartheta_T \not= \frac{\lambda}{2}$, then a Nash equilibrium for $J_B$ 
    exists if and only if $x^i=x^j$ for all $i,j\in\{1,\dots,N\}$.

    \item[(4)] If $\vartheta_0\not=\frac{\lambda (N-1)}{2}$ and 
    $\vartheta_T \not= \frac{\lambda}{2}$, then a Nash equilibrium for $J_B$ 
    exists if and only if $x^i=0$ for all $i=1,\dots,N$.
\end{enumerate}
    When a Nash equilibrium exists, it is uniquely defined through the strategy profile $\boldsymbol{X}^*=(X^{*,1},\dots,X^{*,N})$ given by
    \begin{equation*}
		X^{*,i}_t= \mathbbm{f}_t (x^{i}-\overline{x})+\mathbbm{g}_t \overline{x}, \ \ \ t\in[0,T], \ \ \  i =1,\dots,N,
	\end{equation*}
	where
	\begin{equation*}
		\mathbbm{f}_t =1-\frac{\beta t}{\beta T+1} , \ \ \ t\in[0,T), \ \ \ \mathbbm{f}_{0-} = 1 \ \ \ \text{and} \ \ \ \mathbbm{f}_T = 0,
        \end{equation*}
        \begin{equation*}
		\mathbbm{g}_t =1-\frac{N(\beta t + 1) (N+1) e^{\beta \frac{N+1}{N-1} T}+2Ne^{\beta \frac{N+1}{N-1} t}-(N-1)}{N((\beta T+1)(N+1)+2)e^{\beta \frac{N+1}{N-1} T}-(N-1)}, \ \ \ t\in[0,T] \ \ \ \text{and} \ \ \ \mathbbm{g}_{0-} = 1. 
        \end{equation*}
\end{theorem}

We observe that the equilibrium, while sharing the affine form, has a much simpler expression compared to the one for instantaneous cost in the preceding section.

\begin{remark}
    When $\vartheta_0\not=\frac{\lambda(N-1)}{2}$ and 
    $\vartheta_T\not=\frac{\lambda}{2}$, (4) shows that an 
    equilibrium exists if and only if $x^i = 0$ for all $i=1,\dots,N$. 
    In that case, substituting in $x^i=0$ for all $i$ shows $X^{*,i}\equiv0$. In words, the only possible equilibrium is that all agents have zero inventory and do not trade. More generally, when $\vartheta_0\not=\frac{\lambda(N-1)}{2}$ \emph{or} $\vartheta_T\not=\frac{\lambda}{2}$, the equilibrium exists only for initial inventories such that no block trade occurs at $t=0$ or $t=T$, respectively, and in that sense the value of $\vartheta_t$ does not matter.
\end{remark}

\begin{remark}
  In certain games without a (pure) Nash equilibrium, existence is re-established by allowing for randomization, meaning that agents' strategies can depend on privately observed randomization devices. Seminar participants have rightfully asked whether that is the case in the context of \cref{thm:equil.block.cost}. The answer is negative. In a nutshell, randomized strategies can be de-randomized in a way that strictly decreases the objective function, ruling out the existence of randomized equilibria. The mathematical details are somewhat lengthy and are reported separately in \cite{CampbellNutz.25b}.
\end{remark}

The proof of Theorem~\ref{thm:equil.block.cost} in Appendix \ref{app:N.block.cost} has two parts. The first is based  on the ODE system from Lemma~\ref{lem:block.cost.ODE}, providing  the explicit solution for the ``good'' parameter values and proving that there is no solution for the ``bad'' parameter values. On the strength of Lemma~\ref{lem:nash.eq.a.c.}, the former establishes existence and uniqueness of the equilibrium also for general admissible strategies, for those parameter values. Whereas for the bad parameter values, verifying that the non-existence of a deterministic equilibrium extends to the full class of admissible strategies is a major technical hurdle. Achieving this occupies a sizable part of the proof, which combines original arguments with ideas of \cite[Theorem 4.5(b)]{SchiedStrehleZhang.17}.

The final result of this section provides the equilibrium cost in closed form.

\begin{corollary}\label{cor:eq.cost.block.cost}
    When a Nash equilibrium for $J_B$ exists, the cost for the traders is
    \begin{align*}
        J_{B}(X^{*,i};\boldsymbol{X}^{*,-i})&= \int_0^TI_{t-}dX^{*,i}_t+\frac{1}{2}\left(\Delta I_0\Delta X^{*,i}_0+\Delta I_T \Delta X^{*,i}_T\right)\\
        & \quad \quad \quad \quad  \quad \quad +\frac{1}{2}\left(\vartheta_0(\Delta X^{*,i}_0)^2+\vartheta_T(\Delta X^{*,i}_T)^2\right), \ \ \ i=1,\dots,N,
    \end{align*}
    in terms of the equilibrium impact process
    \[I_t=-\frac{ \lambda  N (N +1)\left(e^{\beta \frac{N +1}{N -1}t}+N e^{\beta \frac{N +1}{N -1}T}\right) }{N \left(( \beta T +1) (N+1) +2\right) e^{\beta \frac{\left(N +1\right) T}{N -1}}-(N -1)}\overline{x}, \ \ \ t\in[0,T], \ \ \ I_{0-}=0.
    \]  
    Stated explicitly, we obtain the equilibrium impact cost
    \begin{align*}
        &\int_0^TI_{t-}dX^{*,i}_t +\frac{1}{2}\left(\Delta I_0\Delta X^{*,i}_0+\Delta I_T \Delta X^{*,i}_T\right) = \frac{\lambda N}{\beta T+1} \overline{x}(x^i-\overline{x})\\
        &\quad \quad \quad + \frac{ \lambda N^3 (N+1)\left( \left(\left(\beta T +\frac{1}{2}\right) (N+1) +3\right) e^{\frac{2 \left(N +1\right) \beta  T}{N -1}}-\frac{2 \left(N -1\right)}{N^2} \left(N e^{\frac{\left(N +1\right) \beta  T}{N -1}}+\frac{1}{4}\right) \right)}{\left(N \left(\left(\beta T +1\right)(N+1) + 2\right) e^{\frac{\left(N +1\right) \beta  T}{N -1}}-(N -1)\right)^{2}}
    \overline{x}^2,
    \end{align*}
    and equilibrium block trade cost
    \begin{multline*}
        \frac{1}{2}\left(\vartheta_0(\Delta X^{*,i}_0)^2+\vartheta_T(\Delta X^{*,i}_T)^2\right)\\= \frac{\vartheta_0 (N+1)^2 (1+ N e^{\beta \frac{N+1}{N-1} T})^2 \overline{x}^2}{2 \left(N((\beta T+1) (N+1) + 2)e^{\beta \frac{N+1}{N-1} T} - (N-1)\right)^2} +\frac{\vartheta_T (x^i-\overline{x})^2}{2 (\beta T+1)^2}.
    \end{multline*}
\end{corollary}

We can make several observations about the structure of the solution. First, we observe that the equilibrium costs are linear in the liquidity parameter $\lambda$. This does not factor into the strategies of Theorem \ref{thm:equil.block.cost}, while the resilience $\beta$ features prominently. An illustration of the solution for the same three player game as Figure \ref{fig:equil.conv.phi} appears in Figure \ref{fig:equil.conv.eps}. As $\beta\downarrow 0$ this approximates permanent impact and we see that $\mathbbm{g}_t\to 0$, $\mathbbm{f}_t\to 1$ on $[0,T)$. Consequently, there is no trading on $(0,T)$ and every player liquidates their inventory in two block trades at $t=0$ and $t=T$. The limiting cost is
\[\lim_{\beta\downarrow0} J_{B}(X^{*,i};\boldsymbol{X}^{*,-i}) = \lambda N \overline{x}(x^i-\overline{x}) + \frac{\lambda N}{2}\overline{x}^2 + \frac{\vartheta_0}{2}\overline{x}^2+\frac{\vartheta_T}{2}(x^i-\overline{x})^2.\]
On the other hand, if we send $\beta\uparrow \infty$ this corresponds to temporary price impact and the limiting cost is $0$. We find $\mathbbm{g}_t,\mathbbm{f}_t\to 1-\frac{t}{T}$ which recovers the familiar time-weighted average price (TWAP) strategy that is optimal in the Almgren--Chriss model.

\section{Identifying the Limit of Small Instantaneous Cost}\label{sec:id.limit}

This section connects the equilibrium with instantaneous cost (Theorem~\ref{thm:equil.liq.constr}) and the equilibrium with block cost (Theorem \ref{thm:equil.block.cost}). Namely, we show that the latter is the limit of the former for vanishing instantaneous cost $\eps\to0$. The equilibrium with instantaneous cost is canonical in that it does not require a particular choice of parameters. The limit of vanishing instantaneous cost then gives rise to the seemingly unprincipled pair of ``good'' block cost parameters in a natural way. Indeed, our result implies that any different choice of parameters would lead to a discontinuity in the equilibrium cost. 

\begin{theorem}\label{thm:small.eps.lim}
    As $\varepsilon\downarrow 0$ the equilibrium $\boldsymbol{X}^*=\boldsymbol{X}^{*,\eps}$ in Theorem \ref{thm:equil.liq.constr} converges uniformly on compact subsets of $(0,T)$ to the equilibrium $\boldsymbol{X}^*=\boldsymbol{X}^{*,0}$ in Theorem \ref{thm:equil.block.cost}. Furthermore, the equilibrium cost in Corollary \ref{cor:eq.cost.liq.constr} converges to that of Corollary \ref{cor:eq.cost.block.cost} when $\vartheta_0 =\frac{\lambda (N-1)}{2}$ and $\vartheta_T = \frac{\lambda}{2}$. In particular, for any $\delta\in(0,T)$,		\begin{align*}
     \varepsilon \int_0^{\delta}(\dot{X}^{*,\eps,i}_t)^2 dt \to   \vartheta_0(\Delta X^{*,0,i}_0)^2 \quad\mbox{and}\quad \varepsilon \int_{\delta}^{T}(\dot{X}^{*,\eps,i}_t)^2 dt \to   \vartheta_T(\Delta X^{*,0,i}_T)^2.
		\end{align*} 
\end{theorem}

\begin{remark}\label{rk:small.eps.lim}
  We can observe that in addition to the locally uniform convergence on $(0,T)$, the strategies converge at $T$ but not at $0$. This is merely due to the convention for the jump; the right continuous modification of the limiting strategy coincides with the strategy in Theorem \ref{thm:equil.block.cost} everywhere on $[0,T]$.
\end{remark} 

\begin{figure}[!htbp]
  \centering
  \includegraphics[scale=0.21]{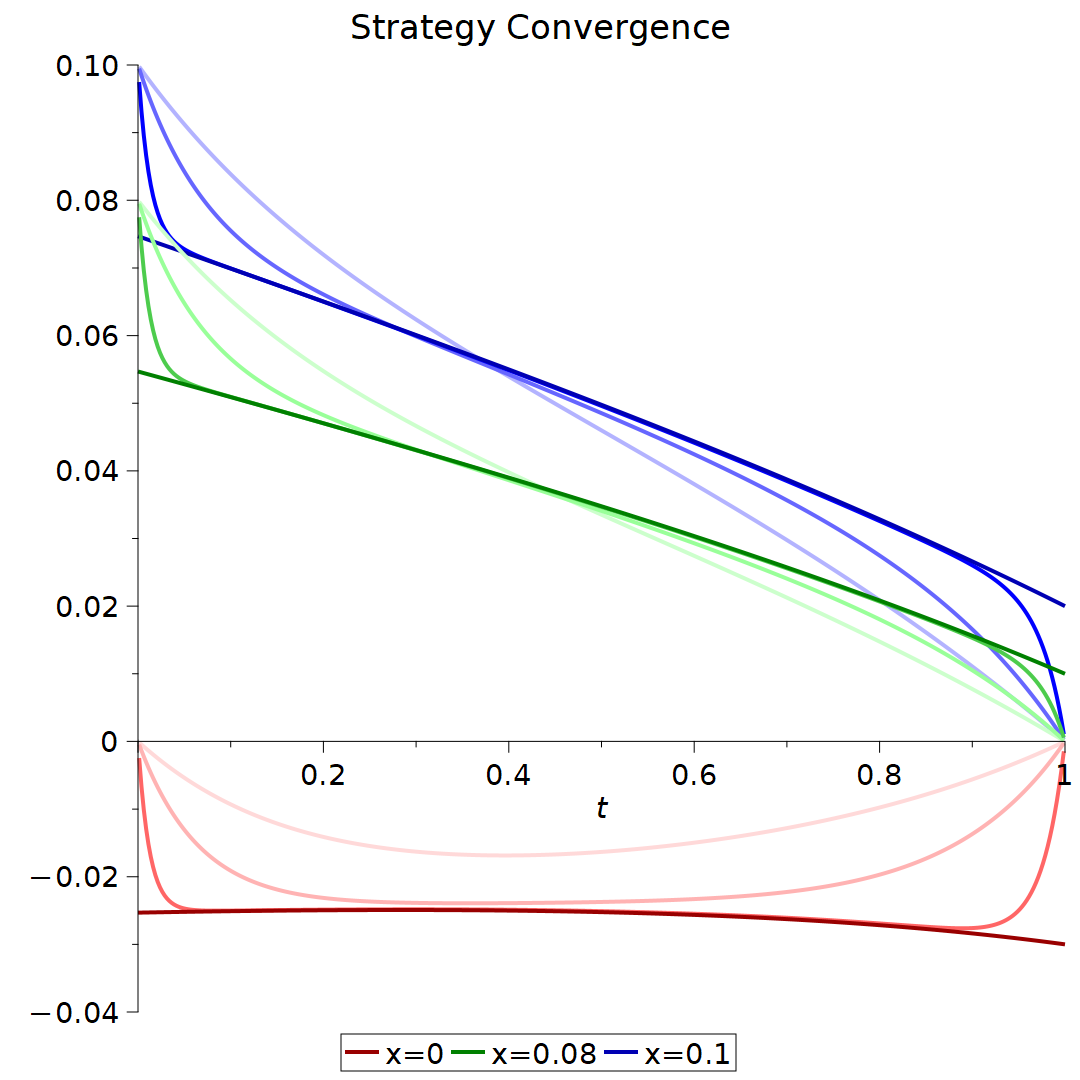}
  \hspace{0.2cm}
  \includegraphics[scale=0.21]{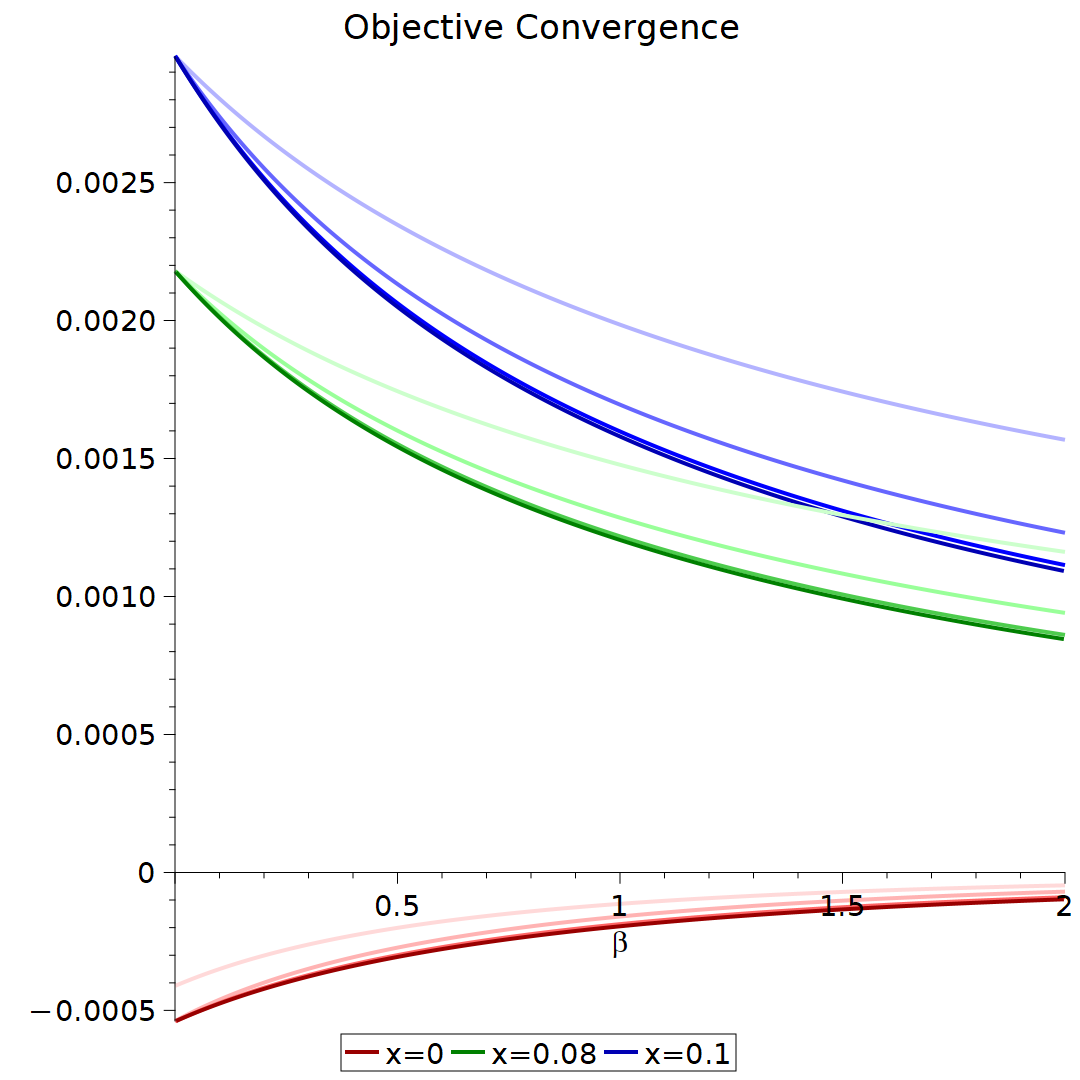}
  \caption{Convergence of the equilibrium in Theorem \ref{thm:equil.liq.constr} to that of Theorem \ref{thm:equil.block.cost} as $\varepsilon \downarrow 0$ when $\lambda=0.2$, $T=1$, and $N=3$. The three colors represent the three agents and range from light to dark as $\varepsilon$ decreases ($\eps=0.1,0.03,0.005,0$). The darkest colors correspond to the block cost. 
 (Left Panel) Equilibrium strategies $X^{*,\varepsilon, i}_t$ for $\beta=1$. (Right Panel) Equilibrium costs from Corollaries \ref{cor:eq.cost.liq.constr} and \ref{cor:eq.cost.block.cost} as a function of $\beta$.}
  \label{fig:equil.conv.eps}
\end{figure}

Figure \ref{fig:equil.conv.eps} illustrates the convergence of the equilibrium strategies and their costs. \Cref{thm:small.eps.lim} not only justifies the particular block costs found in \cref{sec:N.block.cost}, but also shows that optimal trading is more aggressive than in the single player case. The next remark quantifies that observation.

\begin{remark}
The limiting behavior of the instantaneous cost observed in \cref{thm:small.eps.lim} is qualitatively different from the single-player case. If we pose the same objectives when $N=1$, the optimal instantaneous cost becomes negligible as $\varepsilon\downarrow0$.
More precisely, one can show (e.g., based on the formula in \cite[Theorem~1]{ChenHorstTran.19}) that
\begin{align*}
    \frac{\varepsilon}{2}\int_0^T(\dot{X}^{*,\varepsilon,i}_t)^2\,dt \sim C \varepsilon^{1/2}
\end{align*}
for a constant $C=C(\lambda,\beta, T, x^i)\geq 0$ that is strictly positive when $x^i\neq0$.
In particular, this implies that $\|\dot{X}^{*,\varepsilon,i}\|_{L^2[0,T]}$ is of order $\varepsilon^{-1/4}$ for nonzero initial inventories.
Whereas in the game with $N>1$, the instantaneous cost converges to the (generally non-zero) block cost by \cref{thm:small.eps.lim}, implying that $\|\dot{X}^{*,\varepsilon,i}\|_{L^2[0,T]}$ is of order $\varepsilon^{-1/2}$ when
at least one of $x^i$ or $\overline{x}$ is nonzero.
\end{remark}

\section{Costs of Anarchy and Predation}\label{se:cost.anarchy.predation}

In this section we compare the equilibrium cost of Section~\ref{sec:N.block.cost} with the classic Obizhaeva--Wang solution for a single trader in the absence of competition.

\begin{remark}
Recall that Corollary~\ref{cor:eq.cost.block.cost} decomposes each trader’s equilibrium cost into an impact component and the additional block costs at the initial and terminal times. Below, we define the costs of anarchy and predation using only the impact costs. The motivation is twofold. First, impact costs quantify the price distortion imposed on the market and therefore isolate the strategic externalities caused by competition and predation. Second, aggregated impact costs depend only on the net inventory $x=N\overline{x}$, allowing for clean benchmarks against the single-agent Obizhaeva--Wang problem. By contrast, the additional block costs depend on the cross-sectional distribution of inventories, so any welfare comparison including those costs would have to specify an inventory configuration.
\end{remark}

\subsection{Cost of Anarchy}\label{se:cost.anarchy}

For $N\geq2$ traders we can define a notion of ``population impact cost,'' $\mathsf{PIC}_N(x)$, in the game by aggregating the impact cost from Corollary \ref{cor:eq.cost.block.cost} across traders. If we let $x=N\overline{x}$ be the net inventory in the market and use the formula in Corollary~\ref{cor:eq.cost.block.cost}, 
we get that the population impact cost is
\begin{align*}\mathsf{PIC}_N(x):=\frac{ \lambda N^2 (N+1)\left( \left(\left(\beta T +\frac{1}{2}\right) (N+1) +3\right) e^{\frac{2 \left(N +1\right) \beta  T}{N -1}}-\frac{2 \left(N -1\right)}{N^2} \left(N e^{\frac{\left(N +1\right) \beta  T}{N -1}}+\frac{1}{4}\right) \right)}{\left(N \left(\left(\beta T +1\right)(N+1) + 2\right) e^{\frac{\left(N +1\right) \beta  T}{N -1}}-(N -1)\right)^{2}}x^2.
\end{align*}
This can be compared to the impact cost of liquidating the net inventory optimally using the single-player solution that a central planner would employ, 
$\mathsf{PIC}_1(x):=\lambda x^2(\beta T+2)^{-1}$.

We define the (relative, excess) \textit{system cost of anarchy}, $\mathsf{CoA}_N$, as the percent increase in the population impact cost of liquidating the net inventory,
\begin{equation}\label{eqn:CoA}\mathsf{CoA}_N:=\left[\frac{\mathsf{PIC}_N(x)}{\mathsf{PIC}_1(x)}-1\right]\cdot 100\%, \ \ \ x\not=0.
\end{equation}
Note that, from the form of the population impact, $\mathsf{CoA}_N$ does not depend on the net inventory, $x$, or liquidity parameter, $\lambda$. Keeping the trading horizon $T$ fixed, it depends only on the price impact decay per unit time, $\beta$, and the dimension $N$. We visualize this dependence in Figure \ref{fig:incr.sys.imp.cost}.  We note that when the net inventory in the market is $0$, the net impact cost in the game coincides with the cost when $N=1$, $\mathsf{PIC}_N(0)=\mathsf{PIC}_1(0)=0$, and there is no cost of anarchy.

\begin{figure}[!htb]
  \centering
  \includegraphics[scale=0.21]{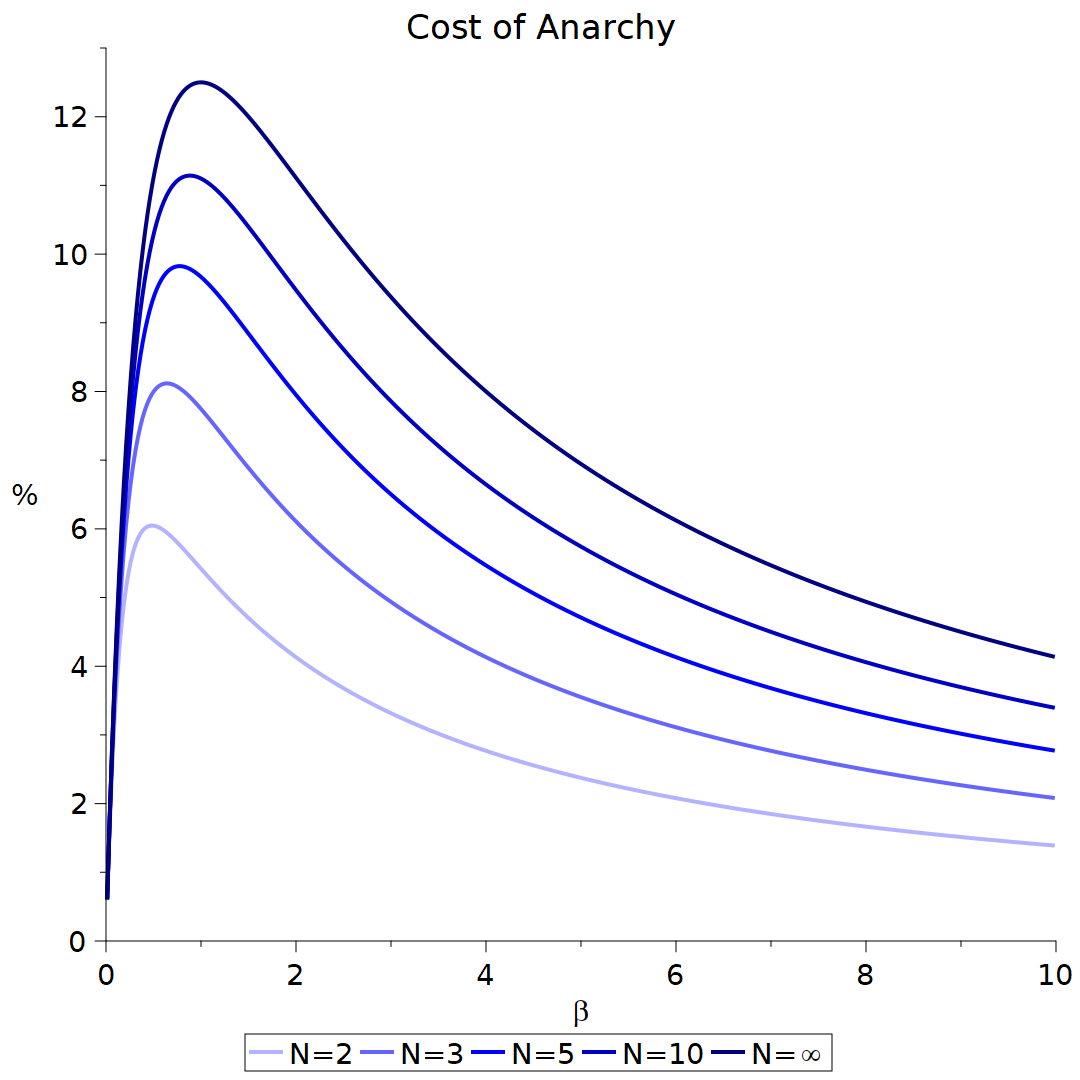}
  \hspace{0.2cm}
  \includegraphics[scale=0.21]{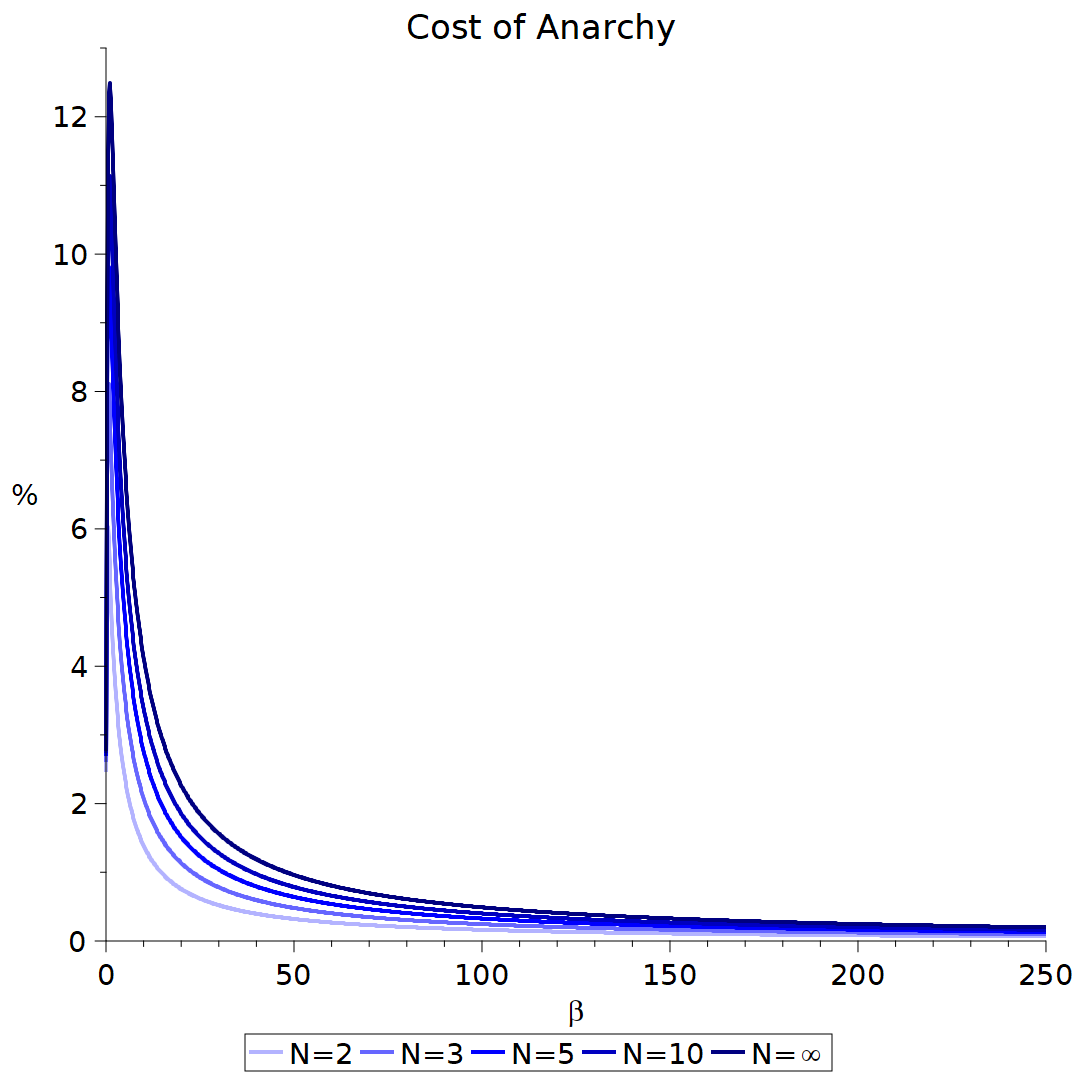}
  \caption{Cost of Anarchy $\mathsf{CoA}_N$ of \eqref{eqn:CoA}, illustrated as a function of $\beta>0$ for various population sizes $N$ when $T=1$. This is the percent increase (over the $N=1$ case) in impact cost incurred by the population to liquidate their inventory. Both panels show the same function; the right panel shows a larger range of $\beta$ to highlight the limit $\beta\to\infty$.}
  \label{fig:incr.sys.imp.cost}
\end{figure}

We can describe the behavior of the cost of anarchy at the extremes of system size and impact decay. The cost of anarchy increases as $N$ increases, but has a finite limit, \[\lim_{N\uparrow\infty} \mathsf{CoA}_N = \frac{\beta T}{2(\beta T+1)^2}\cdot 100\%.\]
Interestingly, as illustrated in Figure \ref{fig:incr.sys.imp.cost}, the cost of anarchy is maximized at an intermediate value of $\beta$. For large $N$ this maximum occurs near $\beta\approx \frac{1}{T}$ and amounts to roughly $12.5\%$. A partial explanation is provided by the limits,
\[\lim_{\beta\downarrow0} \mathsf{CoA}_N = \lim_{\beta\uparrow\infty} \mathsf{CoA}_N  = 0\%.\]
For $\beta\to0$, there is essentially no resilience and the transient impact behaves like permanent impact. Hence, all liquidation strategies have the same cumulative (across agents) impact cost and the cost of anarchy tends to zero. For $\beta\to\infty$, the behavior is akin to temporary price impact with no block costs. Agents unwind their inventory at a constant rate in equilibrium, and that is also the central planner's limiting strategy. The impact costs tend to zero in either case, $\lim_{\beta\uparrow\infty}\mathsf{PIC}_N(x) =\lim_{\beta\uparrow\infty}\mathsf{PIC}_1(x)=0$. They do so at the same rate, resulting in $\lim_{\beta\uparrow\infty} \mathsf{CoA}_N =0$. 

\subsection{Cost of Predation}\label{se:cost.predation}

\begin{figure}[htb]
  \centering
  \includegraphics[scale=0.21]{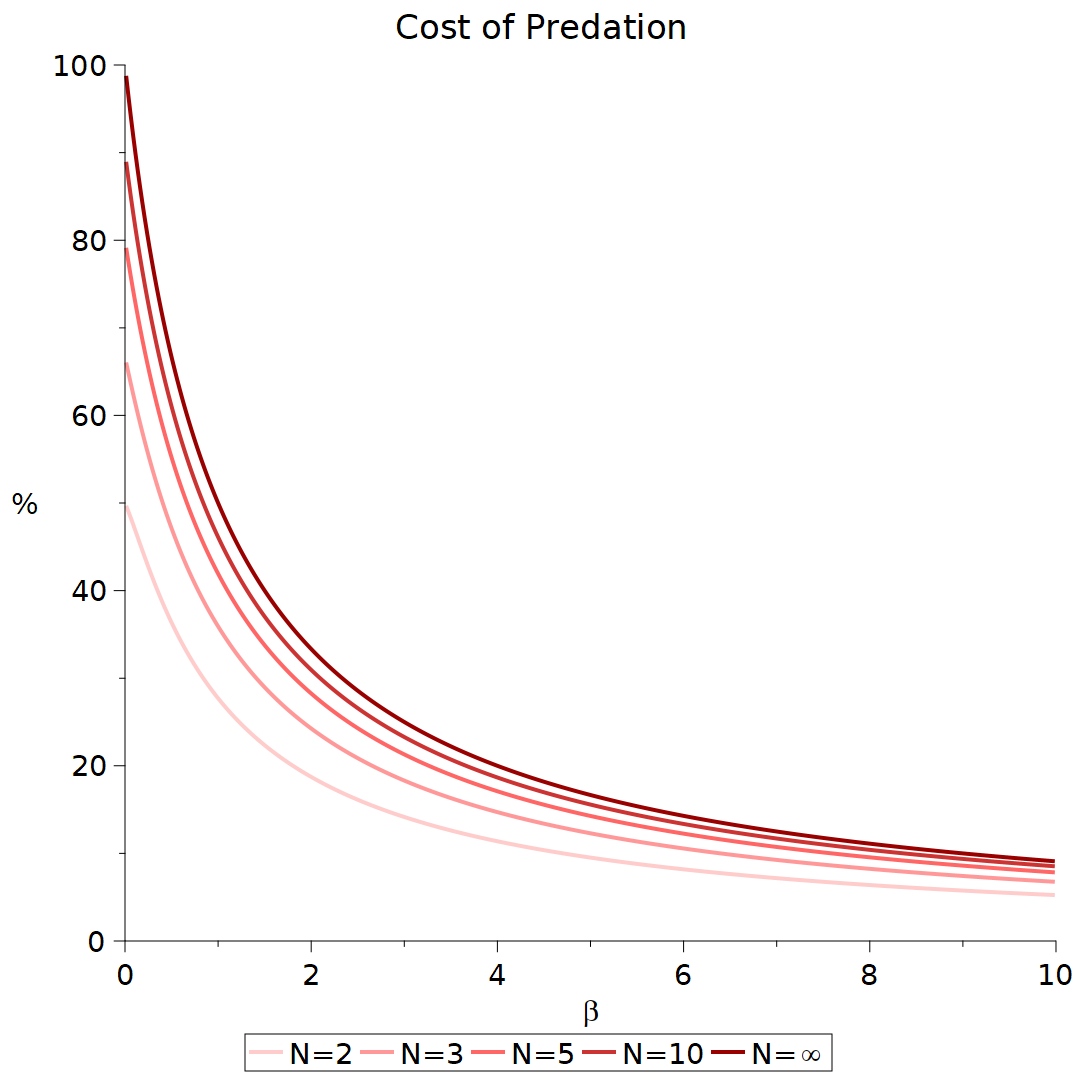}
  \hspace{0.2cm}
  \includegraphics[scale=0.21]{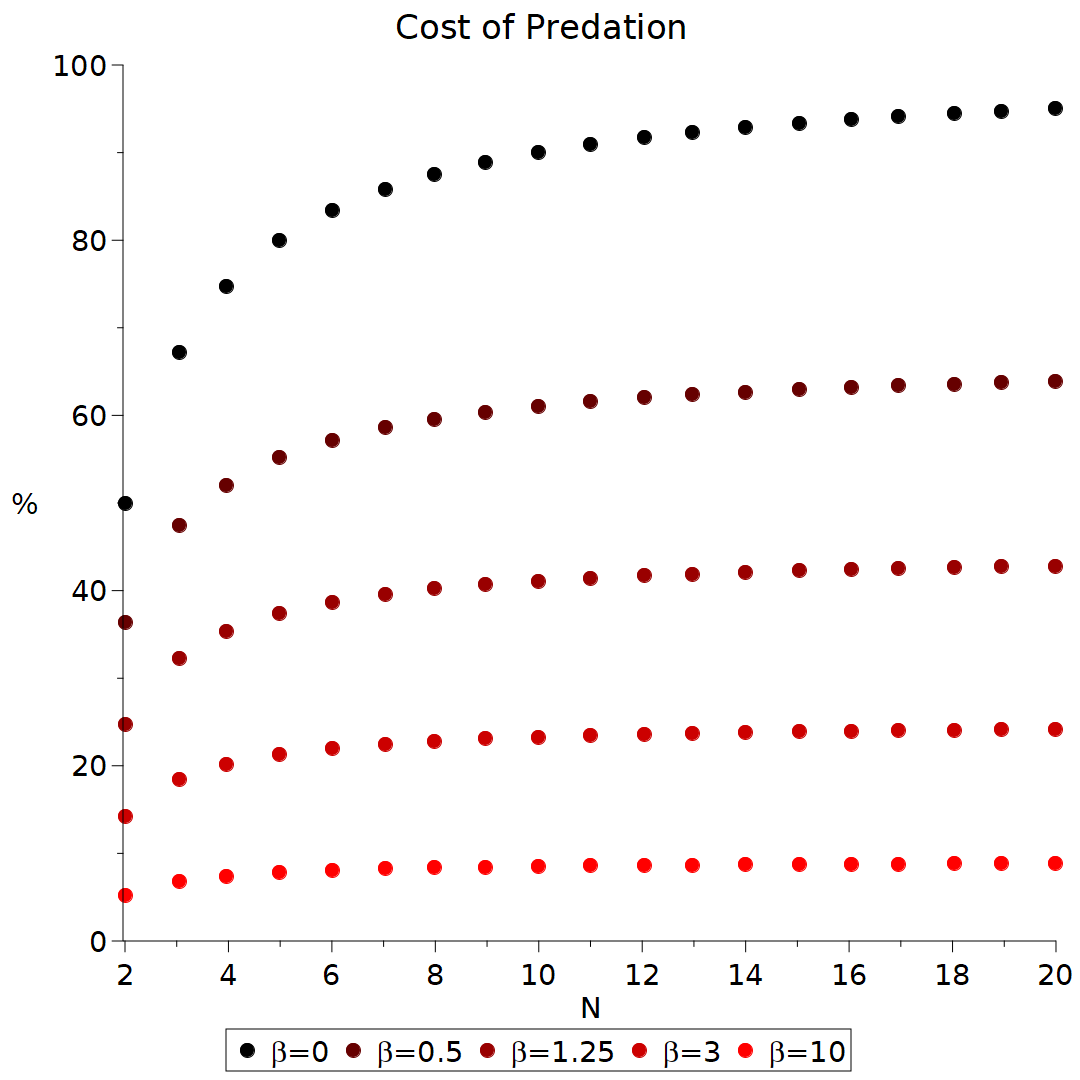}
\caption{Cost of Predation $\mathsf{CoP}_N$ of \eqref{eqn:CoP} when $T=1$. This is the percent increase (over the $N=1$ case) in impact cost for the liquidator when there are $N-1$ predators in the market. (Left Panel) Illustration as a function of $\beta>0$ for various population sizes $N$.  (Right Panel) Illustration as a function of $N$ for several values of $\beta$.}
  \label{fig:incr.sys.imp.cost2}
\end{figure}

Suppose now that there is a single ``liquidator'' in the market that must unwind some non-zero inventory $x$. We investigate what happens when $N-1$ ``predators''---i.e., agents with zero initial inventory---are introduced to the system (cf.\ \cite{PedersenBrunnermeier.05, Schoneborn.08}). In this case, the mean inventory becomes $\overline{x}=x/N$ and we can compare the impact cost in the absence of other traders with the cost faced in an equilibrium with predators. We find the liquidator's impact cost, $\mathsf{LIC}_N(x)$ for $N\geq2$ by substituting $x^i=x$ and $\overline{x}=x/N$ into the impact cost of Corollary \ref{cor:eq.cost.block.cost},
\begin{align*}
        \mathsf{LIC}_N(x)&:=\lambda\Bigg[ \frac{ (N-1)}{N(\beta T+1)}\\
        &\quad + \frac{ N (N+1)\left( \left(\left(\beta T +\frac{1}{2}\right) (N+1) +3\right) e^{\frac{2 \left(N +1\right) \beta  T}{N -1}}-\frac{2 \left(N -1\right)}{N^2} \left(N e^{\frac{\left(N +1\right) \beta  T}{N -1}}+\frac{1}{4}\right) \right)}{\left(N \left(\left(\beta T +1\right)(N+1) + 2\right) e^{\frac{\left(N +1\right) \beta  T}{N -1}}-(N -1)\right)^{2}}\Bigg]
    x^2.
\end{align*}

When $N=1$ the liquidator's impact cost and the population impact cost coincide, $\mathsf{LIC}_1(x) =\lambda x^2 (\beta T+2)^{-1}=\mathsf{PIC}_1(x)$. We define the (relative, excess) \textit{cost of predation}, $\mathsf{CoP}_N$ as the percent increase in the liquidator's impact cost due to the presence of $N-1$ predators,
\begin{equation}\label{eqn:CoP}
\mathsf{CoP}_N :=\left[\frac{\mathsf{LIC}_N(x)}{\mathsf{LIC}_1(x)}-1\right]\cdot 100\%.
\end{equation}
Once again, this cost depends only on $\beta$ and $N$. We illustrate the dependence on these parameters in Figure \ref{fig:incr.sys.imp.cost2} where we see that the cost is increasing in $N$ and decreasing in $\beta$. 

By rearranging, we can decompose the cost of predation into two parts,
\begin{align}\label{eq:eq:Cop.decomp}
\mathsf{CoP}_N = \left[\frac{(N-1)(\beta T+2)}{N(\beta T+1)}-\frac{N-1}{N}\right]\cdot 100\%+\frac{1}{N}\mathsf{CoA}_N. 
\end{align} 
Using this expression, we see that the limiting behavior is
\[\lim_{N\uparrow\infty} \mathsf{CoP}_N = (1+\beta T)^{-1}\cdot 100\%, \ \ \ \text{while} \ \ \ \lim_{\beta\uparrow\infty} \mathsf{CoP}_N = 0\% \ \ \ \text{and} \ \ \ \lim_{\beta\downarrow0} \mathsf{CoP}_N = \frac{N-1}{N} \cdot 100\%.\]
From~\eqref{eq:eq:Cop.decomp} we also see immediately that $\mathsf{CoP}_N>0$ for all $N\geq2$; that is, the presence of predators is always costly. While this seems natural, it should be contrasted with the result in \cite[Corollary~3.3]{SchiedZhang.17}. There, using the Almgren--Chriss model with temporary and permanent (but no transient) impact, the ``predator'' indeed acts as a predator when the permanent impact parameter is large enough, but acts as a liquidity provider when the permanent impact parameter is small relative to the temporary impact parameter.

\begin{remark}
    A similar analysis can be conducted when there are $m\geq 1$ liquidators seeking to liquidate the \emph{same} inventory $x$. The mean is then replaced by $\overline{x}=\frac{m}{N}x$ and we can compare $\mathsf{LIC}_{1}(x)$ to the liquidator's impact cost with $N-m$ predators and $m$ liquidators
    \begin{align*}\mathsf{LIC}_{N,m}(x)&:= \lambda\Bigg[\frac{ m(N-m) }{N(\beta T+1)}\\
        & + \frac{ m^2 N (N+1)\left( \left(\left(\beta T +\frac{1}{2}\right) (N+1) +3\right) e^{\frac{2 \left(N +1\right) \beta  T}{N -1}}-\frac{2 \left(N -1\right)}{N^2} \left(N e^{\frac{\left(N +1\right) \beta  T}{N -1}}+\frac{1}{4}\right) \right)}{\left(N \left(\left(\beta T +1\right)(N+1) + 2\right) e^{\frac{\left(N +1\right) \beta  T}{N -1}}-(N -1)\right)^{2}}\Bigg]x^2.
    \end{align*}
   This expression is quadratic in the parameter $m$ and the limiting behavior of the Cost of Predation in this setting will depend on the relative scaling of $m$ and $N$. For instance, in the extreme case $m=N$ we have $\mathsf{LIC}_{N,N}(x)=N\mathsf{PIC}_N(x)$.
\end{remark}

\section{Conclusion}\label{sec:conclusion}

We have derived the equilibria for $N$-player games with transient price impact, with and without regularization by an instantaneous cost on the trading rate. The equilibria are obtained in closed form. For the unregularized case, we have shown that an equilibrium exists only when a particular, time-dependent block cost is added. The equilibrium block cost is explained as the limit of the equilibrium instantaneous costs when their parameter $\eps\to0$. The fact that these costs do not vanish in the limit reflects that strategies become aggressive in the presence of competition. The limiting model is particularly tractable and we have given simple expressions for the impact costs of anarchy and predation. The presence of
predators is costly to the liquidator, with the cost stabilizing as the number of predators increases.

Our results give rise to several follow-up questions. First, one can ask if the aforementioned observations are specific to exponential kernels. We hope to show that similar results hold for a broad class of regular decay kernels, whereas singular kernels have a different behavior. In particular, because singular kernels charge an infinite cost to block trades, equilibria may exist without additional regularization. Second, the limit $N\to\infty$ gives rise to mean-field games (e.g., \cite{CardaliaguetLehalle.18}). It turns out that the two limits $N\to\infty$ and $\eps\to0$ do not commute, and we will relate the various limits in future research. Third, the present game uses the full information setup, like the vast majority of the literature. One can ask what happens, for instance, if agents do not know their competitors' initial inventories. Results in that direction are scarce; see \cite{MoallemiParkVanRoy.12} for a model where all impact is permanent and \cite{CasgrainJaimungal.19,CasgrainJaimungal.20} for a mean-field model. One can expect that in a partial information setup, traders' strategies reflect an incentive to hide the trading motive, as seen in \cite{MoallemiParkVanRoy.12}. Finally, an interesting extension is to allow for alpha signals and time-varying liquidity, and to study how these features alter the analysis.

\appendix
\section*{Appendix}

\section{Notational Remarks}\label{se:notation}

\paragraph{Integration.} We write $\int_a^b f\,d\mu$ to denote the integral $\int_{[a,b]}f\,d\mu$ of a function $f$ against a measure $\mu$ on $[a,b]\subset \mathbb{R}$, including any atoms $\mu$ may have at $a$ or $b$. When the left or right endpoint is not included, we write $\int_{a+}^bf\,d\mu :=\int_{(a,b]}f\,d\mu$ or $\int_{a}^{b-}f\,d\mu :=\int_{[a,b)}f\,d\mu$, respectively.

\paragraph{Spaces.} As usual, $L^2[0,T]$ is the Hilbert space of (equivalence classes of) real-valued square-integrable functions on $[0,T]$ with inner product
\[\langle v,w\rangle_{L^2} = \int_0^Tv_tw_t\,dt.\]
Similarly, $H^1[0,T]$ is the Sobolev space of functions $v\in L^2[0,T]$ that admit weak derivatives $\dot{v}\in L^2[0,T]$. It is a Hilbert space with inner product
\[\langle v,w\rangle_{H^1}=\langle v,w\rangle_{L^2} +\langle \dot{v},\dot{w}\rangle_{L^2}.\]
If $\mathbb{H}$ is any Hilbert space we write $\mathbb{H}^{\times N}$ to denote its $N$-fold Cartesian product.

\section{Reminder on Gateaux Derivative and $\Gamma$-Convergence}\label{app:opt.theory}

For ease of reference, this section collects some standard definitions and results to be used in the proofs below. 
Let $\mathbb{H}$ be a real separable Hilbert space with inner product $\langle\cdot,\cdot\rangle_\mathbb{H}$. For a real-valued function $F$ on $\mathbb{H}$ we define the \textit{Gateaux differential} of~$F$ at~$h\in \mathbb{H}$ in the direction $\eta\in \mathbb{H}$ by 
\[\delta_\eta F(h) = \lim_{\epsilon\to0} \frac{F(h+\epsilon\eta) - F(h)}{\epsilon} \]
when the limit exists. If the Gateaux differential exists for all directions $\eta \in\mathbb{H}$ and 
the mapping $\eta \mapsto \delta_\eta F(h)$ is a continuous linear functional on $\mathbb{H}$ then
we say that~$F$ is \textit{Gateaux differentiable} at~$h$ and call $\delta_\cdot F(h)$ the \textit{Gateaux derivative} of~$F$ at~$h$. %
By Riesz' representation theorem we can identify the derivative with an element $DF(h)\in \mathbb{H}$,
    \[\delta_\eta F(h) = \langle DF(h), \eta \rangle_{\mathbb{H}}, \ \ \ \forall\eta\in\mathbb{H}.\]
The following result is standard (see, e.g., \cite[Theorem 3.24]{Peypouquet.15}).
\begin{proposition}
If $F:\mathbb{H}\to \mathbb{R}\cup \{\infty\}$ is a proper\footnote{We say that $F:\mathbb{H}\to \mathbb{R}\cup \{\infty\}$ is proper if $\{h\in\mathbb{H}: F(h)<\infty\} \not=\emptyset$.} convex function with Gateaux derivative at every $h\in\mathbb{H}$ then $h^*$ is a minimizer of $F$ if and only if $\delta_\eta F(h^*) = 0$ for all $\eta \in \mathbb{H}$.
\end{proposition}
In particular, if $F$ is a proper convex function with a %
Gateaux derivative at every $h\in \mathbb{H}$, then $h^*$ is a minimizer of $F$ if and only if $DF(h^*) =0$. We will use this fact frequently.

Another notion used below is $\Gamma$-convergence, for which \cite{DalMaso.93} is a standard reference.
\begin{definition}\label{def:gamma.conv}
    We say that a sequence $F_n:\mathbb{H}\to\mathbb{R}\cup\{\infty\}$, $n\geq0$, $\Gamma$-converges to a function $F:\mathbb{H}\to\mathbb{R}\cup\{\infty\}$ (and write $F_n\xrightarrow{\Gamma} F$) if:
    \begin{enumerate}
        \item For all $h\in\mathbb{H}$ and all sequences $(h_n)_{n\geq0}$ with $h_n\to h$,
        \[F(h)\leq \liminf_{n\to\infty} F_n(h_n).\]
        \item For all $h\in \mathbb{H}$ there exists a sequence $(h_n)_{n\geq0}$ with $h_n\to h$ and
        \[F(h)\geq \limsup_{n\to\infty} F_n(h_n).\]
    \end{enumerate}
\end{definition}

The following result (see, e.g., \cite[Proposition 7.18]{DalMaso.93}) connects this type of convergence to the consistency of minimizers.

\begin{theorem}[Fundamental Theorem of $\Gamma$-Convergence\footnote{Other common variants of this theorem include an equicoercivity assumption on $(F_n)_{n\geq0}$ to imply the convergence of the minimum values.
}]\label{thm:fund.thm.gamma.conv}
Suppose that $F_n\xrightarrow{\Gamma} F$ and $h^*_n\in\argmin_{h\in\mathbb{H}}F_n(h)$, $n\geq0$. Then every limit point of $(h^*_n)_{n\geq0}$ is a minimizer of $F$. Moreover, if $(h^*_n)_{n\geq0}$ converges, so do the minimum values,
\[\lim_{n\to\infty} \inf_{h\in\mathbb{H}}F_n(h)=\inf_{h\in\mathbb{H}}F(h).\]
\end{theorem}

\section{Proofs for Section \ref{sec:prob.formulation}}\label{app:proofsProblemformulationAndUniqueness}

\subsection{Proposition \ref{prop:obj.func.rep}}\label{app:pf.obj.func.rep}

We can expand \eqref{eqn:orig.obj} to get
\begin{align*}J_{\cdot}(X^{i};\boldsymbol{X}^{-i})&=\mathbb{E}\bigg[\int_0^T P_{t-} dX_t^{i}+\frac{1}{2}\sum_{t\in[0,T]}  \Delta P_t \Delta X_t^{i} +\int_0^T I_{t-} dX_t^{i} + \frac{1}{2}\sum_{t\in[0,T]} \Delta I_t \Delta X_t ^{i}\\
&\quad \quad \quad \quad +(X_{0-}^iP_{0-}-X^{i}_TP_T)+C_{\cdot}(X^{i})\bigg].
\end{align*}
Using integration by parts,
\begin{align*}
    \int_0^T P_{t-} dX_t^{i}  &= P_TX_T^i-P_{0-}X_{0-}^i-\int_0^T X_{t-}^i dP_t -[X^i,P]_T
\end{align*}
and thus
\begin{align*}J_{\cdot}(X^{i};\boldsymbol{X}^{-i})&=\mathbb{E}\bigg[-\int_0^T X^{i}_{t-} dP_t - [X^{i},P]_T+\frac{1}{2}\sum_{t\in[0,T]}  \Delta P_t \Delta X_t^{i} \\
&\quad \quad \quad \quad +\int_0^T I_{t-} dX_t^{i} + \frac{1}{2}\sum_{t\in[0,T]} \Delta I_t \Delta X_t^{i}  +C_{\cdot}(X^{i})\bigg].
\end{align*}
Both $\int_0^T X^{i}_{t-} dP_t$ and $[X^{i},P]_T=\sum_{t\in[0,T]} \Delta X_t^{i} \Delta P_t=\int_0^T\Delta X^i_tdP_t$ are terminal values of martingales since $P$ is a c\`adl\`ag square-integrable martingale and $X^i$ is a bounded predictable process of essentially bounded variation. Hence, these terms vanish in expectation. 
Since these arguments hold for arbitrary admissible~$X^{i}$, we obtain the expression in \eqref{eqn:obj.impact.rep}  of Proposition \ref{prop:obj.func.rep}. 

To obtain the remaining representation in \eqref{eqn:obj.schied.rep}, we insert the form of $I$ into \eqref{eqn:obj.impact.rep} to get
\begin{align}\label{eqn:obj.intermediate.rep}J_{\cdot}(X^{i};&\boldsymbol{X}^{-i})
= \mathbb{E}\bigg[C_{\cdot}(X^{i})+\lambda \int_0^T\int_0^{t-} e^{-\beta (t-s)} \sum_{j=1}^NdX_s^{j} dX_t^{i} + \frac{\lambda}{2}\sum_{j=1}^N\sum_{t\in[0,T]}  \Delta X_t^{j} \Delta X_t ^{i}\bigg].%
\end{align}
Observe that by splitting and interchanging the order of integration, we can write
\begin{align*}\int_0^T\int_0^{t-} &e^{-\beta (t-s)} dX_s^{i} dX_t^{i}\\
&=\frac{1}{2}\int_0^T\int_0^{t-} e^{-\beta (t-s)} dX_s^{i} dX_t^{i}+\frac{1}{2}\int_0^T\int_0^{t-} e^{-\beta (t-s)} dX_s^{i} dX_t^{i}\\
&=\frac{1}{2}\int_0^T\int_0^{t-} e^{-\beta (t-s)} dX_s^{i} dX_t^{i}+\frac{1}{2}\int_0^T\int_{s}^T e^{-\beta (t-s)} dX_t^{i} dX_s^{i}-\frac{1}{2}\sum_{s\in[0,T]}(\Delta X^i_s)^2\\
&=\frac{1}{2}\int_0^T\int_0^T e^{-\beta |t-s|} dX_s^{i} dX_t^{i}-\frac{1}{2}\sum_{t\in[0,T]}(\Delta X^i_t)^2.
\end{align*}
Substituting this expression into \eqref{eqn:obj.intermediate.rep} completes the proof.
\qed

\subsection{Proposition \ref{prop:nash.eq.unique}}\label{app:pf.nash.eq.unique}

Suppose there are two distinct equilibria $\boldsymbol{X}^0$ and $\boldsymbol{X}^1$.
Define
    \[\boldsymbol{X}^\alpha = \alpha \boldsymbol{X}^1 +(1-\alpha)\boldsymbol{X}^0, \ \ \ \alpha\in[0,1], \]
and
\begin{equation*}V(\alpha) = \sum_{i=1}^N\left(J_\cdot(X^{\alpha,i};\boldsymbol{X}^{0,-i})+J_\cdot(X^{1-\alpha,i};\boldsymbol{X}^{1,-i})\right).\end{equation*}
A consequence of Lemma \ref{lem:strict.convex} is that $V$ is strictly convex. Moreover, by the Nash equilibrium definition,
\[V(\alpha)\geq \sum_{i=1}^N \left(J(X^{0,i};\boldsymbol{X}^{0,-i})+J(X^{1,i};\boldsymbol{X}^{1,-i}) \right),\]
and this lower bound is attained by setting $\alpha=0$. Next, we will differentiate $V$ at $0$ to obtain the contradiction $\dot{V}(0)<0$.

To facilitate this we deal with each term in the expression for $J_{\cdot}$ separately. Let 
\begin{align*}\mathfrak{J}_1^i(\alpha) &= \frac{\lambda}{2}\int_0^T\int_0^T e^{-\beta |t-s|} dX_s^{\alpha,i} dX_t^{\alpha, i}+\frac{\lambda}{2}\int_0^T\int_0^T e^{-\beta |t-s|} dX_s^{1-\alpha,i} dX_t^{1-\alpha, i}.
\end{align*}
By differentiating and setting $\alpha=0$ we find
\begin{align*}
    \dot{\mathfrak{J}}^i_1(0) &=-\lambda\int_0^T\int_0^T e^{-\beta |t-s|} d(X_s^{1,i} -X_s^{0, i})d(X_t^{1,i} -X_t^{0, i}).
\end{align*}
Summing over $i$,
\begin{equation} \label{eqn:J1.sum} \sum_{i=1}^N\dot{\mathfrak{J}}^i_1(0)=-\lambda\sum_{i=1}^N\int_0^T\int_0^T e^{-\beta |t-s|} d(X_s^{1,i} -X_s^{0, i})d(X_t^{1,i} -X_t^{0, i}).
\end{equation}
Next, let
\begin{align*}
    \mathfrak{J}_2^i(\alpha) &= \lambda \int_0^T\int_0^{t-} e^{-\beta (t-s)} \sum_{j\not=i}dX_s^{0,j} dX_t^{\alpha, i}+\lambda \int_0^T\int_0^{t-} e^{-\beta (t-s)} \sum_{j\not=i}dX_s^{1,j} dX_t^{1-\alpha, i}.
\end{align*}
Differentiating,
\begin{align*}\dot{\mathfrak{J}}_2^i(\alpha)&=\lambda \int_0^T\int_0^{t-} e^{-\beta (t-s)} \sum_{j\not=i}dX_s^{0,j} d(X^{1,i}-X^{0,i})\\
&\quad -\lambda \int_0^T\int_0^{t-} e^{-\beta (t-s)} \sum_{j\not=i}dX_s^{1,j} d(X^{1,i}- X^{0,i})\\
&=-\lambda \int_0^T\int_0^{t-} e^{-\beta (t-s)} \sum_{j\not=i}d(X_s^{1,j} - X_s^{0,j})d(X^{1,i}_t- X^{0,i}_t),
\end{align*}
which is constant in $\alpha$. If we sum over $i$ and manipulate the terms,

\begin{align*}\sum_{i=1}^N\dot{\mathfrak{J}}_2^i(0)&= -\lambda \sum_{i=1}^N\sum_{j\not=i}\int_0^T\int_0^{t-} e^{-\beta (t-s)} d(X_s^{1,j} - X_s^{0,j}) d(X^{1,i}_t- X^{0,i}_t)\\
&= -\frac{\lambda}{2} \sum_{i=1}^N\sum_{j\not=i}\int_0^T\int_0^{t-} e^{-\beta (t-s)} d(X_s^{1,j} - X_s^{0,j}) d(X^{1,i}_t- X^{0,i}_t)\\
& \quad -\frac{\lambda}{2} \sum_{i=1}^N\sum_{j\not=i}\int_0^T\int_{s+}^T e^{-\beta (t-s)}  d(X^{1,i}_t- X^{0,i}_t) d(X_s^{1,j} - X_s^{0,j})\\
&= -\frac{\lambda}{2} \sum_{i=1}^N\sum_{j\not=i}\int_0^T\int_0^T e^{-\beta |t-s|} d(X_s^{1,j} - X_s^{0,j}) d(X^{1,i}_t- X^{0,i}_t)\\
&\quad +\frac{\lambda}{2}\sum_{i=1}^N\sum_{j\not=i}\sum_{t\in[0,T]}\Delta(X_t^{1,j} - X_t^{0,j})\Delta(X_t^{1,i} - X_t^{0,i}).
\end{align*}
Due to the positive definiteness of the kernel $e^{-\beta |t-s|}$ (in the sense of Bochner) we have %
\[-\sum_{i}\sum_{j\not=i}\int_0^T\int_0^Te^{-\beta|t-s|}dM_s^jdM_t^i\leq \sum_{i}\int_0^T\int_0^Te^{-\beta|t-s|}dM_s^idM_t^i\] for arbitrary Lebesgue--Stieltjes integrators $M^i$. Applying this,
\begin{align}\label{eqn:J2.sum.bd}
\sum_{i=1}^N\dot{\mathfrak{J}}_2^i(0)&\leq  \frac{\lambda}{2} \sum_{i=1}^N\int_0^T\int_0^T e^{-\beta |t-s|} d(X_s^{1,i} - X_s^{0,i}) d(X^{1,i}_t- X^{0,i}_t)\\
&\quad +\frac{\lambda}{2}\sum_{i=1}^N\sum_{j\not=i}\sum_{t\in[0,T]}\Delta(X_t^{1,j} - X_t^{0,j})\Delta(X_t^{1,i} - X_t^{0,i}).\nonumber
\end{align}
Next, we treat the jumps. Let
\begin{align*}\mathfrak{J}_3^i(\alpha)&=\frac{\lambda}{2}\sum_{j\not=i}\sum_{t\in[0,T]} \Delta X_t ^{\alpha, i} \Delta X_t^{0,j}+\frac{\lambda}{2}\sum_{j\not=i}\sum_{t\in[0,T]} \Delta X_t ^{1-\alpha, i} \Delta X_t^{1,j}.
\end{align*}
Differentiating gives
\begin{align*}\dot{\mathfrak{J}}_3^i(\alpha)&=-\frac{\lambda}{2}\sum_{j\not=i}\sum_{t\in[0,T]} \Delta \left( X_t^{0,i}-X_t ^{1, i}\right)\Delta X_t^{0,j}+\frac{\lambda}{2}\sum_{j\not=i}\sum_{t\in[0,T]} \Delta \left(X_t^{0,i}-X_t ^{1, i}\right) \Delta X_t^{1,j}.
\end{align*}
Then, by collecting terms and summing over $i$,
\begin{equation}\label{eqn:J3.sum.bd}\sum_{i=1}^N\dot{\mathfrak{J}}_3^i(0)=-\frac{\lambda}{2}\sum_{i=1}^N\sum_{j\not=i}\sum_{t\in[0,T]} \Delta \left(X_t^{1,i}-X_t^{0, i}\right)\Delta \left(X_t^{1,j}-X_t^{0,j}\right).
\end{equation}

It remains to treat the terms arising from the cost $C_{\cdot}$. Since $\boldsymbol{X}^0$ and $\boldsymbol{X}^1$ are assumed to be Nash equilibria they (and their convex combinations) must have a finite cost (as measured by $C_{\cdot}$) almost surely. Therefore, we may ignore the characteristic function terms. We only consider the case of $C_B(\cdot)$ as the proof with $C_A(\cdot)$ or $C_{A'}(\cdot)$ is similar. Let
\begin{align*}\mathfrak{J}_4^i(\alpha) &=C_B(X^{\alpha,i})+C_B(X^{1-\alpha,i})
=\sum_{t\in[0,T]} \frac{\vartheta_t}{2}(\Delta X_t^{\alpha, i})^2+\sum_{t\in[0,T]} \frac{\vartheta_t}{2}(\Delta X_t^{1-\alpha, i})^2.
\end{align*}
Differentiating and setting $\alpha=0$,
\begin{align*}\dot{\mathfrak{J}}_4^i(0)&=\sum_{t\in[0,T]} \vartheta_t\left(\Delta X_t^{0,i}(\Delta X_t^{1,i}-\Delta X_t^{0,i}) -\Delta X_t^{1,i}(\Delta X_t^{1,i}-\Delta X_t^{0,i})\right)\nonumber\\
&=-\sum_{t\in[0,T]} \vartheta_t(\Delta X_t^{1,i}-\Delta X_t^{0,i})^2.\nonumber
\end{align*}
Summing over $i$, we have
\begin{align}\sum_{i=1}^N\dot{\mathfrak{J}}_4^i(0)=-\sum_{i=1}^N\sum_{t\in[0,T]} \vartheta_t(\Delta X_t^{1,i}-\Delta X_t^{0,i})^2\leq 0. \label{eqn:J4.sum.bd}
\end{align}
Aggregating the above expressions we recover $V(\alpha)$,
\[V(\alpha) = \mathbb{E}\left[\sum_{i=1}^N(\mathfrak{J}_1^i(\alpha)+\mathfrak{J}_2^i(\alpha)+\mathfrak{J}_3^i(\alpha)+\mathfrak{J}_4^i(\alpha))\right].\]
By an application of the dominated convergence theorem (or a direct verification) we may pass the derivative under the expectation to obtain
\begin{align*}\dot{V}(0) &= \mathbb{E}\left[\sum_{i=1}^N(\dot{\mathfrak{J}}_1^i(0)+\dot{\mathfrak{J}}_2^i(0)+\dot{\mathfrak{J}}_3^i(0)+\dot{\mathfrak{J}}_4^i(0))\right]\\
&\leq -\mathbb{E}\left[\frac{\lambda}{2}\sum_{i=1}^N\int_0^T\int_0^T e^{-\beta |t-s|} d(X_s^{1,i} -X_s^{0, i})d(X_t^{1,i} -X_t^{0, i})\right].
\end{align*}
The last inequality follows from \eqref{eqn:J1.sum}, \eqref{eqn:J2.sum.bd}, \eqref{eqn:J3.sum.bd}, and \eqref{eqn:J4.sum.bd}. Finally, since $\boldsymbol{X}^0\not=\boldsymbol{X}^1$ and the kernel $e^{-\beta |t-s|}$ is strictly positive definite, we conclude $\dot{V}(0)<0$. This contradiction completes the proof.
\qed

\section{Proofs for Section \ref{sec:N.inst.cost}}\label{app:N.inst.cost}

\subsection{Equilibrium with Terminal Inventory Penalty}

\subsubsection{Lemma \ref{lem:term.pen.ODE}}
Taking the Gateaux differential of $\cJ_A(\cdot; \boldsymbol{v}^{-i})$ in~\eqref{eqn:cJ.A} in an arbitrary direction $\eta\in L^2[0,T]$,
\begin{align*}
    \delta_\eta \cJ_A(v^{i};\boldsymbol{v}^{-i})&=\int_0^T(\delta_\eta I_tv_t^{i}+\eta_tI_t+\varepsilon v_t^{i}\eta_t)dt+\varphi X_T^{i}\int_0^T\eta_tdt,
\end{align*}
where
\begin{equation*}
    \delta_\eta I_t=\int_0^t e^{-\beta(t-s)}\lambda\eta_sds, \ \ t\in[0,T].
\end{equation*}
By changing the order of integration,
\begin{align*}
    \int_0^T\delta_\eta I_tv_t^{i}dt&=\int_0^T\int_0^t\lambda e^{-\beta(t-s)}v_t^{i}\eta_sdsdt
    =\int_0^T\int_s^T\lambda e^{-\beta(t-s)}v_t^{i}\eta_sdtds.
\end{align*}
Swapping the roles of the integration variables $s$ and $t$, we get the expression
\begin{equation*}
\int_0^T\delta_\eta I_tv_t^{i}dt=\int_0^T\int_t^T\lambda e^{-\beta(s-t)}v_s^{i}ds\eta_tdt,
\end{equation*}
which leads to the representation
\begin{equation*}
    \delta_\eta \mathcal{J}_A (v^{i};\boldsymbol{v}^{-i})=\int_0^T[Y_t^{i}+I_t+\varepsilon v_t^{i}+\varphi X_T^{i}]\eta_tdt
\end{equation*}
for
\begin{equation*}%
    Y_t^{i}:=\int_t^T\lambda e^{-\beta (s-t)}v_s^{i} ds, \ \ \ t\in [0,T].
\end{equation*}

From this, we see that $\eta \mapsto \delta_\eta \mathcal{J}_A(v^{i};\boldsymbol{v}^{-i})$ is a well-defined continuous linear operator on $L^2[0,T]$ for every $v^{i}\in L^2[0,T]$. Letting
\begin{equation*}
   D \mathcal{J}_A(v^{i};\boldsymbol{v}^{-i})=\left(Y_t^{i}+I_t+\varepsilon v_t^{i}+\varphi X_T^{i}\right)\in L^2[0,T] 
\end{equation*}
we have the following form of the derivative in terms of the usual $L^2[0,T]$ inner product,
\[\delta_\eta \mathcal{J}_A(v^{i};\boldsymbol{v}^{-i})=\langle D \mathcal{J}_A(v^{i};\boldsymbol{v}^{-i}),\eta\rangle_{L^2[0,T]}.\]
By standard optimality theory in Hilbert spaces (see Appendix \ref{app:opt.theory}), an element $v^{i}\in L^2[0,T]$ is a minimizer of $\cJ_A(\cdot,\boldsymbol{v}^{-i})$ if and only if  \[D \mathcal{J}_A(v^{i};\boldsymbol{v}^{-i}) =0,\]
where the equality is to be understood in the $L^2$ sense. Equivalently, the representative $D \mathcal{J}_A(v^{i};\boldsymbol{v}^{-i})$ of the Gateaux derivative must be $0$ almost everywhere; that is,
\begin{equation}\label{eqn:foc.inst.cost}
    \boxed{Y_t^{i}+I_t+\varepsilon v_t^{i}=-\varphi X_T^{i},  \ \ \ \text{a.e.} \ \ \  t\in[0,T].}
\end{equation}
Since $Y^{i}$ and $I$ are continuous by definition, we have that any $v^{i}$ satisfying \eqref{eqn:foc.inst.cost} is equal almost everywhere to a continuous function. From here on, we identify any such $v^{i}\in L^2[0,T]$ with its continuous version. Then, by observing that $Y^{i}$ and $I$ are differentiable (and that $X_T^{i}$ is a constant), we have that any $v^{i}$ satisfying \eqref{eqn:foc.inst.cost} is also differentiable. The condition \eqref{eqn:foc.inst.cost} then implies that
\begin{equation}\label{eqn:foc.der.inst.cost}
    \boxed{\dot{Y}_t^{i}+\dot{I}_t+\varepsilon \dot{v}_t^{i}=0, \ \ \  t\in[0,T].}
\end{equation}

To form a Nash equilibrium we require the simultaneous optimality of each of the controls in the strategy profile $\boldsymbol{v}$. As a result, \eqref{eqn:foc.inst.cost} and \eqref{eqn:foc.der.inst.cost} must hold for all $i=1,\dots,N$. In this way, we arrive at the following equilibrium system of ODEs,
\begin{align*}
    0&=\dot{Y}^{i}_t+\dot{I}_t+\varepsilon \dot{v}_t^{i}, \ \ \ i=1,\dots,N,\\
    \dot{Y}^{i}_t&=\beta Y^{i}_t-\lambda v^{i}_t, \ \ \ i=1,\dots,N,\\
    \dot{I}_t&=-\beta I_t+\lambda\sum_{i=1}^Nv_t^{i},
\end{align*}
subject to the initial and terminal conditions
\begin{align*}
    I_0&=0,\\
    Y_T^{i}&=0, \ \ \ i=1,\dots, N,\\
    v_T^{i}&= -\varepsilon^{-1}\left[\varphi X_T^{i}+I_T\right],  \ \ \ i=1,\dots, N,
\end{align*}
where the last equation arises by taking $t\uparrow T$ in \eqref{eqn:foc.inst.cost}.

It is easy to check that the satisfaction of this system implies \eqref{eqn:foc.inst.cost} for all $i=1,\dots,N$. Hence, enforcing the ODEs or the concurrent satisfaction of \eqref{eqn:foc.inst.cost} for all $i=1,\dots, N$ is equivalent. To complete the proof, we rearrange the system of ODEs to recover the standard form reported in Lemma~\ref{lem:term.pen.ODE}. \qed

\subsubsection{Theorem \ref{thm:equil.term.pen}}\label{se:proof.of.equil.term.pen}

Let $\boldsymbol{F} = (I, Y^{1},\dots, Y^{N},v^{1},\dots, v^{N})^\top$. The system in Lemma \ref{lem:term.pen.ODE} can be written in matrix form,
\begin{equation}\label{eqn:term.pen.ode.matrix}
\dot{\boldsymbol{F}}_t=\mathsf{A} \boldsymbol{F}_t
\end{equation}
for 
\[\mathsf{A} = \begin{bmatrix}
    -\beta & 0 & 0 & \dots  & 0 & 0 & \lambda & \lambda  & \dots & \lambda & \lambda \\
    0 & \beta & 0 & \dots & 0 & 0 & -\lambda & 0 & \dots & 0 & 0 \\
     0 & 0 & \beta & \dots & 0 & 0 & 0 & -\lambda & \dots & 0 & 0\\
      \vdots & \vdots & \vdots & \ddots & \vdots & \vdots & \vdots & \vdots & \ddots &  \vdots & \vdots \\
     0 & 0 & 0 & \dots  & \beta & 0 & 0 &  0 & \dots & -\lambda & 0\\
     0 & 0 & 0 & \dots & 0 & \beta & 0 &  0 & \dots & 0 & -\lambda\\
     \frac{\beta}{\varepsilon} & -\frac{\beta}{\varepsilon} & 0 & \dots & 0 & 0 & 0 & -\frac{\lambda}{\varepsilon} & \dots  & -\frac{\lambda}{\varepsilon} & -\frac{\lambda}{\varepsilon} \\
     \frac{\beta}{\varepsilon} & 0 & -\frac{\beta}{\varepsilon} & \dots & 0 & 0 & -\frac{\lambda}{\varepsilon} & 0 & \dots  & -\frac{\lambda}{\varepsilon} & -\frac{\lambda}{\varepsilon} \\
    \vdots & \vdots & \vdots & \ddots & \vdots & \vdots & \vdots & \vdots & \ddots &  \vdots & \vdots \\
     \frac{\beta}{\varepsilon} & 0 & 0 & \dots & -\frac{\beta}{\varepsilon} & 0 & -\frac{\lambda}{\varepsilon} & - \frac{\lambda}{\varepsilon} & \dots  & 0 & -\frac{\lambda}{\varepsilon} \\
     \frac{\beta}{\varepsilon} & 0 & 0 & \dots & 0 & -\frac{\beta}{\varepsilon} & -\frac{\lambda}{\varepsilon} & - \frac{\lambda}{\varepsilon} & \dots  & -\frac{\lambda}{\varepsilon} & 0 \\
\end{bmatrix}.\]
It can be directly verified that the matrix $\mathsf{A}$ has eigenvalues
\[z_1 = \frac{-\lambda(N-1)+\sqrt{(N-1)^2\lambda^2+4\beta \varepsilon (N+1) \lambda +4\beta^2\varepsilon^2}}{2\varepsilon},\]
\[z_2 = \frac{-\lambda(N-1)-\sqrt{(N-1)^2\lambda^2+4\beta \varepsilon (N+1) \lambda +4\beta^2\varepsilon^2}}{2\varepsilon},\]
\[z_{2+i} = \beta + \varepsilon^{-1} \lambda, \ \ \ \  i=1,\dots,N-1,\] and $z_{N+1+i}=0$ for $i=1,\dots,N$, with associated eigenvectors
\[\boldsymbol{q}_1=\frac{N\lambda}{z_1+\beta} \boldsymbol{e}_1-\sum_{j=1}^N\frac{\lambda}{z_1-\beta}\boldsymbol{e}_{1+j}+\sum_{j=1}^N\boldsymbol{e}_{N+1+j},\]
\[\boldsymbol{q}_2=\frac{N\lambda}{z_2+\beta} \boldsymbol{e}_1-\sum_{j=1}^N\frac{\lambda}{z_2-\beta}\boldsymbol{e}_{1+j}+\sum_{j=1}^N\boldsymbol{e}_{N+1+j},\]
\[\boldsymbol{q}_{2+i} = \varepsilon \boldsymbol{e}_2-\varepsilon \boldsymbol{e}_{2+i}-\boldsymbol{e}_{N+2}+\boldsymbol{e}_{N+2+i}, \ \ \ i=1,\dots,N-1,\]
\[\boldsymbol{q}_{N+1+i} = \frac{\lambda}{\beta} \boldsymbol{e}_1+\frac{\lambda}{\beta}\boldsymbol{e}_{1+i}+\boldsymbol{e}_{N+1+i}, \ \ \ i=1,\dots,N.\] 
Here, we write $\boldsymbol{e}_i$ for the $i$th Euclidean basis vector. These eigenvalues and eigenvectors define the columns of the fundamental matrix
\[\mathsf{Q}_t=\begin{bmatrix}
    \frac{N\lambda}{z_1+\beta}e^{z_1t} & \frac{N\lambda}{z_2+\beta}e^{z_2t} & 0 & 0 & \dots & 0 & \frac{\lambda}{\beta} & \frac{\lambda}{\beta} & \frac{\lambda}{\beta}  & \dots & \frac{\lambda}{\beta}\\
    -\frac{\lambda}{z_1-\beta}e^{z_1t} & -\frac{\lambda}{z_2-\beta}e^{z_2t} & \varepsilon e^{z_3t} & \varepsilon e^{z_3t} & \dots & \varepsilon e^{z_3t} & \frac{\lambda}{\beta} & 0 & 0 & \dots & 0 \\
    -\frac{\lambda}{z_1-\beta}e^{z_1t} & -\frac{\lambda}{z_2-\beta}e^{z_2t} & -\varepsilon e^{z_3t} & 0 & \dots & 0 & 0 & \frac{\lambda}{\beta} & 0 & \dots & 0\\
    -\frac{\lambda}{z_1-\beta}e^{z_1t} & -\frac{\lambda}{z_2-\beta}e^{z_2t} & 0  & -\varepsilon e^{z_3t} & \dots & 0 & 0 & 0 & \frac{\lambda}{\beta}  & \dots & 0\\
    \vdots & \vdots & \vdots  & \vdots & \ddots & \vdots & \vdots & \vdots  & \vdots  & \ddots & \vdots\\
    -\frac{\lambda}{z_1-\beta}e^{z_1t} & -\frac{\lambda}{z_2-\beta}e^{z_2t} & 0  & 0 & \dots & -\varepsilon e^{z_3t} & 0 & 0 & 0 & \dots  & \frac{\lambda}{\beta} \\
    e^{z_1t} & e^{z_2t} & -e^{z_3t} & -e^{z_3t} & \dots & -e^{z_3t} & 1 & 0 & 0 & \dots & 0\\
    e^{z_1t} & e^{z_2t} & e^{z_3t} & 0 & \dots & 0 & 0 & 1 & 0 & \dots & 0\\
    e^{z_1t} & e^{z_2t} & 0 & e^{z_3t} & \dots & 0 & 0 & 0 & 1 & \dots & 0\\
    \vdots & \vdots & \vdots & \vdots & \ddots & \vdots & \vdots & \vdots & \vdots & \ddots & \vdots\\
    e^{z_1t} & e^{z_2t} & 0 & 0 & \dots & e^{z_3t} & 0 & 0 & 0 & \ddots & 1\\
\end{bmatrix}.\]
Note that we have made use of the equality $z_3=z_{2+i}$ for $i=1,\dots,N-1$. Let  $\boldsymbol{c} = (c_1,\dots,c_{2N+1})^\top$ be a vector of constants. The fundamental solution to \eqref{eqn:term.pen.ode.matrix} takes the form
\begin{equation}\label{eqn:fund.sol.inst.cost}
    \boldsymbol{F}_t = \mathsf{Q}_t\boldsymbol{c} =c_1\boldsymbol{q}_1e^{z_1 t}+c_2\boldsymbol{q}_2e^{z_2 t}+\sum_{i=1}^{N-1} c_{2+i}\boldsymbol{q}_{2+i}e^{z_3 t}+\sum_{i=1}^N c_{N+1+i}\boldsymbol{q}_{N+1+i}.
\end{equation}
To solve for $\boldsymbol{c}$ we enforce the boundary conditions from Lemma \ref{lem:term.pen.ODE}. 
To begin, we need an expression for the inventory processes. By integrating the expressions for $v^{i}$ in \eqref{eqn:fund.sol.inst.cost} we have that $X^{1}$ is given by
\begin{equation}\label{eqn:x1}
    X^{1}_t = x^{1}+ \frac{c_1}{z_1}(e^{z_1 t}-1) + \frac{c_2}{z_2}(e^{z_2 t}-1) - (e^{z_{3} t}-1)\sum_{j=2}^{N}\frac{c_{1+j}}{z_3}  + c_{N+2}t,
\end{equation}
and
\begin{equation}\label{eqn:xi}
X^{i}_t = x^{i}+ \frac{c_1}{z_1}(e^{z_1 t}-1) + \frac{c_2}{z_2}(e^{z_2 t} -1)+\frac{c_{1+i}}{z_3}(e^{z_{3} t}-1)+ c_{N+1+i}t, \ \ \ i=2,\dots, N.
\end{equation}
Then, writing the boundary conditions in terms of $\boldsymbol{c}$ we get the system of $2N+1$ equations,

\begin{equation*}
    \begin{cases}
    c_1\frac{N\lambda}{z_1+\beta}+c_2\frac{N\lambda}{z_2+\beta}+\frac{\lambda}{\beta}\sum_{j=1}^Nc_{N+1+j}=0,\\
    \\
    -c_1\frac{\lambda}{z_1-\beta}e^{z_1 T}-c_2\frac{\lambda}{z_2-\beta}e^{z_2 T}+\varepsilon e^{z_3 T}\sum_{j=2}^Nc_{1+j}+\frac{\lambda}{\beta}c_{N+2}=0,\\
    \\
    -c_1\frac{\lambda}{z_1-\beta}e^{z_1 T}-c_2\frac{\lambda}{z_2-\beta}e^{z_2 T}-\varepsilon e^{z_3 T}c_{1+i}+\frac{\lambda}{\beta}c_{N+1+i}=0, & i=2,\dots,N,\\
    \\
    c_1e^{z_1 T} + c_2 e^{z_2 T} - e^{z_{3} T}\sum_{j=2}^{N}c_{1+j} + c_{N+2} \\
    \quad =-\varepsilon^{-1}\varphi\left(x^{1}+ \frac{c_1}{z_1}(e^{z_1 T}-1) + \frac{c_2}{z_2}(e^{z_2 T}-1)\right)\\
    \quad\quad +\varepsilon^{-1}\varphi\left( (e^{z_{3} T}-1)\sum_{j=2}^{N}\frac{c_{1+j}}{z_3}  - c_{N+2}T\right)\\
    \quad \quad - \varepsilon^{-1}\left(c_1 \frac{N\lambda}{z_1+\beta} e^{z_1 T} + c_2 \frac{N\lambda}{z_2+\beta} e^{z_2 T}+\frac{\lambda}{\beta}\sum_{j=1}^Nc_{N+1+j}\right),\\
    \\
     c_1e^{z_1 T} + c_2 e^{z_2 T} +c_{1+i}e^{z_{3} T} + c_{N+1+i}\\
     \quad =-\varepsilon^{-1}\varphi \left(x^{i}+ \frac{c_1}{z_1}(e^{z_1 T}-1) + \frac{c_2}{z_2}(e^{z_2 T} -1)\right)\\
     \quad\quad -\varepsilon^{-1}\varphi\left(\frac{c_{1+i}}{z_3}(e^{z_{3} T}-1)+ c_{N+1+i}T\right)\\
     \quad \quad - \varepsilon^{-1}\left(c_1 \frac{N\lambda}{z_1+\beta} e^{z_1 T} + c_2 \frac{N\lambda}{z_2+\beta} e^{z_2 T}+\frac{\lambda}{\beta}\sum_{j=1}^Nc_{N+1+j}\right), & i = 2,\dots,N.
    \end{cases}
\end{equation*}
We will solve this system in parts. First, we reduce to a tractable three dimensional system. Summing over the last $N$ equations and rearranging yields
\begin{align}
    -\varepsilon^{-1}\varphi\sum_{j=1}^Nx^{j}&=Nc_1\left[e^{z_1 T} +\frac{\varepsilon^{-1}\varphi}{z_1}[e^{z_1 T}-1] +\varepsilon^{-1} \frac{N\lambda}{z_1+\beta} e^{z_1 T}\right]  \nonumber \\
    &\quad+ Nc_2 \left[e^{z_2 T}+\frac{\varepsilon^{-1}\varphi}{z_2}[e^{z_2 T} -1] +\varepsilon^{-1} \frac{N\lambda}{z_2+\beta} e^{z_2 T}\right] \label{eqn:sum.1} \\ 
    &\quad +\left[1+\varepsilon^{-1}\varphi T+N\varepsilon^{-1}\frac{\lambda}{\beta}\right] \sum_{j=1}^Nc_{N+1+j}.\nonumber
\end{align}
At the same time, summing over the $2$nd to $(N+1)$th equations gives
\begin{equation}\label{eqn:sum.2}
    0=-c_1\frac{N\lambda}{z_1-\beta}e^{z_1 T}-c_2\frac{N\lambda}{z_2-\beta}e^{z_2 T}+\frac{\lambda}{\beta}\sum_{i=1}^Nc_{N+1+i}.
\end{equation}
Coupling \eqref{eqn:sum.1} and \eqref{eqn:sum.2} with the first equation we can solve for $c_1$, $c_2$ and $\sum_{i=1}^Nc_{N+1+i}$. Written in terms of the constants in Table \ref{tab:constants}, we get
\begin{equation}\label{eqn:c1.c2.sum}
    c_1 = -\frac{\varphi}{\varepsilon \psi} \overline{x}, \ \ \ c_2 = \frac{\gamma_1\varphi}{\varepsilon \gamma_2\psi} \overline{x}, \ \ \ \sum_{i=1}^Nc_{N+1+i}= -\frac{N\beta \rho_-\varphi}{\varepsilon \psi} \overline{x}.
\end{equation}
We now pare off the $2$nd and the $(N+2)$th equations. After substituting in \eqref{eqn:c1.c2.sum} and rearranging we get (again in terms of the constants in Table \ref{tab:constants})
\begin{equation}\label{eqn:simplified.2}
\frac{\lambda}{\beta}c_{N+2}+\varepsilon e^{z_3 T}\sum_{j=2}^Nc_{1+j}=-\frac{\lambda \rho_{-}\varphi}{\varepsilon\psi}\overline{x},
\end{equation}
\begin{align}\label{eqn:simplified.N+2}
    -\left[e^{z_3 T} +\frac{\varphi}{\varepsilon }\frac{e^{z_3 T}-1}{z_3}\right]\sum_{j=2}^Nc_{1+j} + &\left[1+\frac{\varphi}{\varepsilon} T\right] c_{N+2}
    =-\frac{\varphi}{\varepsilon}( x^{1}-\overline{x})-\beta\rho_{-}\left[1+\frac{\varphi}{\varepsilon} T\right]\frac{\varphi}{\varepsilon\psi}\overline{x}. %
\end{align}
Solving \eqref{eqn:simplified.2} and \eqref{eqn:simplified.N+2} yields
\begin{equation}\label{eqn:cN+2}
c_{N+2} = -\frac{\beta \varphi}{\varepsilon\xi}(x^{1}-\overline{x})-\frac{\beta \rho_{-}\varphi}{\varepsilon \psi}\overline{x},
\end{equation}
\begin{equation}\label{eqn:sum.c1+i}
    \sum_{j=2}^Nc_{1+j}=\frac{\lambda\varphi}{\varepsilon^2 \xi  e^{z_3T}}(x^{1}-\overline{x}).
\end{equation}
At last, we turn to the $3$rd to $(N+1)$th equations and the last $N-1$ equations. Substituting in the existing solutions and collecting terms we arrive at the system
\begin{equation*}
    -\varepsilon e^{z_3 T}c_{1+i}+\frac{\lambda}{\beta}c_{N+1+i}=-\frac{\lambda \varphi\rho_{-}}{\varepsilon\psi}\overline{x}, \ \ \ i=2,\dots,N,
\end{equation*}
\begin{align*}
    \left[e^{z_3 T} +\frac{\varphi}{\varepsilon }\frac{e^{z_3 T}-1}{z_3}\right]& c_{1+i} + \left[1+\frac{\varphi}{\varepsilon}T\right]c_{N+1+i}\\
    &=-\frac{\varphi}{\varepsilon}( x^{i}-\overline{x})-\beta\rho_{-}\left[1+\frac{\varphi}{\varepsilon} T\right]\frac{\varphi}{\varepsilon\psi}\overline{x} ,\ \ \ i = 2,\dots,N.\nonumber %
\end{align*}
These equations can be solved in pairs. By doing so, one finds
\begin{equation}\label{eqn:cN+1+i}
c_{N+1+i} = -\frac{\beta\varphi}{\varepsilon\xi}(x^{i}-\overline{x})-\frac{\beta\rho_{-}\varphi}{\varepsilon \psi}\overline{x}, \ \ \ i=2,\dots,N,
\end{equation}
\begin{equation}\label{eqn:c1+i}
c_{1+i}=-\frac{\lambda\varphi}{\varepsilon^2 \xi e^{z_3T}}(x^{i}-\overline{x}), \ \ \ i=2,\dots,N.
\end{equation}
To complete the proof we substitute \eqref{eqn:c1.c2.sum}, \eqref{eqn:cN+2}, \eqref{eqn:sum.c1+i},  \eqref{eqn:cN+1+i}, and \eqref{eqn:c1+i} into \eqref{eqn:x1} and \eqref{eqn:xi}, and collect terms. \qed

\subsubsection{Corollary \ref{cor:eq.cost.term.pen}}

Using the form of the equilibrium in Theorem \ref{thm:equil.term.pen}, we compute the equilibrium impact
\begin{equation*}I_t=-N\lambda\left[\rho_{-}+\frac{1}{z_1+\beta}e^{z_1t} -\frac{\gamma_1}{ \gamma_2}\frac{1}{z_2+\beta}e^{z_2t}\right]\frac{\varphi}{\varepsilon\psi}\overline{x}.
\end{equation*}
By differentiating we similarly obtain
\[\dot{X}^{*,i}_t= -\left[\beta\rho_{-}+e^{z_1 t}- \frac{\gamma_1}{\gamma_2} e^{z_2 t}\right]\frac{\varphi}{\varepsilon \psi}\overline{x}-\left[\beta+\frac{\lambda e^{z_{3} t}}{\varepsilon e^{z_3T}}\right]\frac{\varphi}{\varepsilon \xi}(x^{i}-\overline{x}), \ \ \ i=1,\dots,N.\]
Here we have used the constants $\psi$ and $\xi$ from Table \ref{tab:constants}.
Defining
\[h^1_t := \left[\beta\rho_{-}+e^{z_1 t}- \frac{\gamma_1}{\gamma_2} e^{z_2 t}\right]\left[\rho_{-}+\frac{1}{z_1+\beta}e^{z_1t} -\frac{\gamma_1}{ \gamma_2}\frac{1}{z_2+\beta}e^{z_2t}\right],\]
\[h^2_t := \left[\beta+\frac{\lambda e^{z_{3} t}}{\varepsilon e^{z_3T}}\right]\left[\rho_{-}+\frac{1}{z_1+\beta}e^{z_1t} -\frac{\gamma_1}{ \gamma_2}\frac{1}{z_2+\beta}e^{z_2t}\right],\]
one can verify that
\[I_t\dot{X}^{*,i}_t = h^1_t\frac{N\lambda\varphi^2}{\varepsilon^2\psi^2}\overline{x}^2+h^2_t\frac{N\lambda\varphi^2}{\varepsilon^2\xi\psi}(x^{i}-\overline{x})\overline{x}.\]
At the same time, defining
\[h^3_t:=\left[\beta\rho_{-}+e^{z_1 t}- \frac{\gamma_1}{\gamma_2} e^{z_2 t}\right]^2, \ \ \ h^4_t:=\left[\beta+\frac{\lambda e^{z_{3} t}}{\varepsilon e^{z_3T}}\right]^2,\]
\[h^5_t:=\left[\beta\rho_{-}+e^{z_1 t}- \frac{\gamma_1}{\gamma_2} e^{z_2 t}\right]\left[\beta+\frac{\lambda e^{z_{3} t}}{\varepsilon e^{z_3T}}\right],\]
we have
\[(\dot{X}^{*,i}_t)^2=h^3_t\frac{\varphi^2}{\varepsilon^2 \psi^2}\overline{x}^2+h^4_t\frac{\varphi^2}{\varepsilon^2 \xi^2}(x^{i}-\overline{x})^2+h^5_t\frac{2\varphi^2}{\varepsilon^2 \xi\psi}\overline{x}(x^{i}-\overline{x}).\]
By expanding the product form of the functions $h^i_t$, $i=1,\dots,5$ and integrating over $[0,T]$ we obtain the identities
\[\mathfrak{h}_i=\int_0^T h_t^i dt, \ \ \ i=1,\dots, 5,\]
for the constants $\mathfrak{h}_i$ defined in Table \ref{tab:constants}. The form of the integrals
\[ \int_0^T I_t dX_t^{*,i} = \int_0^TI_t\dot{X}^{*,i}_tdt \ \ \ \text{and} \ \ \ \frac{\varepsilon}{2}\int_0^T(\dot{X}_t^{*,i})^2dt\]
in the statement of the corollary follows by substitution. Finally, an algebraic manipulation allows us to write $X_T^{*,i}$ in terms of the constant $\mathfrak{p}$ of Table \ref{tab:constants},
\begin{equation*}
    X_T^{*,i}=\frac{ z_3}{ \xi}(x^i-\overline{x})+\frac{\mathfrak{p}}{\psi}\overline{x}.
\end{equation*}
The reported form of $\frac{\varphi}{2}(X_T^{*,i})^2$ is then immediate. \qed

\subsection{Equilibrium with Liquidation Constraint}
\subsubsection{Lemma \ref{lem:conv.equilibria}}\label{se:proof.of.conv.equilibria}
To emphasize the dependence of the equilibria in Theorem \ref{thm:equil.term.pen} on $\varphi$ we will write $\boldsymbol{X}^{*}(\varphi)$ and $\dot{\boldsymbol{X}}^*(\varphi)=\boldsymbol{v}^*(\varphi)$. Under our assumptions there is $\tilde{\boldsymbol{X}}\in H^1[0,T]^{\times N}$ such that $\boldsymbol{X}^{*}(\varphi)\xrightarrow[\varphi\to\infty]{H^{1^{\times N}}} \tilde{\boldsymbol{X}}$. As a consequence, there is also $\tilde{\boldsymbol{v}}\in L^2[0,T]^{\times N}$ such that $\dot{\tilde{\boldsymbol{X}}}=\tilde{\boldsymbol{v}}$ and $\boldsymbol{v}^{*}(\varphi)\xrightarrow[\varphi\to\infty]{L^{2^{\times N}}} \tilde{\boldsymbol{v}}$. We begin with a technical result about the $\Gamma$-convergence of the objective functions when indexed by $\varphi$ and $\boldsymbol{v}^{*,-i}(\varphi)$. See Appendix \ref{app:opt.theory} for the definition of $\Gamma$-convergence.

\begin{lemma}\label{lem:gamma.conv.J}
We have $\mathcal{J}_{A}(\cdot;\boldsymbol{v}^{*,-i}(\varphi))\xrightarrow[\varphi\to\infty]{\Gamma} \mathcal{J}_{A'}(\cdot;\tilde{\boldsymbol{v}}^{-i})$ for $i=1,\dots,N$.    
\end{lemma}

\begin{proof}
We begin by considering $I$ as a functional that takes $L^2[0,T]$ to itself. That is, for any $v\in L^2[0,T]$ we write $I[v;\boldsymbol{v}^{*,-i}(\varphi)]$ to denote the function
\[t\mapsto I_t[v;\boldsymbol{v}^{*,-i}(\varphi)]= \lambda\int_0^t e^{-\beta(t-s)}\left(v_s+\sum_{j\not=i} v^{*,j}_s(\varphi)\right)ds.\]
Standard estimates verify that $I[v;\boldsymbol{v}^{*,-i}(\varphi)]\in L^2[0,T]$. We also consider $X_T^i$ as a functional from $L^2[0,T]$ to $\mathbb{R}$,
\[v\mapsto X_T^i[v] = x^i+\int_0^Tv_tdt.\] With this, we can express the objectives $\cJ_A$ and $\cJ_{A'}$ of \eqref{eqn:cJ.A} and \eqref{eqn:cJ.A'} as
\[\mathcal{J}_{A}(v;\boldsymbol{v}^{*,-i}(\varphi))=\left\langle I[v;\boldsymbol{v}^{*,-i}(\varphi)], v\right\rangle_{L^2}+\frac{\varepsilon}{2}\|v\|_{L^2}^2 + \frac{\varphi}{2}\left(X_T^i[v]\right)^2,\]
\[\mathcal{J}_{A'}(v;\tilde{\boldsymbol{v}}^{-i})=\left\langle I[v;\tilde{\boldsymbol{v}}^{-i}], v\right\rangle_{L^2}+\frac{\varepsilon}{2}\|v\|_{L^2}^2 + \chi_{\left\{X_T^i[v]\not=0\right\}}.\]

Next, fix a sequence of positive numbers $(\varphi^n)_{n\geq0}$ satisfying $\varphi^n\to\infty$. Let $v\in L^2[0,T]$ be arbitrary and fix any convergent sequence $(v^n)_{n\geq0}$ in $L^2[0,T]$ with limit $v$. It is easy to verify that
$I[v^n;\boldsymbol{v}^{*,-i}(\varphi^n)]\xrightarrow[n\to\infty]{L^2}I[v;\tilde{\boldsymbol{v}}^{-i}]$. Using this and the continuity of the inner product,
\begin{equation}\label{eqn:cont.conv.inner.prod}
    \lim_{n\to\infty}\left(\left\langle I[v^n;\boldsymbol{v}^{*,-i}(\varphi^n)], v^n\right\rangle_{L^2}+\frac{\varepsilon}{2}\|v^n\|_{L^2}^2\right)=\left\langle I[v;\tilde{\boldsymbol{v}}^{-i}], v\right\rangle_{L^2}+\frac{\varepsilon}{2}\|v\|_{L^2}^2 .
\end{equation}
It is similarly straightforward to check that
\begin{equation}\label{eqn:gamma.conv.pen}
    \frac{\varphi^n}{2}\left(X_T^i[v]\right)^2\xrightarrow[n\to\infty]{\Gamma} \chi_{\{X_T^i[v]\not=0\}}
\end{equation}
as a functional on $L^2[0,T]$. Combining the conclusions of \eqref{eqn:cont.conv.inner.prod} and \eqref{eqn:gamma.conv.pen} with \cite[Proposition 6.20]{DalMaso.93} gives the $\Gamma$-convergence stated in the lemma. 
\end{proof}

As $v^{*,i}(\varphi)$ minimizes $\mathcal{J}_{A}(\cdot;\boldsymbol{v}^{*,-i}(\varphi))$ and $v^{*,i}(\varphi)\xrightarrow[\varphi\to\infty]{L^2} \tilde{v}^i$, Lemma \ref{lem:gamma.conv.J} and Theorem \ref{thm:fund.thm.gamma.conv} yield that $\tilde{v}^i$ minimizes $\mathcal{J}_{A'}(\cdot;\tilde{\boldsymbol{v}}^{-i})$ and the costs converge. As this holds for all $i$, $\tilde{\boldsymbol{v}}$ is a Nash equilibrium for $\mathcal{J}_{A'}$. By reparametrizing in terms of the original process~$\tilde{\boldsymbol{X}}$, it follows that~$\tilde{\boldsymbol{X}}$ is a Nash equilibrium for $J_{A'}$. \qed

\subsubsection{Theorem \ref{thm:equil.liq.constr}}\label{app:pf.equil.liq.constr}
Passing to the limit in $f_t$ and $g_t$ from Theorem \ref{thm:equil.term.pen} gives the stated form of $\mathfrak{f}_t$ and $\mathfrak{g}_t$. To see that the equilibria converge in $H^1[0,T]^{\times N}$ (so that we may apply Lemma \ref{lem:conv.equilibria}), we first verify that $f$ and $g$ converge uniformly to $\mathfrak{f}$ and $\mathfrak{g}$ on $[0,T]$. Indeed, the functions take the form
\[ f_t = 1-\frac{\mathsf{F}_t\varphi}{\varepsilon z_3+\varphi\Xi}, \ \ \ g_t = 1-\frac{\mathsf{G}_t\varphi}{\varepsilon\mathfrak{p}+\varphi\Psi}, \ \ \ \mathfrak{f}_t = 1-\frac{\mathsf{F}_t}{\Xi}, \ \ \ \mathfrak{g}_t = 1-\frac{\mathsf{G}_t}{\Psi},\]
for given continuous $\mathsf{F}$ and $\mathsf{G}$, and constants $z_3,\mathfrak{p},\Psi$ and $\Xi$ that are all independent of $\varphi$.  It follows that
\[\sup_{t\in[0,T]}|f_t-\mathfrak{f}_t|\leq \sup_{t\in[0,T]}|\mathsf{F}_t|\left|\frac{\varepsilon z_3}{\varepsilon z_3\Xi+\varphi\Xi^2}\right|, \ \ \ \sup_{t\in[0,T]}|g_t-\mathfrak{g}_t|\leq \sup_{t\in[0,T]}|\mathsf{G}_t|\left|\frac{\varepsilon\mathfrak{p}}{\varepsilon\mathfrak{p}\Psi+\varphi\Psi^2}\right|.\]
This gives the requisite uniform convergence as $\varphi\uparrow\infty$. In fact, by a similar argument it is not hard to see that the derivatives  \textit{to all orders} of $f$ and $g$ also converge uniformly. Moreover, these estimates ensure that, when parametrized by $\varphi>0$, the equilibria (and their derivatives) in Theorem \ref{thm:equil.term.pen} are uniformly bounded in the supremum norm\footnote{i.e., there exists $C>0$ such that for all $\varphi$, the equilibrium $X^{*,i}(\varphi)$ corresponding to $\varphi$ satisfies $\sup_{t\in[0,T]}|X^{*,i}_t(\varphi)|\leq C$, and an analogous statement holds for the derivatives $\dot{X}^{*,i}(\varphi)$.}. Taken together, the dominated convergence theorem implies the claimed convergence of the equilibrium strategy profile.

\subsubsection{Corollary \ref{cor:eq.cost.liq.constr}}
This result could be shown directly by substituting the equilibrium strategies in Theorem~\ref{thm:equil.liq.constr}. However, to simplify computations we leverage Lemma \ref{lem:conv.equilibria} and pass to the limit in the expressions of Corollary \ref{cor:eq.cost.term.pen}. The constants $\mathfrak{p}$ and $\mathfrak{h}_j$, $j=1,\dots,5$ appearing in Corollary~\ref{cor:eq.cost.term.pen} are independent of $\varphi$ (see Table \ref{tab:constants}). Consequently, we may focus on the behavior of~$\psi$ and~$\xi$.
From Table \ref{tab:constants} we see that $\varphi^{-1}\psi\to\varepsilon^{-1}\Psi$ and $\varphi^{-1}\xi\to \varepsilon^{-1}\Xi$ as $\varphi\uparrow\infty$. 
With this, the expression for the cost follows immediately.
\qed

\section{Proofs for Section \ref{sec:N.block.cost}}\label{app:N.block.cost}

\subsection{Proposition \ref{prop:no.interior.jumps}}

We first show an auxiliary result related to the optimality of jumps; it follows the template of \cite[Proposition 4.11]{SchiedStrehleZhang.17}.

\begin{lemma}\label{lem:gen.foc.jumps}
     Fix a strategy profile $\boldsymbol{X}$. If $X^{i}$ is optimal for $J_B(\cdot,\boldsymbol{X}^{-i})$ then for any $[0,T]$-valued predictable time $\tau$ there exists an $\mathcal{F}_{\tau-}$-measurable random variable $\boldsymbol{\Upsilon}$ satisfying
     \begin{equation*}
         \mathbb{E}\left[\lambda\int_0^Te^{-\beta |\sigma-t|}dX_t^{i}+\lambda\int_0^{\sigma-} e^{-\beta (\sigma-t)}\sum_{j\not=i}dX_t^{j} + \vartheta_\sigma\Delta X_\sigma^{i}+\frac{\lambda}{2}\sum_{j\not=i}\Delta X_\sigma^{j}\bigg|\mathcal{F}_{\tau-}\right]=\boldsymbol{\Upsilon} \ \ \ a.s.
     \end{equation*}
     for all predictable times $\sigma$ satisfying $\tau\leq \sigma \leq T$. In particular, we may take
     \[\boldsymbol{\Upsilon} = \mathbb{E}\left[\lambda\int_0^Te^{-\beta |\tau-t|}dX_t^{i}+\lambda\int_0^{\tau-} e^{-\beta (\tau-t)}\sum_{j\not=i}dX_t^{j} + \vartheta_\tau\Delta X_\tau^{i}+\frac{\lambda}{2}\sum_{j\not=i}\Delta X_\tau^{j}\bigg|\mathcal{F}_{\tau-}\right].\]

\end{lemma}

\begin{proof}
    Let $\tau$ and $\sigma\geq \tau$ be arbitrary $[0,T]$-valued predictable stopping times. Fixing some $A\in\mathcal{F}_{\tau-}$ we can consider the round trip trade $Z$ defined by
    \[Z_t=\mathds{1}_{A}\left(\mathds{1}_{t\geq \tau}-\mathds{1}_{t\geq\sigma}\right). \]
    If we perturb $X^{i}$ by $\alpha Z$ ($\alpha\in\mathbb{R}$) we get from Proposition \ref{prop:obj.func.rep} the cost
    \begin{align*}
    &J_B(X^{i}+\alpha Z;\boldsymbol{X}^{-i})\\
    &\quad =\frac{1}{2}\mathbb{E}\left[\sum_{t\in[0,T]} \vartheta_t(\Delta X_t^{i}+\alpha \Delta Z_t)^2\right] +\lambda\mathbb{E}\bigg[\frac{1}{2}\int_0^T\int_0^T e^{-\beta |t-s|} d(X_s^{i} +\alpha Z_s)d(X_t^{i}+\alpha Z_t)\\
    &\quad \quad \quad  + \int_0^T\int_0^{t-} e^{-\beta (t-s)} \sum_{j\not=i}dX_s^{j} d(X_t^{i}+\alpha Z_t) + \frac{1}{2}\sum_{j\not=i}\sum_{t\in[0,T]} (\Delta X_t ^{i} +\alpha \Delta Z_t)\Delta X_t^{j}\bigg ].
    \end{align*}
    Differentiating this expression with respect to $\alpha$ and setting $\alpha=0$ yields
    \begin{align*}
        \frac{d}{d\alpha}J_B(X^{i}+\alpha Z;\boldsymbol{X}^{-i})|_{\alpha=0}&=\mathbb{E}\left[\sum_{t\in[0,T]} \vartheta_t\Delta X_t^{i}\Delta Z_t\right]+ \lambda\mathbb{E}\bigg[\int_0^T\int_0^T e^{-\beta |t-s|} dX_s^{i}dZ_t\\
        &\quad + \int_0^T\int_0^{t-} e^{-\beta (t-s)} \sum_{j\not=i}dX_s^{j} d Z_t + \frac{1}{2}\sum_{j\not=i}\sum_{t\in[0,T]} \Delta Z_t\Delta X_t^{j}\bigg ].
    \end{align*}
    Therefore, a necessary first-order condition for optimality is
    \begin{align*}0&=\mathbb{E}\left[\sum_{t\in[0,T]} \vartheta_t\Delta X_t^{i}\Delta Z_t\right]+ \lambda\mathbb{E}\bigg[\int_0^T\int_0^T e^{-\beta |t-s|} dX_s^{i}dZ_t\\
        &\quad + \int_0^T\int_0^{t-} e^{-\beta (t-s)} \sum_{j\not=i}dX_s^{j} d Z_t + \frac{1}{2}\sum_{j\not=i}\sum_{t\in[0,T]} \Delta Z_t\Delta X_t^{j}\bigg ].
    \end{align*}
    By substituting in the form of $Z$ we obtain
    \begin{align*}0&=\mathbb{E}\left[\mathds{1}_A \left(\vartheta_\tau\Delta X_\tau^{i}-\vartheta_\sigma\Delta X_\sigma^{i}\right)\right]+ \lambda\mathbb{E}\Bigg[\mathds{1}_A\bigg(\int_0^T(e^{-\beta |\tau-t|} -e^{-\beta |\sigma-t|})dX_t^{i}\\
        &\quad + \int_0^{\tau-} e^{-\beta (\tau-t)}\sum_{j\not=i}dX_t^{j} -\int_0^{\sigma-} e^{-\beta(\sigma-t)}\sum_{j\not=i}dX_t^{j} + \frac{1}{2}\sum_{j\not=i}(\Delta X_\tau^{j}-\Delta X_\sigma^{j})\bigg)\Bigg ],
    \end{align*}
    which becomes, after rearranging, 
\begin{align*}
    \mathbb{E}&\left[\mathds{1}_A \left(\vartheta_\tau\Delta X_\tau^{i}+\lambda\int_0^Te^{-\beta |\tau-t|}dX_t^{i}+\lambda\int_0^{\tau-} e^{-\beta (\tau-t)}\sum_{j\not=i}dX_t^{j}+\frac{\lambda}{2}\sum_{j\not=i}\Delta X_\tau^{j}\right)\right]\\
    &=\mathbb{E}\left[\mathds{1}_A \left(\vartheta_\sigma\Delta X_\sigma^{i}+\lambda\int_0^Te^{-\beta |\sigma-t|}dX_t^{i}+\lambda\int_0^{\sigma-} e^{-\beta (\sigma-t)}\sum_{j\not=i}dX_t^{j}+\frac{\lambda}{2}\sum_{j\not=i}\Delta X_\sigma^{j}\right)\right].
\end{align*}
Iterative conditioning on $\mathcal{F}_{\tau-}$ under the expectation completes the proof after noting that $A\in\mathcal{F}_{\tau-}$ was arbitrary.
\end{proof}

For the proof of Lemma~\ref{lem:main.step.for.no.interior.jumps}  below it will be helpful to approximate predictable times with times where the inventories do not jump. That is feasible according to the next lemma, which is a straightforward generalization of \cite[Lemma B.1]{SchiedStrehleZhang.17}. We omit the proof.

\begin{lemma}\label{lem:no.jump.delay.times}
    Let $\boldsymbol{X}$ be an admissible strategy profile and let $\tau,\sigma$ be predictable stopping times satisfying $\tau\leq \sigma\leq T$. There exists a sequence $(\tau_n)_{n\geq0}$ of predictable times satisfying
    \begin{enumerate}
        \item[(i)] $\tau\leq \tau_n\leq \sigma$ with $\tau_n\downarrow \tau$ a.s.,
        \item[(ii)] $\tau<\tau_n$ on $\{\tau<\sigma\}$,
        \item[(iii)] $\Delta X_{\tau_n}^i=0$ on $\{\tau_n<\sigma\}$ for $i=1,\dots,N$.
    \end{enumerate}
    In particular, $\lim_{n\to\infty} \Delta X_{\tau_n}^i = 0$ on $\{\tau<\sigma\}$ for all $i$.
\end{lemma}

With this we are ready to prove the following necessary conditions for the jumps of a best response to a strategy profile $\boldsymbol{X}^{-i}$.

\begin{lemma}\label{lem:main.step.for.no.interior.jumps} 
    Fix an admissible strategy profile $\boldsymbol{X}$ and suppose that $X^{i}$ is optimal for $J_B(\cdot,\boldsymbol{X}^{-i})$. If
    $\tau$ is any $[0,T]$-valued predictable stopping time, then
     \begin{equation}\label{eqn:eq.jumps.not.0}
         \left(\vartheta_\tau \Delta X_{\tau}^{i}+\frac{\lambda}{2}\sum_{j\not=i}\Delta X_{\tau}^{j}\right)\mathds{1}_{\{\tau>0\}}=0, \ \ \ a.s.
     \end{equation}
     and
     \begin{equation}\label{eqn:eq.jumps.not.T}
         \left( \vartheta_\tau \Delta X_{\tau}^{i}-\frac{\lambda}{2}\sum_{j\not=i}\Delta X_{\tau}^{j}\right)\mathds{1}_{\{\tau<T\}}=0, \ \ \ a.s.
     \end{equation}
\end{lemma}

\begin{proof}
    By predictability there exists an announcing sequence of stopping times $\tau_n\uparrow \tau$ with $\tau_n<\tau$ on $\{\tau>0\}$ and $\tau_n= 0$ on $\{\tau=0\}$. Moreover, we can take these times to be predictable (see \cite[Corollary 2.1]{Siorpaes.14}). At the same time, for each $\tau_n$ we can apply Lemma~\ref{lem:no.jump.delay.times} to find a sequence $(\tau_{n,m})_{m\geq0}$ satisfying (i) $\tau_n\leq \tau_{n,m}\leq \tau$ with $\tau_{n,m}\downarrow \tau_n$ as $m\uparrow\infty$, (ii) $\tau_n<\tau_{n,m}$ on $\{\tau_n<\tau\}$, and (iii) $\Delta X_{\tau_{n,m}}^{i} =0$ for all $i$ on $\{\tau_{n,m}<\tau\}$.
    
    Applying Lemma \ref{lem:gen.foc.jumps} twice and using the absence of jumps at $\tau_{n,m}$ on $\{\tau_{n,m}<\tau\}$,
    \begin{align*}
    \mathbb{E}&\Bigg[\lambda\int_0^Te^{-\beta |\tau_{n,m}-t|}dX_t^{i}+\lambda\int_0^{\tau_{n,m}-} e^{-\beta (\tau_{n,m}-t)}\sum_{j\not=i}dX_t^{j}\\
    &\quad \quad \quad  +\left( \vartheta_{\tau_{n,m}} \Delta X_{\tau_{n,m}}^{i}+\frac{\lambda}{2}\sum_{j\not=i}\Delta X_{\tau_{n,m}}^{j}\right)\mathds{1}_{\{\tau_{n,m}=\tau\}}\bigg|\mathcal{F}_{\tau_n-}\Bigg]\\
    &=\mathbb{E}\left[\lambda\int_0^Te^{-\beta |\tau_n-t|}dX_t^{i}+\lambda\int_0^{\tau_n-} e^{-\beta (\tau_n-t)}\sum_{j\not=i}dX_t^{j} + \vartheta_{\tau_n}\Delta X_{\tau_n}^{i}+\frac{\lambda}{2}\sum_{j\not=i}\Delta X_{\tau_n}^{j}\bigg|\mathcal{F}_{\tau_n-}\right]\\
    &=\mathbb{E}\left[\lambda\int_0^Te^{-\beta |\tau-t|}dX_t^{i}+\lambda\int_0^{\tau-} e^{-\beta (\tau-t)}\sum_{j\not=i}dX_t^{j} + \vartheta_{\tau}\Delta X_{\tau}^{i}+\frac{\lambda}{2}\sum_{j\not=i}\Delta X_{\tau}^{j}\bigg|\mathcal{F}_{\tau_n-}\right].
    \end{align*}
    We focus on the first and last expectation in this chain of equalities. Taking $m\uparrow \infty$ we get, by the dominated convergence theorem,
    \begin{align*}
    \mathbb{E}&\Bigg[\lambda\int_0^Te^{-\beta |\tau_n-t|}dX_t^{i}+\mathds{1}_{\{\tau_n<\tau\}}\lambda\int_0^{\tau_n} e^{-\beta (\tau_n-t)}\sum_{j\not=i}dX_t^{j}+\mathds{1}_{\{\tau_n=\tau\}}\lambda\int_0^{\tau-} e^{-\beta (\tau-t)}\sum_{j\not=i}dX_t^{j}\\
    &\quad \quad \quad  +\left( \vartheta_\tau \Delta X_{\tau}^{i}+\frac{\lambda}{2}\sum_{j\not=i}\Delta X_{\tau}^{j}\right)\mathds{1}_{\{\tau_n=\tau\}}\bigg|\mathcal{F}_{\tau_n-}\Bigg]\\
    &=\mathbb{E}\left[\lambda\int_0^Te^{-\beta |\tau-t|}dX_t^{i}+\lambda\int_0^{\tau-} e^{-\beta (\tau-t)}\sum_{j\not=i}dX_t^{j} + \vartheta_\tau \Delta X_{\tau}^{i}+\frac{\lambda}{2}\sum_{j\not=i}\Delta X_{\tau}^{j}\bigg|\mathcal{F}_{\tau_n-}\right].
    \end{align*}
    Here, we have used that $\lim_{m\to\infty}\Delta X^{i}_{\tau_{n,m}}=0$ on $\{\tau_n<\tau\}$ for all $i$  and $\Delta X^{i}_{\tau_{n,m}}=\Delta X^{i}_{\tau}$ on $\{\tau_n=\tau\}$ for all $i$. By definition of the announcing sequence we have $\{\tau_n=\tau\}=\{\tau=0\}$, so
    \begin{align*}
    \mathbb{E}&\Bigg[\lambda\int_0^Te^{-\beta |\tau_n-t|}dX_t^{i}+\mathds{1}_{\{\tau>0\}}\lambda\int_0^{\tau_n} e^{-\beta (\tau_n-t)}\sum_{j\not=i}dX_t^{j}+\left( \vartheta_0\Delta X_{0}^{i}+\frac{\lambda}{2}\sum_{j\not=i}\Delta X_{0}^{j}\right)\mathds{1}_{\{\tau=0\}}\bigg|\mathcal{F}_{\tau_n-}\Bigg]\\
    &=\mathbb{E}\left[\lambda\int_0^Te^{-\beta |\tau-t|}dX_t^{i}+\lambda\int_0^{\tau-} e^{-\beta (\tau-t)}\sum_{j\not=i}dX_t^{j} + \vartheta_\tau \Delta X_{\tau}^{i}+\frac{\lambda}{2}\sum_{j\not=i}\Delta X_{\tau}^{j}\bigg|\mathcal{F}_{\tau_n-}\right].
    \end{align*}
    Then, taking $n\uparrow \infty$ and applying the dominated convergence theorem for conditional expectations \cite[Theorem 4.6.10]{Durrett.19},
    \begin{align*}
    \mathbb{E}&\Bigg[\lambda\int_0^Te^{-\beta |\tau-t|}dX_t^{i}+\mathds{1}_{\{\tau>0\}}\lambda\int_0^{\tau-} e^{-\beta (\tau-t)}\sum_{j\not=i}dX_t^{j}+\left( \vartheta_0\Delta X_{0}^{i}+\frac{\lambda}{2}\sum_{j\not=i}\Delta X_{0}^{j}\right)\mathds{1}_{\{\tau=0\}}\bigg|\mathcal{F}_{\tau-}\Bigg]\\
    &=\mathbb{E}\left[\lambda\int_0^Te^{-\beta |\tau-t|}dX_t^{i}+\lambda\int_0^{\tau-} e^{-\beta (\tau-t)}\sum_{j\not=i}dX_t^{j} + \vartheta_\tau \Delta X_{\tau}^{i}+\frac{\lambda}{2}\sum_{j\not=i}\Delta X_{\tau}^{j}\bigg|\mathcal{F}_{\tau-}\right].
    \end{align*}
    Rearranging shows that
    \[\mathbb{E}\left[\left(\vartheta_\tau\Delta X_{\tau}^{i}+\frac{\lambda}{2}\sum_{j\not=i}\Delta X_{\tau}^{j}\right)\mathds{1}_{\{\tau>0\}}\bigg|\mathcal{F}_{\tau-}\right]=0.\]
    In view of the predictability of $\boldsymbol{X}$  we recover \eqref{eqn:eq.jumps.not.0}.

    On the other hand, we can once again choose a sequence of stopping times $\tau_n\downarrow \tau$ satisfying the conditions of Lemma \ref{lem:no.jump.delay.times} with the choice of $\sigma\equiv T$. Applying Lemma \ref{lem:gen.foc.jumps} we have
    \begin{align*}
    \mathbb{E}&\left[\lambda\int_0^Te^{-\beta |\tau_n-t|}dX_t^{i}+\lambda\int_0^{\tau_n-} e^{-\beta (\tau_n-t)}\sum_{j\not=i}dX_t^{j} +\left( \vartheta_T\Delta X_{T}^{i}+\frac{\lambda}{2}\sum_{j\not=i}\Delta X_{T}^{j}\right)\mathds{1}_{\{\tau=T\}}\bigg|\mathcal{F}_{\tau-}\right]\\
    &=\mathbb{E}\left[\lambda\int_0^Te^{-\beta |\tau-t|}dX_t^{i}+\lambda\int_0^{\tau-} e^{-\beta (\tau-t)}\sum_{j\not=i}dX_t^{j} + \vartheta_\tau \Delta X_{\tau}^{i}+\frac{\lambda}{2}\sum_{j\not=i}\Delta X_{\tau}^{j}\bigg|\mathcal{F}_{\tau-}\right].
    \end{align*}
    Passing to the limit, by the dominated convergence theorem,
    \begin{align*}
    \mathbb{E}&\Bigg[\lambda\int_0^Te^{-\beta |\tau-t|}dX_t^{i}+\mathds{1}_{\tau<T}\lambda\int_0^{\tau} e^{-\beta (\tau-t)}\sum_{j\not=i}dX_t^{j} \\
    &\quad \quad +\left( \lambda\int_0^{T-} e^{-\beta (T-t)}\sum_{j\not=i}dX_t^{j} +\vartheta_T\Delta X_{T}^{i}+\frac{\lambda}{2}\sum_{j\not=i}\Delta X_{T}^{j}\right)\mathds{1}_{\{\tau=T\}}\bigg|\mathcal{F}_{\tau-}\Bigg]\\
    &\quad \quad \quad =\mathbb{E}\left[\lambda\int_0^Te^{-\beta |\tau-t|}dX_t^{i}+\lambda\int_0^{\tau-} e^{-\beta (\tau-t)}\sum_{j\not=i}dX_t^{j} + \vartheta_\tau \Delta X_{\tau}^{i}+\frac{\lambda}{2}\sum_{j\not=i}\Delta X_{\tau}^{j}\bigg|\mathcal{F}_{\tau-}\right].
    \end{align*}
    Subtracting the left-hand side from the right we obtain
    \[\mathbb{E}\left[\left( \vartheta_\tau \Delta X_{\tau}^{i}-\frac{\lambda}{2}\sum_{j\not=i}\Delta X_{\tau}^{j}\right)\mathds{1}_{\{\tau<T\}}\bigg|\mathcal{F}_{\tau-}\right]=0.\]
    Using predictability once more yields the condition \eqref{eqn:eq.jumps.not.T}.
\end{proof}

Proposition \ref{prop:no.interior.jumps} is a direct consequence of the following, more precise result.

\begin{proposition}\label{prop:admissible.jumps.extended}
    Fix an admissible strategy profile $\boldsymbol{X}$ and suppose $X^i$ is optimal for $J_B(\cdot,\boldsymbol{X}^{-i})$.
    \begin{enumerate}
    \item[(i)] The initial and terminal jumps of $X^i$ satisfy
    \[\vartheta_0\Delta X_{0}^{i}=\frac{\lambda}{2}\sum_{j\not=i}\Delta X_{0}^{j}, \ \ \ \vartheta_T\Delta X_{T}^{i}=-\frac{\lambda}{2}\sum_{j\not=i}\Delta X_{T}^{j}.\]
    \item[(ii)] If $\vartheta_t>0$ at $t\in(0,T)$ then $X^i$ has no interior jump at $t$,
    \[\Delta X_{t}^{i} = 0, \ \ \ a.s.\]
    \item[(iii)] If there exists $t\in[0,T]$ such that $\vartheta_t=0$ and $\mathbb{P}\left(\sum_{j\not=i}\Delta X_t^j \not=0\right)>0$, then no optimal strategy $X^i$ exists.
    \end{enumerate}
    If $\boldsymbol{X}$ is a Nash Equilibrium, then there are no interior jumps, irrespective of~$\vartheta$. Furthermore, if $\vartheta_0=0$ or $\vartheta_T=0$, then there are no jumps at $0$ or $T$ in equilibrium, respectively.
\end{proposition}

\begin{proof}
(i) Taking $\tau\equiv T$ in \eqref{eqn:eq.jumps.not.0} and $\tau\equiv0$ in \eqref{eqn:eq.jumps.not.T} gives the two equations in the proposition. (ii) Let $\vartheta_t>0$ for some $t\in(0,T)$. Then we can take $\tau\equiv t$ and sum across \eqref{eqn:eq.jumps.not.0} and \eqref{eqn:eq.jumps.not.T} to get
\[2\vartheta_t\Delta X_{t}^{i}=0, \ \ \ \text{a.s.},\]
which implies the result.
(iii) Fix $t$ such that $\vartheta_t=0$ and $\mathbb{P}\left(\sum_{j\not=i}\Delta X_t^j \not=0\right)>0$. By Lemma~\ref{lem:main.step.for.no.interior.jumps} if $X^{i}$ is admissible and optimal for $\boldsymbol{X}^{-i}$ then one (or both) of \eqref{eqn:eq.jumps.not.0} and \eqref{eqn:eq.jumps.not.T} hold at $t$. In this case these read
    \begin{equation}\label{eqn:eqns.when.theta.is.0} \frac{\lambda}{2}\sum_{j\not=i}\Delta X_{t}^{j}\mathds{1}_{\{t>0\}}=0 \ \ \ \mbox{a.s.} \ \ \ \text{and} \ \ \ -\frac{\lambda}{2}\sum_{j\not=i}\Delta X_{t}^{j}\mathds{1}_{\{t<T\}}=0 \ \ \ \mbox{a.s.}
    \end{equation}
    which, as $\lambda>0$, implies the contradiction that $\sum_{j\not=i}\Delta X_{t}^{j}=0$ a.s. 

To characterize the jumps in equilibrium we leverage the concurrent satisfaction of \eqref{eqn:eq.jumps.not.0} and \eqref{eqn:eq.jumps.not.T} for all $i$. If $\vartheta_t>0$ for $t\in(0,T)$, the absence of interior jumps at $t$ follows from~(ii). If $\vartheta_t=0$ for $t\in[0,T]$ (note the inclusion of $0$ and $T$) then in equilibrium we can sum over $i=1,\dots,N$ in either \eqref{eqn:eq.jumps.not.0} or \eqref{eqn:eq.jumps.not.T} to get
\[\frac{\lambda(N-1)}{2}\sum_{i=1}^N\Delta X_t^{*,i}=0.\]
Again, as $N\geq2$ and $\lambda>0$ this implies $\sum_{i=1}^N\Delta X_t^{*,i}=0$ which, when combined with \eqref{eqn:eq.jumps.not.0} or \eqref{eqn:eq.jumps.not.T} holding at $t$ for all $i$ (see \eqref{eqn:eqns.when.theta.is.0}), implies that $\Delta X_t^{*,i}=0$ for all $i$. This completes the proof. 
\end{proof}

\subsection{Lemma \ref{lem:nash.eq.a.c.}}
      Let $\boldsymbol{X}^*=(X^{*,1},\dots,X^{*,N})$ be a Nash equilibrium in the class of deterministic strategies that are absolutely continuous on $(0,T)$. 
      By Lemma \ref{lem:nash.eq.det} it suffices to show that $\boldsymbol{X}^*$ is also a  Nash equilibrium in the class of deterministic strategies. Suppose for contradiction that $\boldsymbol{X}^*$ is not a Nash equilibrium in the class of deterministic strategies. Then, there exists an $i\in\{1,\dots,N\}$ and a deterministic admissible strategy $Z=(Z_t)_{t\geq0}$ for trader $i$ satisfying
     \[J_B(Z;\boldsymbol{X}^{*,-i})<J_B(X^{*,i};\boldsymbol{X}^{*,-i}). \]
     By definition of $C_B(
    \cdot)$ (see \eqref{eqn:CB}) it is safe to assume that $Z_T=0$. 
    
    We will approximate $Z$ by a sequence of admissible controls that are absolutely continuous on $(0,T)$.  For concreteness, we will use the Bernstein polynomials $\mathbb{B}_n(f)$ defined for functions $f$ on $[0,1]$ via the Bernstein operator $\mathbb{B}_n$, $n= 1, 2, \ldots$,
    \[
    \mathbb{B}_n(f)(t)=\sum_{k=0}^nf\left(\frac{k}{n}\right)\binom{n}{k}t^k(1-t)^{n-k}, \quad t\in[0,1].
    \]
    We extend the definition to functions on $[0,T]$ via the isomorphism $\iota(t) = t/T$ and abuse notation by still writing $\mathbb{B}_n(f)$ for this approximation. Define an auxiliary process $\tilde{Z}$ by $\tilde{Z}_t=Z_t$ for $t\in[0,T)$ and $\tilde{Z}_{T}=Z_{T-}$. We define our approximating sequence $(Z^n)_{n\geq1}$ by $Z^n_t:=\mathbb{B}_n(\tilde{Z})(t)$ on $[0,T)$ with $Z_{0-}^n=Z_{0-}=x^i$ and $Z_{T}^n=Z_{T}=0$.
    
    We collect here several critical properties of the approximating sequence.
    \begin{enumerate}
        \item[(a)] (Matching Endpoints) By definition of $\mathbb{B}_n$ and $Z^n$, there is no approximation error at the endpoints $\{0-,0,T-,T\}$.
        \item[(b)] (Smoothing)  Let $TV(\cdot;[a,b])$ denote the total variation of a function on $[a,b]$. The smoothing property (see \cite{lorentz1937theorie} or \cite[Proposition 4.16]{bustamante2017bernstein}) of $\mathbb{B}_n$ gives $TV(Z^n;[0,T))\leq TV(Z;[0,T))$. By additivity of the total variation on intervals and (a), we can extend this to include any jumps at $0$ and $T$. That is, $TV(Z^n;[0-,T])\leq TV(Z;[0-,T])$. 
        \item[(c)] (Uniformly Bounded) As $TV(Z^n;[0-,T])$ is uniformly bounded by (b) and $Z^n_{0-}=x^i$ for all $n$, we have that $\|Z^n\|_\infty$ and $\|Z^n-Z\|_\infty$ are uniformly bounded.
        \item[(d)] (Consistency) Since $Z$ is c\`adl\`ag, it only has discontinuities of the first kind. As a result (see, e.g., %
        \cite[Section~4.5.1]{bustamante2017bernstein})
        \[Z^n_t\to \frac{Z_{t-}+Z_{t+}}{2}, \ \ \ \forall t\in(0,T).\]
        Since the discontinuities of $Z$ are at most countable, $Z^n$ converges to $Z$ almost everywhere. Moreover, by (c) and the dominated convergence theorem, $Z^n\to Z$ in~$L^1$.
    \end{enumerate}
    In view of (a), it is clear that $Z^n$ is an admissible deterministic strategy. Moreover, by the definition of $\mathbb{B}_n$ it is absolutely continuous on $(0,T)$. The objective representation in \eqref{eqn:obj.schied.rep} tells us that for $Z$,
    \begin{align*}J_B(Z;\boldsymbol{X}^{*,-i})&=\lambda\Bigg(\frac{1}{2}\int_0^T\int_0^T e^{-\beta |t-s|} dZ_s dZ_t + \int_0^T\int_0^{t-} e^{-\beta (t-s)} \sum_{j\not=i}dX_s^{*,j} dZ_t\\
    &\quad \quad \quad \quad + \frac{1}{2}\sum_{j\not=i}\sum_{t\in\{0,T\}} \Delta X_t^{*,j}\Delta Z_t \Bigg )+\frac{1}{2}\sum_{t\in[0,T]} \vartheta_t(\Delta Z_t)^2.
    \end{align*}
    Note that we have enforced in \eqref{eqn:obj.schied.rep} the assumption that $\Delta X_t^{*,j}=0$ for all $t\in(0,T)$ when $j\not=i$. For $Z^n$ we have
    \begin{align*}J_B(Z^n;\boldsymbol{X}^{*,-i})&=\lambda\Bigg(\frac{1}{2}\int_0^T\int_0^T e^{-\beta |t-s|} dZ^n_s dZ^n_t + \int_0^T\int_0^{t-} e^{-\beta (t-s)} \sum_{j\not=i}dX_s^{*,j} dZ^n_t\\
    &\quad \quad \quad \quad + \frac{1}{2}\sum_{j\not=i}\sum_{t\in\{0,T\}} \Delta X_t^{*,j}\Delta Z^n_t \Bigg )+\frac{1}{2}\sum_{t\in\{0,T\}} \vartheta_t(\Delta Z^n_t)^2\\
    &=\lambda\Bigg(\frac{1}{2}\int_0^T\int_0^T e^{-\beta |t-s|} dZ^n_s dZ^n_t + \int_0^T\int_0^{t-} e^{-\beta (t-s)} \sum_{j\not=i}dX_s^{*,j} dZ^n_t\\
    &\quad \quad \quad \quad + \frac{1}{2}\sum_{j\not=i}\sum_{t\in\{0,T\}} \Delta X_t^{*,j}\Delta Z_t \Bigg )+\frac{1}{2}\sum_{t\in\{0,T\}} \vartheta_t(\Delta Z_t)^2.
    \end{align*}
    In the second equality we have used (a) to identify the jumps at the endpoints with those of~$Z$.
    
    Next, we claim that
    \begin{align}
    \lim_{n\to\infty}&\left(J_B(Z^n;\boldsymbol{X}^{*,-i})-J_B(Z;\boldsymbol{X}^{*,-i})\right)\label{eqn:limit.obj.approx.seq}\\
    &=\lim_{n\to\infty}\Bigg[\lambda\bigg(\frac{1}{2}\int_0^T\int_0^T e^{-\beta |t-s|} dZ^n_s dZ^n_t -\frac{1}{2}\int_0^T\int_0^T e^{-\beta |t-s|} dZ_s dZ_t\nonumber\\
    &\quad + \int_0^T\int_0^{t-} e^{-\beta (t-s)} \sum_{j\not=i}dX_s^{*,j} dZ^n_t-\int_0^T\int_0^{t-} e^{-\beta (t-s)} \sum_{j\not=i}dX_s^{*,j} dZ_t\bigg)-\frac{1}{2}\sum_{t\in (0,T)} \vartheta_t(\Delta Z_t)^2\Bigg]\nonumber \\
    &=-\frac{1}{2}\sum_{t\in (0,T)} \vartheta_t(\Delta Z_t)^2\leq 0.\nonumber
    \end{align}
    We will prove this by treating each of the integrals in turn. Using Lebesgue--Stieltjes integration by parts and (a),
    \begin{align}
        \int_0^T e^{-\beta|t-s|}dZ^n_s &= e^{-\beta(T-t)} Z_T^n-e^{-\beta t}Z^n_{0-}-\beta \int_0^T \mathrm{sign}(t-s)Z_s^ne^{-\beta|t-s|}ds\nonumber\\
        &=-e^{-\beta t}x^i-\beta \int_0^T \mathrm{sign}(t-s)Z_s^ne^{-\beta|t-s|}ds, \ \ \ t\in[0,T]. \label{eqn:IBP.H}
    \end{align}
    Repeating this for $\int_0^T e^{-\beta|t-s|}dZ_s$ gives, using (d), that
    \begin{align}
        \left|\int_0^T e^{-\beta|t-s|}dZ^n_s-\int_0^T e^{-\beta|t-s|}dZ_s\right| &\leq \beta \int_0^T|Z_s^n-Z_s|e^{-\beta|t-s|}ds\leq \beta \|Z^n-Z\|_{L^1}\to 0. \label{eqn:unif.est.conv.H}
    \end{align}
    To keep notation compact define $H^n_t:=\int_0^T e^{-\beta|t-s|}dZ^n_s$ and $H_t:=\int_0^T e^{-\beta|t-s|}dZ_s$. By \eqref{eqn:IBP.H} and dominated convergence, it is clear that $H^n$ (and $H$) are continuous in $t$. Moreover, the uniform estimate in \eqref{eqn:unif.est.conv.H} shows that $\|H^n-H\|_\infty \to 0$. As a result, by applying (b),
    \begin{align*}
        \left|\int_0^T H^n_tdZ^n_t-\int_0^TH_tdZ^n_t\right|&\leq \|H^n-H\|_\infty TV(Z^n;[0-,T])\leq \|H^n-H\|_\infty TV(Z;[0-,T])\to 0.
    \end{align*}
    At the same time, by Fubini's theorem and the symmetry $|t-s|=|s-t|$,
    \begin{align*}
        \left|\int_0^T H_t dZ^n_t-\int_0^T H_t dZ_t \right|&= \left|\int_0^T \int_0^T e^{-\beta|t-s|} dZ_s dZ^n_t-\int_0^T H_t dZ_t\right|\\
        &=\left|\int_0^T \int_0^T e^{-\beta|t-s|} dZ^n_t dZ_s-\int_0^T H_t dZ_t\right|\\
        &=\left|\int_0^T H^n_s dZ_s-\int_0^T H_t dZ_t\right|\\
        &\leq \|H^n-H\|_\infty TV(Z;[0-,T])\to 0.
    \end{align*}
    Taken together,
    \begin{align}
        &\left|\int_0^T\int_0^T e^{-\beta |t-s|} dZ^n_s dZ^n_t -\int_0^T\int_0^T e^{-\beta |t-s|} dZ_s dZ_t\right|\label{eqn:first.est.obj.approx}\\
        &\quad \quad \quad \quad \quad \quad \quad \quad \quad \quad \quad \quad=\left|\int_0^TH^n_tdZ^n_t-\int_0^TH_tdZ_t\right|\nonumber\\
        &\quad \quad \quad \quad \quad \quad \quad \quad \quad \quad \quad \quad \leq \left|\int_0^TH^n_t dZ^n_t-\int_0^T H_t dZ^n_t\right|+\left|\int_0^TH_t dZ_t^n-\int_0^T H_t dZ_t\right|\to 0.\nonumber
    \end{align}

    We turn to the final set of integrals in \eqref{eqn:limit.obj.approx.seq}. Define the function $I^{-i}$ by
    \[I^{-i}_t:= \int_0^te^{-\beta(t-s)}\sum_{j\not=i}dX^{*,j}_s, \ \ \ t\in[0,T],\]
    with $I_{0-}^{-i}=0$. Notice that this does not depend on $Z^n$ or $Z$, and by our assumption on $\boldsymbol{X}^*$, $I^{-i}$ has finite variation and is absolutely continuous on $(0,T)$. By integration by parts (using $Z^n_T=0$ and $I^{-i}_{0-}=0$),
    \begin{align*}
        \int_0^T\int_0^{t-} e^{-\beta (t-s)} \sum_{j\not=i}dX_s^{*,j} dZ^n_t&=\int_0^TI^{-i}_{t-} dZ^n_t= -\int_0^T Z_{t-}^n dI_t^{-i}-\sum_{t\in\{0,T\}} \Delta I^{-i}_t \Delta Z^n_t.
    \end{align*}
    Repeating this for $Z$, it follows that
    \begin{align*}\int_0^T\int_0^{t-} e^{-\beta (t-s)} \sum_{j\not=i}dX_s^{*,j} dZ^n_t-&\int_0^T\int_0^{t-} e^{-\beta (t-s)} \sum_{j\not=i}dX_s^{*,j} dZ_t\\
    &=\int_0^T (Z_{t-}-Z_{t-}^n) dI_t^{-i}+\sum_{t\in\{0,T\}} \Delta I^{-i}_t (\Delta Z_t-\Delta Z^n_t).
    \end{align*}
    The second term on the right-hand side is zero by the construction of $Z^n$ and property $(a)$ which gives that $\Delta Z_0=\Delta Z^n_0$ and $\Delta Z_T=\Delta Z^n_T$. It follows that
    \begin{align}
        &\left|\int_0^T\int_0^{t-} e^{-\beta (t-s)} \sum_{j\not=i}dX_s^{*,j} dZ^n_t-\int_0^T\int_0^{t-} e^{-\beta (t-s)} \sum_{j\not=i}dX_s^{*,j} dZ_t\right|\label{eqn:second.est.obj.approx}\\
        &\quad \quad \quad \quad \quad \quad \quad \quad \quad \quad \quad \quad \quad \quad = \left|\int_0^T (Z_{t-}^n -Z_{t-})dI_t^{-i}\right|\leq \int_0^T\left| Z_{t-}^n -Z_{t-}\right||dI_t^{-i}|,\nonumber
    \end{align}
    where $|dI_t^{-i}|$ denotes the total variation measure associated with $I^{-i}$. This measure is finite and its atoms, if any, are located at $0$ and $T$ where the integrand is zero by (a). On $(0,T)$ $|dI_t^{-i}|$ is absolutely continuous with respect to the Lebesgue measure. 
    
    Then, since $\|Z^n-Z\|_\infty$ is bounded by (c), $Z^n$ is continuous on $(0,T)$, $Z$ is c\`adl\`ag with $Z_{t-}=Z_t$ Lebesgue-a.e., and $Z^n\to Z$ Lebesgue-a.e., we have $Z^n_{t-}\to Z_{t-}$ Lebesgue-a.e. on $(0,T)$. As $|dI^{-i}|$ is absolutely continuous with respect to the Lebesgue measure on $(0,T)$, and as the endpoint atoms have already been accounted for above, this convergence also holds $|dI^{-i}|$-a.e.
\begin{equation}\label{eqn:third.est.obj.approx} \int_0^T\left| Z_{t-}^n -Z_{t-}\right||dI_t^{-i}|\to 0
    \end{equation}
    by dominated convergence.
    Combining \eqref{eqn:first.est.obj.approx},  \eqref{eqn:second.est.obj.approx}, and \eqref{eqn:third.est.obj.approx} proves \eqref{eqn:limit.obj.approx.seq}.

    We conclude
    \[\lim_{n\to\infty}J_B(Z^n;\boldsymbol{X}^{*,-i})\leq J_B(Z;\boldsymbol{X}^{*,-i})<J_B(X^{*,i};\boldsymbol{X}^{*,-i}).\]
    But then
     $J_B(Z^{n};\boldsymbol{X}^{*,-i})<J_B(X^{*,i};\boldsymbol{X}^{*,-i})$
     for $n$ sufficiently large, contradicting the optimality of $X^{*,i}$ in the class of strategies that are absolutely continuous on $(0,T)$. \qed

\subsection{Lemma \ref{lem:block.cost.ODE}}

To derive the characterization in Lemma \ref{lem:block.cost.ODE} we will make use of the fact that $\mathbb{R}\times L^2[0,T]$ is a Hilbert space when equipped with the inner product
\[\langle (a,v), (a',v') \rangle = a a' + \langle v,v'\rangle_{L^2[0,T]}.\]
As in the proof of Lemma \ref{lem:term.pen.ODE}, we will take the Gateaux differential of $\mathcal{J}_{B}$, but this time in an arbitrary direction $(h,\eta)\in \mathbb{R}\times L^2[0,T]$. We begin by taking the Gateaux differential of the objective $\mathcal{J}_B(\cdot,\cdot;\boldsymbol{v}^{-i})$ in the direction $(0,\eta)$ in steps. Observe
\begin{equation*}
\delta_{(0,\eta)} b^{i}=-\int_0^T\eta_tdt, \ \ \ \text{and} \ \ \ 
    \delta_{(0,\eta)} \left(\frac{\theta_a}{2}(a^{i})^2+\frac{\theta_b}{2}(b^{i})^2\right)=\theta_b b^{i}\delta_{(0,\eta)} b^{i}=-\int_0^T\theta_b b^{i}\eta_tdt.
\end{equation*}
Similarly,
\begin{equation*}
    \delta_{(0,\eta)} I_t=\int_0^t e^{-\beta(t-s)}\lambda\eta_sds, \ \ \  t\in[0,T)
\end{equation*}
and
\begin{align*}
    \delta_{(0,\eta)}\left[\frac{1}{2}(I_{T-}+I_T)b^{i}\right]&=\delta_{(0,\eta)}\left[I_{T-}b^{i}+\frac{\lambda}{2}\left(\sum_{j=1}^Nb^{j}\right)b^{i}\right]\\
    &=\delta_{(0,\eta)} I_{T-}b^{i}+I_{T}\delta_{(0,\eta)} b^{i}-\frac{\lambda}{2} \left(\sum_{j\not=i}b^{j}\right)\delta_{(0,\eta)} b^{i}.
\end{align*}
Working now directly with $\mathcal{J}_B$ we get
\begin{align*}
    \delta_{(0,\eta)} \mathcal{J}_{B}(a^{i},v^{i};\boldsymbol{v}^{-i})
    &=\int_0^T[\delta_{(0,\eta)} I_tv_t^{i}+\eta_tI_t]dt+\delta_{(0,\eta)} I_{T-}b^{i}+I_{T}\delta_{(0,\eta)} b^{i} \\
    &\quad \quad -\frac{\lambda}{2} \left(\sum_{j\not=i}b^{j}\right)\delta_{(0,\eta)} b^{i}-\int_0^T\theta_b b^{i}\eta_tdt\\
    &=\int_0^T[\delta_{(0,\eta)} I_tv_t^{i}+\eta_tI_t]dt+b^{i}\int_0^T e^{-\beta(T-t)}\lambda\eta_tdt-\int_0^TI_{T}\eta_tdt\\
    &\quad \quad +\int_0^T\frac{\lambda}{2} \left(\sum_{j\not=i}b^{j}\right)\eta_tdt-\int_0^T\theta_b b^{i}\eta_tdt\\
    &=\int_0^T\delta_{(0,\eta)} I_tv_t^{i}dt+\int_0^T\left[I_t-I_{T}+\lambda b^{i} e^{-\beta(T-t)}+\frac{\lambda}{2}\left(\sum_{j\not=i}b^{j}\right)-\theta_b b^{i}\right]\eta_tdt.
\end{align*}
By changing the order of integration, the first term can be written
\begin{align*}
    \int_0^T\delta_{(0,\eta)} I_tv_t^{i}dt&=\int_0^T\int_0^t\lambda e^{-\beta(t-s)}v_t^{i}\eta_sdsdt=\int_0^T\int_s^T\lambda e^{-\beta(t-s)}v_t^{i}\eta_sdtds.
\end{align*}
If we let
\begin{equation*}
    Y_t^{i}:=\lambda b^{i} e^{-\beta(T-t)}+\int_t^T\lambda e^{-\beta (s-t)}v_s^{i} ds, \ \ \ t\in [0,T],
\end{equation*}
then the Gateaux differential can be written
\begin{equation}\label{eqn:dJ.eta.block}
    \delta_{(0,\eta)} \mathcal{J}_{B}(a^{i},v^{i};\boldsymbol{v}^{-i})=\int_0^T\left[Y^{i}_t+I_t-I_{T}+\frac{\lambda}{2}\left(\sum_{j\not=i}b^{j}\right)-\theta_b b^{i}\right]\eta_tdt.
\end{equation}

Next, we take the Gateaux differential in the direction $(1,0)\in\mathbb{R}\times L^2[0,T]$. Note that 
\begin{equation*}
    \delta_{(1,0)} I_t=\lambda e^{-\beta t}, \ \ t\in[0,T),
\end{equation*}
\begin{equation*}
    \delta_{(1,0)} \left(\frac{\theta_a}{2}(a^{i})^2+\frac{\theta_b}{2}(b^{i})^2\right)=\theta_a a^{i}+\theta_b b^{i} \delta_{(1,0)} b^{i}=\theta_aa^{i}-\theta_b b^{i},
\end{equation*}
and
\begin{align*}\delta_{(1,0)}\left[\frac{1}{2}I_0a^{i}\right]&=\frac{1}{2}\left[\delta_{(1,0)} I_0 a^{i}+\lambda \left(a^{i}+\sum_{j\not=i}a^{j}\right)\right]=\lambda a^{i}+\frac{\lambda}{2}\sum_{j\not=i}a^{j}.
\end{align*}
Moreover,
\begin{align*}
    \delta_{(1,0)}\left[\frac{1}{2}\left(I_{T-}+I_T\right)b^{i}\right]&=\delta_{(1,0)}\left[I_{T-}b^{i}+\frac{\lambda}{2}\left(b^{i}+\sum_{j\not=i}b^{j}\right)b^{i}\right]\\
    &=\delta_{(1,0)}I_{T-}b^{i}+I_{T-}\delta_{(1,0)} b^{i}+\lambda b^{i}\delta_{(1,0)} b^{i}+\frac{\lambda}{2}\left(\sum_{j\not=i}b^{j}\right)\delta_{(1,0)} b^{i}\\
    &=\delta_{(1,0)} I_{T-}b^{i}-I_{T-}-\lambda b^{i}-\frac{\lambda}{2}\left(\sum_{j\not=i}b^{j}\right)\\
    &=\delta_{(1,0)} I_{T-}b^{i}-I_T+\frac{\lambda}{2}\left(\sum_{j\not=i}b^{j}\right).
\end{align*}
Putting all this together we find
\begin{align*}
\delta_{(1,0)} & \mathcal{J}_{B}(a^i,v^i;\boldsymbol{v}^{-i}) \\
&=  \lambda a^{i}+\frac{\lambda}{2}\sum_{j\not=i}a^{j}+ \int_0^T\delta_{(1,0)} I_t v_t^{i} dt+\delta_{(1,0)}I_{T-}b^{i} -I_T+\frac{\lambda}{2}\left(\sum_{j\not=i}b^{j}\right)+\theta_a a^{i}-\theta_b b^{i}\\
    &=\lambda a^{i}+\frac{\lambda}{2}\sum_{j\not=i}a^{j}+ \int_0^T\lambda v_t^{i} e^{-\beta t}dt +\lambda e^{-\beta T}b^{i} -I_{T}+\frac{\lambda}{2}\left(\sum_{j\not=i}b^{j}\right)+\theta_a a^{i}-\theta_b b^{i}.
\end{align*}
Rewriting the last equality yields
\begin{equation}\label{eqn:dJ.a.block}
\delta_{(1,0)}\mathcal{J}_{B}(a^i,v^i;\boldsymbol{v}^{-i})= I_0-I_{T} + Y_0^i+\theta_a a^{i}-\frac{\lambda}{2}\sum_{j\not=i}a^{j}+\frac{\lambda}{2}\sum_{j\not=i}b^{j}-\theta_b b^{i}.
\end{equation}
Combining \eqref{eqn:dJ.eta.block} and \eqref{eqn:dJ.a.block} we get the Gateaux differential in an arbitrary direction $(h,\eta)$,
\begin{align*}
    \delta_{(h,\eta)}\mathcal{J}_{B}(a^i,v^i;\boldsymbol{v}^{-i}) &= \left(I_0-I_{T} + Y_0^i+\theta_a a^{i}-\frac{\lambda}{2}\sum_{j\not=i}a^{j}+\frac{\lambda}{2}\sum_{j\not=i}b^{j}-\theta_b b^{i}\right) h\\
    & \quad \quad + \int_0^T\left[Y^{i}_t+I_t-I_{T}+\frac{\lambda}{2}\left(\sum_{j\not=i}b^{j}\right)-\theta_b b^{i}\right]\eta_tdt.
\end{align*}
Reasoning as in Lemma \ref{lem:term.pen.ODE} (see also Appendix \ref{app:opt.theory}) we have that the necessary and sufficient first-order conditions for optimality are
\begin{equation}\label{eqn:foc.block.tax.eta}    \boxed{I_{T}=Y^{i}_t+I_t+\frac{\lambda}{2}\left(\sum_{j\not=i}b^{j}\right)-\theta_b b^{i}, \ \ \ t\in[0,T)}
\end{equation}
and
\begin{equation}\label{eqn:foc.block.tax.a}    \boxed{I_T=Y_{0}^{i}+I_0+\frac{\lambda}{2}\left(\sum_{j\not=i}b^{j}\right)-\theta_b b^{i}+\theta_a a^{i}-\frac{\lambda}{2}\left(\sum_{j\not=i}a^{j}\right).}
\end{equation}

By symmetry, an equilibrium is achieved if and only if these first-order conditions hold simultaneously for all $i=1,\dots,N$. Differentiating \eqref{eqn:foc.block.tax.eta} (note that $I$ and $Y$ are differentiable almost everywhere by definition) leads to the equilibrium system of ODEs,
\begin{align*}
    0&=\dot{Y}^{i}_t+\dot{I}_t, \ \ \ i=1,\dots,N,\\
    \dot{Y}^{i}_t&=\beta Y^{i}_t-\lambda v^{i}_t, \ \ \ i=1,\dots,N,\\
    \dot{I}_t&=-\beta I_t+\lambda\sum_{i=1}^Nv_t^{i}.
\end{align*}
Writing this in terms of $X_t^{i}$, for almost every $t\in(0,T)$,
\begin{align*}
    0&=\dot{Y}^{i}_t+\dot{I}_t, \ \ \ i=1,\dots,N,\\
    \dot{Y}^{i}_t&=\beta Y^{i}_t-\lambda \dot{X}_t^{i}, \ \ \ i=1,\dots,N,\\
    \dot{I}_t&=-\beta I_t+\lambda\sum_{i=1}^N\dot{X}_t^{i},
\end{align*}
subject to the initial and terminal conditions
\begin{align*}
    I_0&=\lambda\sum_{i=1}^Na^{i}, \ \ \ 
    X_0^{i}=x^{i}+a^{i}, \ \ \ 
    Y_T^{i}=\lambda b^{i}, \ \ \ i=1,\dots,N.
\end{align*}
By rearranging the system we can write it in the standard form reported in Lemma \ref{lem:block.cost.ODE}. Moreover, from that representation we see that any solution to this system must have derivatives that are equal almost everywhere to continuous functions. Thus, without loss of generality, we identify the derivatives of the equilibrium $I$, $Y^i$ and $X^i$ (if an equilibrium exists) with their continuous versions and interpret the ODE in the classical sense.

We now turn to the additional consistency conditions that the initial and terminal block trades must satisfy. Specifically, letting $t\downarrow 0 $ in \eqref{eqn:foc.block.tax.eta} and comparing with \eqref{eqn:foc.block.tax.a} we find that the optimal $a$ must satisfy
\begin{equation}\label{eqn:a.cond}\boxed{\theta_a a^{i} =\frac{\lambda}{2}\sum_{j\not=i}a^{j}.}
\end{equation}
Sending $t\uparrow T$ in \eqref{eqn:foc.block.tax.eta} we get
\begin{equation*}
    I_{T}=\lambda b^{i}+I_{T-}+\frac{\lambda}{2}\left(\sum_{j\not=i}b^{j}\right)-\theta_b b^{i}.
\end{equation*}
By rearranging,
\[\lambda\sum_{j=1}^Nb^{j}=\lambda b^{i}+\frac{\lambda}{2}\left(\sum_{j\not=i}b^{j}\right)-\theta_b b^{i},\]
which yields
\begin{equation}\label{eqn:b.cond.1}
\boxed{\theta_b b^{i}=-\frac{\lambda}{2}\sum_{j\not=i}b^{j}.}
\end{equation}
Moreover, since we require liquidation by $T$,
\begin{equation}\label{eqn:b.cond.2}
    \boxed{b^{i}=-X^{i}_{T-}, \ \ \ i=1,\dots,N.}
\end{equation}
Again, by symmetry in equilibrium, we arrive at the conditions reported in the lemma. Note that \eqref{eqn:a.cond} and \eqref{eqn:b.cond.1} are exactly the conditions of Proposition \ref{prop:no.interior.jumps}. In summary, these ODEs and consistency conditions necessarily hold in equilibrium.

For sufficiency we use that an equilibrium is attained if and only if \eqref{eqn:foc.block.tax.eta} and \eqref{eqn:foc.block.tax.a}  hold. Enforcing \eqref{eqn:a.cond}, \eqref{eqn:b.cond.1} and \eqref{eqn:b.cond.2} we have that \eqref{eqn:foc.block.tax.eta} and \eqref{eqn:foc.block.tax.a} reduce (after some manipulation) to
\[I_{T-}-I_t=-(Y^{i}_T-Y^{i}_t) \ \ \ \text{and} \ \ \ I_{T-}-I_0=-(Y^{i}_T-Y^{i}_{0}).\]
This must be satisfied if the ODE holds for $I$ on $[0,T)$ and $Y$ on $[0,T]$ since
\[I_{T-}-I_t=\int_t^T\dot I_sds \ \ \ \text{and} \ \ \ -(Y^{i}_T-Y^{i}_t)=-\int_t^T\dot Y_s^ids.\]
The right-hand side of both these equations must coincide by the ODE for all $t\in[0,T)$.
Thus, to obtain an equilibrium for the game it is sufficient to solve the ODE for fixed boundary conditions and enforce the consistency equations \eqref{eqn:a.cond},  \eqref{eqn:b.cond.1}, and \eqref{eqn:b.cond.2} for the boundaries. \qed

\subsection{Theorem \ref{thm:equil.block.cost}}

We will prove Theorem \ref{thm:equil.block.cost} in stages. First we will fix the boundary conditions and solve the ODE from Lemma \ref{lem:block.cost.ODE}. Then, we will use the consistency conditions to show conclusions (1)--(4) for a deterministic equilibrium and derive the form of the equilibrium (when it exists). Finally, we will show that when a deterministic equilibrium does not exist, neither does an equilibrium in the full class of admissible strategies.

\subsubsection*{Step 1: Solving the ODE for fixed boundary conditions.}
Defining the vector $\boldsymbol{F}=(I,Y^{1},\dots,Y^{N},X^{1},\dots,X^{N})$, we can write the ODE  from Lemma~\ref{lem:block.cost.ODE} in the matrix form
\[\dot{\boldsymbol{F}}_t=A \boldsymbol{F}_t\]
where
\[A=\frac{\beta}{N-1}\begin{bmatrix}
1 & -1 & -1 &  \dots & -1 & -1 & 0 & \dots & 0 \\
-1 & 1 & 1 & \dots & 1 & 1 & 0 & \dots & 0 \\
\vdots & \vdots & \vdots & \vdots & \vdots & \vdots & \vdots & \vdots & \vdots \\
-1 & 1 & 1 & \dots & 1 & 1 & 0 & \dots & 0 \\
\lambda^{-1} & (N-2)\lambda^{-1} & -\lambda^{-1} & \dots & -\lambda^{-1} & -\lambda^{-1} & 0 & \dots & 0 \\
 \lambda^{-1} &  -\lambda^{-1}  & (N-2)\lambda^{-1} & \dots & -\lambda^{-1} & -\lambda^{-1} & 0 & \dots & 0 \\
\vdots &  \vdots  & \vdots & \ddots & \vdots & \vdots  & \vdots & \vdots \\
\lambda^{-1} &  -\lambda^{-1}  &  -\lambda^{-1}  & \dots & (N-2)\lambda^{-1} & -\lambda^{-1} & 0 & \dots & 0 \\
\lambda^{-1} &  -\lambda^{-1}  &  -\lambda^{-1}  & \dots & -\lambda^{-1} & (N-2)\lambda^{-1} & 0 & \dots & 0 \\
\end{bmatrix}.\]
The only non-zero eigenvalue of this matrix is $\mathbbm{z}_1=\frac{N+1}{N-1}\beta$. The associated eigenvector is
\[\boldsymbol{q}_1=\begin{bmatrix}
    \frac{N+1}{2}\lambda \\
    -\frac{N+1}{2}\lambda \\
    \vdots\\
     -\frac{N+1}{2}\lambda \\
    1\\
    \vdots\\
    1
\end{bmatrix}.\]
The remaining (non-zero) eigenvectors are $\boldsymbol{q}_{1+j}=\boldsymbol{e}_{N+1+j}$ for $j=1,\dots,N$ where $\boldsymbol{e}_i$ is the $i$th Euclidean basis vector. This is the exhaustive set of linearly independent eigenvectors corresponding to the eigenvalue $0$ since the dimension of the solution to the eigenvector equation
\[A\boldsymbol{q}=\boldsymbol{0}\]
is exactly the dimension of the nullspace of $A$. We can see that $A$ is of rank $N+1$ (as the first $N+1$ columns are linearly independent) and thus, by the rank--nullity theorem the nullspace has dimension $N$.

In order to completely characterize the solution to the ODE we must find $N$ additional generalized eigenvectors corresponding to the eigenvalue $0$. These generalized eigenvectors $(\boldsymbol{q}_{N+1+j})_{j=1}^N$ solve
\[A\boldsymbol{q}_{N+1+j}=\boldsymbol{q}_{1+j}, \ \ \ j=1,\dots,N\]
and must be linearly independent of $\boldsymbol{q}_{j}$, $j=1,\dots,N+1$. It is easy to verify that we can take $\boldsymbol{q}_{N+1+j}=\frac{\lambda}{\beta}[\boldsymbol{e}_1+\boldsymbol{e}_{1+j}]$ for $j=1,\dots,N$.
Therefore, the general solution to the matrix ODE is given by
\[\boldsymbol{F}_t=c_1\boldsymbol{q}_{1}e^{\mathbbm{z}_1t}+\sum_{j=1}^{N}(c_{1+j}\boldsymbol{q}_{1+j}+c_{N+1+j}  (t\boldsymbol{q}_{1+j}+\boldsymbol{q}_{N+1+j})).\]

We now solve for the constants $\boldsymbol{c}:=(c_1,\dots, c_{2N+1})^\top$ by enforcing the boundary conditions. Recall that $x^i$ denotes trader $i$’s initial inventory, and $a^i$ and $b^i$ denote the block trades executed at $t=0$ and $t=T$, respectively.
Reading off of the equation for $\boldsymbol{F}$ we get
\begin{align*}
    I_0 & = c_1 \frac{N+1}{2}\lambda + \frac{\lambda}{\beta}\sum_{i=1}^N c_{N+1+i} = \lambda \sum_{i=1}^N a^{i}, \\
    Y_T^{i} & = -c_1 \frac{N+1}{2}\lambda e^{\mathbbm{z}_1 T}+ \frac{\lambda}{\beta}c_{N+1+i}=\lambda b^{i}, \ \ \ i=1,\dots, N,\\
    X_0^{i}&=c_1+c_{1+i}=x^{i}+a^{i}, \ \ \ i=1,\dots,N.
\end{align*}
Summing the second set of equations over $i$ we find
\[-c_1 \frac{N(N+1)}{2}\lambda e^{\mathbbm{z}_1 T}+\frac{\lambda}{\beta}\sum_{i=1}^N c_{N+1+i}=\lambda \sum_{i=1}^N b^{i}.\]
Subtracting this from the first equation,
\[c_1 \frac{N+1}{2}\lambda+c_1 \frac{N(N+1)}{2}\lambda e^{\mathbbm{z}_1 T}=\lambda \sum_{i=1}^N (a^{i}-b^{i}).\]
Solving for $c_1$ yields
\[c_1 = \frac{2\sum_{i=1}^N (a^{i}-b^{i})}{(N+1) (1+Ne^{\mathbbm{z}_1 T})}.\]
Substituting this back into the initial equations we find
\[c_{1+i} = x^{i}+a^{i}-\frac{2\sum_{j=1}^N (a^{j}-b^{j})}{(N+1) (1+Ne^{\mathbbm{z}_1 T})}, \ \ \ i=1,\dots,N \]
and
\[c_{N+1+i} = \beta b^{i}+\beta e^{\mathbbm{z}_1 T}\frac{\sum_{j=1}^N (a^{j}-b^{j})}{(1+Ne^{\mathbbm{z}_1 T})}, \ \ \ i=1,\dots,N.\]
These define the solution to the ODE uniquely for fixed $x^i,a^i$ and $b^i$, $i=1,\dots,N$.

\subsubsection*{Step 2: Solving for the deterministic equilibrium.}

We now study the conditions that $\boldsymbol{a} = (a^1,\dots,a^N)^\top$ and $\boldsymbol{b}=(b^1,\dots,b^N)^\top$ need to satisfy. Our first result characterizes the solution to the system of equations for $\boldsymbol{a}$ and $\boldsymbol{b}$. 

\begin{lemma}\label{lem:bdy.sys.of.eqns}\mbox{}

\begin{enumerate}
    \item If $\boldsymbol{a}$ satisfies
\[\theta_a a^{i}=\frac{\lambda}{2}\sum_{j\not=i}a^{j}, \ \ \ i=1,\dots,N\]
then $a^i = a^j$ for all $i,j=1,\dots,N$. If, in addition, $\theta_a \not=\frac{\lambda (N-1)}{2}$ then $\boldsymbol{a}=\boldsymbol{0}$.
    \item If $\boldsymbol{b}$ satisfies
\[\theta_b b^{i}=-\frac{\lambda}{2}\sum_{j\not=i}b^{j}, \ \ \ i=1,\dots,N\]
then $\sum_{i=1}^N b^i =0$. If, in addition, $\theta_b \not=\frac{\lambda}{2}$ then $\boldsymbol{b}=\boldsymbol{0}$.
\end{enumerate}
\end{lemma}

\begin{proof}
    Let $\boldsymbol{1}=(1,\dots,1)^\top$ be the vector of ones and $\mathrm{Id}$ be the identity matrix.  The system of equations for $\boldsymbol{a}$ can be written as $M\boldsymbol{a}=\boldsymbol{0}$ where $M=(\theta_a+\frac{\lambda}{2}) \mathrm{Id}-\frac{\lambda}{2}\boldsymbol{1}\boldsymbol{1}^\top$. 
    Similarly, the system of equations for $\boldsymbol{b}$ can be written as $\tilde{M}\boldsymbol{b}=\boldsymbol{0}$
    where $\tilde{M}=(\theta_b-\frac{\lambda}{2}) \mathrm{Id}+\frac{\lambda}{2}\boldsymbol{1}\boldsymbol{1}^\top$.

If $M$ is invertible then we must have $\boldsymbol{a}=\boldsymbol{0}$.  Similarly, if $\tilde{M}$ is invertible, then $\boldsymbol{b}=\boldsymbol{0}$. By the Matrix Determinant Lemma,
    \[\mathrm{det}(M)=\left(\theta_a+\frac{\lambda}{2}\right)^N-N\frac{\lambda}{2}\left(\theta_a+\frac{\lambda}{2}\right)^{N-1},\]
    \[\mathrm{det}(\tilde{M})=\left(\theta_b-\frac{\lambda}{2}\right)^N+N\frac{\lambda}{2}\left(\theta_b-\frac{\lambda}{2}\right)^{N-1}.\]
    As a function of $\theta_a$, the first equation has a root $\theta_a=-\frac{\lambda}{2}$ of multiplicity $N-1$ and a remaining root of $\theta_a=\frac{\lambda(N-1)}{2}$. As a function of $\theta_b$, the second equation has a root $\theta_b=\frac{\lambda}{2}$ of multiplicity $N-1$ and a remaining root of $\theta_b=-\frac{\lambda(N-1)}{2}$. Since $\lambda>0$, $\theta_a,\theta_b\geq0$, the unique feasible value of $\theta_a$ (resp.\ $\theta_b$) that leads to the non-invertibility of $M$ (resp.\ $\tilde{M}$) is $\theta_a=\frac{\lambda(N-1)}{2}$ (resp. $\theta_b=\frac{\lambda}{2}$). 
    
    Suppose now that $\theta_a=\frac{\lambda(N-1)}{2}$ and $\theta_b = \frac{\lambda}{2}$. The nullspace of $M$ is characterized by the $1$-dimensional space of vectors $\boldsymbol{a}$ taking the form $\boldsymbol{a}=c\boldsymbol{1}$,
    for some $c\in\mathbb{R}$. Similarly, the nullspace of $\tilde{M}$ is given by the $N-1$ dimensional space of vectors $\boldsymbol{b}$ satisfying $\boldsymbol{1}^T \boldsymbol{b}=0$.
    This completes the proof.
\end{proof}

By Lemma \ref{lem:bdy.sys.of.eqns} there must be an $\alpha\in \mathbb{R}$ such that
\[a^{i}=\alpha \ \ \ i =1,\dots,N, \ \ \ \text{and} \ \ \  \sum_{i=1}^Nb^{i} = 0.\]
Enforcing the conditions on $\boldsymbol{a}$ and $\boldsymbol{b}$ in the equations for the constants $\boldsymbol{c}$, we get
\[\begin{cases}
    c_1=\frac{2N \alpha}{(N+1)(1+N e^{\mathbbm{z}_1 T})},\\
    c_{1+i} = x^{i}+\alpha -\frac{2N \alpha}{(N+1)(1+Ne^{\mathbbm{z}_1 T})} & i=1,\dots,N,\\
    c_{N+1+i} = \beta b^{i} + \beta \frac{Ne^{\mathbbm{z}_1T}\alpha}{1+Ne^{\mathbbm{z}_1 T}} & i=1,\dots,N.
\end{cases}\]
Let
\[u_N := \frac{2N}{(N+1)(1+Ne^{\mathbbm{z}_1T})}, \ \ \ w_N:= \frac{Ne^{\mathbbm{z}_1 T}}{1+N e^{\mathbbm{z}_1 T}}\]
so that we may write more concisely,
\[\begin{cases}
    c_1=u_N\alpha,\\
    c_{1+i} = x^{i}+(1-u_N)\alpha & i=1,\dots,N,\\
    c_{N+1+i} = \beta b^{i} + \beta w_N\alpha & i=1,\dots,N.
\end{cases}\]
It remains to enforce the terminal conditions \eqref{eqn:b.cond.2}. Thus, we solve the $N+1$ dimensional system
\[\begin{cases}
    b^{i} = -X_{T-}^{i}= -\left[c_1e^{\mathbbm{z}_1 T}+c_{1+i}+c_{N+1+i} T\right],& i=1,\dots,N,\\
    \sum_{i=1}^N b^{i}=0
\end{cases}\]
for the $N+1$ unknowns given by $\alpha$ and $\boldsymbol{b}$. Here we use that $b^i=-X_{T-}^i$ since the last block trade must unwind the remaining inventory and recall the representation
\begin{equation}\label{eqn:X.expression.constants}
X_t^i= c_1e^{\mathbbm{z}_1 t}+c_{1+i}+c_{N+1+i}t, \ \ \ i=1,\dots,N.
\end{equation} Inserting the form of $\boldsymbol{c}$,
\begin{align*}b^{i} 
&=-\left[\left(u_Ne^{\mathbbm{z}_1 T}+(1-u_N)+\beta w_N T\right)\alpha+x^{i}+\beta b^{i}T\right].
\end{align*}
Let
\[r_N:=1+(e^{\mathbbm{z}_1T}-1)u_N+\beta w_N T\]
so that
\[b^{i}=-r_N \alpha - x^{i}-\beta b^{i} T.\]
Summing over $i$ and using that $\sum_{i=1}^Nb^{i}=0$ we have
\[0=-Nr_N\alpha -\sum_{i=1}^N x^{i}.\]
This implies \begin{equation}\label{eqn:alpha.form}
    \boxed{\alpha = -r_N^{-1} \overline{x}.} 
\end{equation}
Substituting this back in, we get $b^{i} = \overline{x}-x^{i}-\beta b^{i} T$
which yields 
\begin{equation}\label{eqn:b.form}
    \boxed{b^{i} = -\frac{x^{i}-\overline{x}}{1+\beta T}.}
\end{equation}
With this the final solution for the constants is
\[\begin{cases}
    c_1=-\frac{u_N}{r_N}\overline{x},\\
    c_{1+i} = x^{i}-\frac{(1-u_N)}{r_N}\overline{x}=(x^{i}-\overline{x})+\frac{r_N-1+u_N}{r_N}\overline{x}& i=1,\dots,N,\\
    c_{N+1+i} = - \frac{\beta}{1+\beta T}(x^{i}-\overline{x})- \frac{\beta w_N}{r_N}\overline{x} & i=1,\dots,N.
\end{cases}\]
Using \eqref{eqn:X.expression.constants} we recover (after simplification) the form of $\boldsymbol{X}^{*}$ reported in Theorem \ref{thm:equil.block.cost}.
If $\theta_a = \frac{(N-1)\lambda}{2}$ \textbf{and} $\theta_b = \frac{\lambda}{2}$, then by Proposition \ref{prop:nash.eq.unique} and Lemma \ref{lem:nash.eq.a.c.} this defines the unique equilibrium for all initial inventories. Hence, we have shown case (1) of Theorem \ref{thm:equil.block.cost}.

Next, we address cases (2)--(4) for deterministic controls. If 
$\theta_a\not=\frac{\lambda(N-1)}{2}$ then 
Lemma~\ref{lem:bdy.sys.of.eqns} implies $\alpha =0$. But then 
\eqref{eqn:alpha.form} can only be true if $\overline{x}=0$. On the 
other hand, if $\theta_b\not=\frac{\lambda}{2}$ then 
Lemma~\ref{lem:bdy.sys.of.eqns} further requires that 
$\boldsymbol{b}=\boldsymbol{0}$. In order to have a consistent solution, 
\eqref{eqn:b.form} mandates that $x^i=\overline{x}=x^j$ for all 
$i,j=1,\dots,N$. Thus, in case (4), both restrictions must hold, and hence 
$x^i=0$ for all $i=1,\dots,N$. In summary, for cases (2)--(4) we have 
derived the unique equilibrium (again by Proposition~\ref{prop:nash.eq.unique}
and Lemma~\ref{lem:nash.eq.a.c.}) when the initial inventories satisfy the 
stated conditions. We also have that a \textbf{deterministic} equilibrium 
cannot exist otherwise by the aforementioned inconsistencies. It remains to 
show that this non-existence generalizes to arbitrary equilibria.

\subsubsection*{Step 3: Extending non-existence to the class of admissible strategies.}

We will treat cases (2) and (3) of Theorem \ref{thm:equil.block.cost}  separately; 
case~(4) will then follow by combining them.

\paragraph*{Case (2):}
It suffices to show that if $\vartheta_0 \not=\frac{\lambda (N-1)}{2}$ and a Nash equilibrium exists, then $\overline{x}=0$.

Proposition \ref{prop:no.interior.jumps} implies an equilibrium system of equations for $\Delta X_{0}^{*,i}$ and $\Delta X_{T}^{*,i}$ that coincides with the one in Lemma \ref{lem:bdy.sys.of.eqns}. We have assumed $\vartheta_0 \not=\frac{\lambda (N-1)}{2}$, so we can conclude that
$\Delta X_0^{*,i} =0$ for all $i$ and $\sum_{i=1}^N\Delta X_T^{*,i} = 0$. Moreover, again by Proposition \ref{prop:no.interior.jumps}, there are no interior jumps. Then, by Lemma \ref{lem:gen.foc.jumps}, for all predictable $\sigma\geq \tau$,
    \begin{align*}
         &\mathbb{E}\left[\lambda\int_0^Te^{-\beta |\sigma-t|}dX_t^{*,i}+\lambda\int_0^{\sigma-} e^{-\beta (\sigma-t)}\sum_{j\not=i}dX_t^{*,j} + \left(\vartheta_T\Delta X_T^{*,i}+\frac{\lambda}{2}\sum_{j\not=i}\Delta X_T^{*,j}\right)\mathds{1}_{\{\sigma=T\}}\bigg|\mathcal{F}_{\tau-}\right]\\
         &=\mathbb{E}\left[\lambda\int_0^Te^{-\beta |\tau-t|}dX_t^{*,i}+\lambda\int_0^{\tau-} e^{-\beta (\tau-t)}\sum_{j\not=i}dX_t^{*,j} + \left(\vartheta_T\Delta X_T^{*,i}+\frac{\lambda}{2}\sum_{j\not=i}\Delta X_T^{*,j}\right)\mathds{1}_{\{\tau=T\}}\bigg|\mathcal{F}_{\tau-}\right].
     \end{align*}
     Recall from Proposition \ref{prop:no.interior.jumps} that
     \[ \vartheta_T\Delta X_{T}^{*,i}=-\frac{\lambda}{2}\sum_{j\not=i}\Delta X_{T}^{*,j}, \ \ \ i=1,\dots,N.\]
     Substituting this in, and rearranging terms,
\begin{align*}
         \mathbb{E}&\left[\lambda\int_\sigma ^Te^{-\beta (t-\sigma)}dX_t^{*,i}+\lambda\int_0^{\sigma-} e^{-\beta (\sigma-t)}\sum_{j=1}^NdX_t^{*,j}\bigg|\mathcal{F}_{\tau-}\right]\\
         &=\mathbb{E}\left[\lambda\int_\tau^Te^{-\beta (t-\tau)}dX_t^{*,i}+\lambda\int_0^{\tau-} e^{-\beta (\tau-t)}\sum_{j=1}^NdX_t^{*,j}\bigg|\mathcal{F}_{\tau-}\right].
\end{align*}
By averaging over $i$ and using that the average process, $\overline{X}$, has no jumps,
\begin{align*}
         \mathbb{E}&\left[\lambda\int_\sigma ^Te^{-\beta (t-\sigma)}d\overline{X}_t+N\lambda\int_0^{\sigma} e^{-\beta (\sigma-t)}d\overline{X}_t\bigg|\mathcal{F}_{\tau-}\right]\\
         &=\mathbb{E}\left[\lambda\int_\tau^Te^{-\beta (t-\tau)}d\overline{X}_t+N \lambda\int_0^{\tau} e^{-\beta (\tau-t)}d\overline{X}_t\bigg|\mathcal{F}_{\tau-}\right].
\end{align*}
Setting $\tau\equiv 0$,
\begin{align*}
         \mathbb{E}&\left[\lambda\int_\sigma ^Te^{-\beta (t-\sigma)}d\overline{X}_t+N\lambda\int_0^{\sigma} e^{-\beta (\sigma-t)}d\overline{X}_t\right]=\mathbb{E}\left[\lambda\int_0^Te^{-\beta t}d\overline{X}_t\right].
\end{align*}
Now the right-hand side is just a constant in $\sigma$. Taking $\sigma$ to be any deterministic time in $[0,T]$,  this implies
\begin{align*}
         \mathbb{E}&\left[\lambda e^{\beta t}\int_t ^Te^{-\beta s}d\overline{X}_s+N\lambda e^{-\beta t}\int_0^{t} e^{\beta s}d\overline{X}_s\right]=\mathbb{E}\left[\lambda\int_0^Te^{-\beta t}d\overline{X}_t\right], \ \ \ \forall t\in[0,T].
\end{align*}
Integrating by parts and dividing by $\lambda$,
\begin{multline*}
        \mathbb{E}\bigg[ e^{-\beta (T-t)}\overline{X}_T-\overline{X}_t +  \beta \int_t ^T\overline{X}_s e^{-\beta (s-t)}ds+N \left( \overline{X}_t -e^{-\beta t}\overline{X}_0-\beta \int_0^{t} e^{-\beta (t-s)}\overline{X}_sds \right)\bigg]\\
        =\mathbb{E}\left[e^{-\beta T} \overline{X}_T -\overline{X}_0 +\beta \int_0^Te^{-\beta t}\overline{X}_tdt\right], \ \ \ \forall t\in[0,T].
\end{multline*}
Using that $\overline{X}_T=0$ and $\overline{X}_0=\overline{x}$ we have, after collecting terms,
\begin{multline*}
        \mathbb{E}\bigg[(N-1) \overline{X}_t + \beta \int_t^T e^{-\beta (s-t)}\overline{X}_sds-N e^{-\beta t}\overline{x}-N \beta \int_0^{t} e^{-\beta (t-s)}\overline{X}_sds \bigg]\\
        =\mathbb{E}\left[-\overline{x} +\beta \int_0^Te^{-\beta t}\overline{X}_tdt\right], \ \ \ \forall t\in[0,T].
\end{multline*}
Writing
$m_t := \mathbb{E}[\overline{X}_t]$, Fubini's theorem yields
\begin{multline*}
        (N-1)m_t + \beta \int_t ^Tm_s e^{-\beta (s-t)}ds+(1-N e^{-\beta t})\overline{x}-N \beta \int_0^{t} e^{-\beta (t-s)}m_sds\\
        =\beta \int_0^Te^{-\beta s}m_sds, \ \ \ \forall t\in[0,T].
\end{multline*}
While we already know that $m$ is continuous (by the continuity of $\overline{X}$ and dominated convergence), this equation further shows that $m_t$ is differentiable. Define
\[y_t:=\int_t ^Te^{-\beta (s-t)}m_s ds \ \ \ \text{and} \ \ \ \ell_t:=\int_0^{t} e^{-\beta (t-s)}m_sds,\]
so that our equation becomes
\begin{align*}
        (N-1)m_t + \beta y_t +(1-N e^{-\beta t})m_0-N \beta \ell_t=\beta y_0, \ \ \ \forall t\in[0,T].
\end{align*}
Differentiating gives a system of ODEs,
\[\begin{cases}
    \dot{m}_t =-\frac{\beta}{N-1} \left(\dot{y}_t +N (e^{-\beta t}\overline{x}- \dot{\ell}_t)\right), & m_0 = \overline{x}, \ \ \ m_T =0,\\
    \dot{y}_t = \beta y_t -m_t, & y_T = 0,\\
    \dot{\ell}_t = -\beta \ell_t +m_t, & \ell_0 = 0.
\end{cases}\]
Ignoring the boundary condition $m_0=\overline{x}$ and solving the system gives the unique solution
\[m_t=\frac{N\left(\left(\beta  \left(T -t \right) N +2+\beta  \left(T -t \right)\right) e^{\frac{\beta  \left(N +1\right) T}{N -1}}-2 e^{\frac{\beta  \left(N +1\right) t}{N -1}}\right)}{\left(\left(\beta  T +1\right) N +\beta  T +3\right) N e^{\frac{\beta  \left(N +1\right) T}{N -1}}-N +1}\overline{x}, \ \ \ t\in[0,T].
\]
If we now enforce that $m_0 = \overline{x}$ we get the necessary condition
\[\frac{\left(\left(N \beta  T +\beta  T +2\right) e^{\frac{\beta  \left(N +1\right) T}{N -1}}-2\right) N}{\left(\left(\beta  T +1\right) N +\beta  T +3\right) N e^{\frac{\beta  \left(N +1\right) T}{N -1}}-N +1}\overline{x}=\overline{x}.\]
If $\overline{x}\not=0$ this can only be true if 
\[\left(\left(N \beta  T +\beta  T +2\right) e^{\frac{\beta  \left(N +1\right) T}{N -1}}-2\right) N=\left(\left(\beta  T +1\right) N +\beta  T +3\right) N e^{\frac{\beta  \left(N +1\right) T}{N -1}}-N +1\]
which, for $N>0$, is equivalent to
\[Ne^{\frac{\beta (N+1)}{N-1} T} = -1.\]
However, the left-hand side is clearly positive, and we conclude that $\overline{x}=0$.

\paragraph*{Case (3):}
For case (3) it suffices to show that if $\vartheta_T\not=\frac{\lambda}{2}$ and a Nash equilibrium exists, then $x^i=x^j$ for all $i,j$.
 
Proposition \ref{prop:no.interior.jumps} tells us that there are no interior jumps and provides us a system of equations for $\Delta X_{0}^{*,i}$ and $\Delta X_{T}^{*,i}$.
This is the same system of equations as in Lemma \ref{lem:bdy.sys.of.eqns}. In this case $\vartheta_T \not=\frac{\lambda}{2}$, so we can conclude that
$\Delta X_T^{*,i} =0$ for all $i$ and $\Delta X_0^{*,i} = \Delta X_0^{*,j}$ for all $i,j$. Then, by Lemma \ref{lem:gen.foc.jumps}, we have that for all predictable $\sigma\geq \tau$,
\begin{multline*}
     \mathbb{E}\left[\lambda\int_0^Te^{-\beta |\sigma-t|}dX_t^{*,i}+\lambda\int_0^{\sigma-} e^{-\beta (\sigma-t)}\sum_{j\not=i}dX_t^{*,j} + \left(\vartheta_0\Delta X_0^{*,i}+\frac{\lambda}{2}\sum_{j\not=i}\Delta X_0^{*,j}\right)\mathds{1}_{\{\sigma=0\}}\bigg|\mathcal{F}_{\tau-}\right]\\
     =\mathbb{E}\left[\lambda\int_0^Te^{-\beta |\tau-t|}dX_t^{*,i}+\lambda\int_0^{\tau-} e^{-\beta (\tau-t)}\sum_{j\not=i}dX_t^{*,j} + \left(\vartheta_0\Delta X_0^{*,i}+\frac{\lambda}{2}\sum_{j\not=i}\Delta X_0^{*,j}\right)\mathds{1}_{\{\tau=0\}}\bigg|\mathcal{F}_{\tau-}\right]
 \end{multline*}
 for all $i=1,\dots,N$. 
 Appealing to Proposition \ref{prop:no.interior.jumps} once again, 
 \[ \vartheta_0\Delta X_{0}^{*,i}=\frac{\lambda}{2}\sum_{j\not=i}\Delta X_{0}^{*,j}, \ \ \ i=1,\dots,N.\]
 If we substitute this in and simplify,
 \begin{align*}
     \mathbb{E}&\left[\lambda\int_0^Te^{-\beta |\sigma-t|}dX_t^{*,i}+\lambda\int_0^{\sigma} e^{-\beta (\sigma-t)}\sum_{j\not=i}dX_t^{*,j}\bigg|\mathcal{F}_{\tau-}\right]\\
     &\quad \quad =\mathbb{E}\left[\lambda\int_0^Te^{-\beta |\tau-t|}dX_t^{*,i}+\lambda\int_0^{\tau} e^{-\beta (\tau-t)}\sum_{j\not=i}dX_t^{*,j}\bigg|\mathcal{F}_{\tau-}\right], \ \ \ i=1,\dots,N.
 \end{align*}
 Then, by subtracting any two equations when $i\not=k$,
 \begin{align*}
     \mathbb{E}&\left[\lambda\int_0^Te^{-\beta |\sigma-t|}d(X_t^{*,i}-X_t^{*,k})-\lambda\int_0^{\sigma} e^{-\beta (\sigma-t)}d(X_t^{*,i}-X_t^{*,k})\bigg|\mathcal{F}_{\tau-}\right]\\
     &=\mathbb{E}\left[\lambda\int_0^Te^{-\beta |\tau-t|}d(X_t^{*,i}-X_t^{*,k})-\lambda\int_0^{\tau} e^{-\beta (\tau-t)}d(X_t^{*,i}-X_t^{*,k})\bigg|\mathcal{F}_{\tau-}\right].
 \end{align*}
 This implies that
 \begin{align*}
  \mathbb{E}&\left[\lambda\int_{\sigma+}^Te^{-\beta |\sigma-t|}d(X_t^{*,i}-X_t^{*,k})\bigg|\mathcal{F}_{\tau-}\right]=\mathbb{E}\left[\lambda\int_{\tau+}^Te^{-\beta |\tau-t|}d(X_t^{*,i}-X_t^{*,k})\bigg|\mathcal{F}_{\tau-}\right].
 \end{align*}
 Summing over $k\not=i$,
 \begin{align*}
     \mathbb{E}&\left[\lambda\int_{\sigma+}^Te^{-\beta |\sigma-t|}d(N X_t^{*,i}-\sum_{k=1}^NX_t^{*,k})\bigg|\mathcal{F}_{\tau-}\right]=\mathbb{E}\left[\lambda\int_{\tau+}^Te^{-\beta |\tau-t|}d(NX_t^{*,i}-\sum_{k=1}^NX_t^{*,k})\bigg|\mathcal{F}_{\tau-}\right],
 \end{align*}
 which is equivalent to
 \begin{align*}
     \mathbb{E}&\left[\lambda\int_{\sigma+}^Te^{-\beta (t-\sigma)}d(X_t^{*,i}-\overline{X}_t)\bigg|\mathcal{F}_{\tau-}\right]=\mathbb{E}\left[\lambda\int_{\tau+}^Te^{-\beta (t-\tau)}d(X_t^{*,i}-\overline{X}_t)\bigg|\mathcal{F}_{\tau-}\right]
 \end{align*}
 where $\overline{X} = N^{-1}\sum_{i=1}^NX_t^{*,i}$.
 Setting $\sigma\equiv T$ and using predictability we get
 \begin{equation}\label{eqn:block.cost.case.3.v1}
     0=e^{\beta \tau}\mathbb{E}\left[\lambda\int_{\tau+}^Te^{-\beta t}d(X_t^{*,i}-\overline{X}_t)\bigg|\mathcal{F}_{\tau-}\right].
 \end{equation}
 Note that the process $X^{*,i} - \overline{X}$ has no jumps since $\Delta X_0^{*,i} =\Delta \overline{X}_0$ for all $i=1,\dots,N$. Hence, \eqref{eqn:block.cost.case.3.v1} implies that
 \begin{equation}\label{eqn:block.cost.case.3.v2}
     0=\mathbb{E}\left[\lambda\int_{\tau}^Te^{-\beta t}d(X_t^{*,i}-\overline{X}_t)\bigg|\mathcal{F}_{\tau-}\right].
 \end{equation}
 Let us define 
 \[M_t =\mathbb{E}\left[\lambda\int_{0}^Te^{-\beta s}d(X_s^{*,i}-\overline{X}_s)\bigg|\mathcal{F}_{t}\right], \ \ \ t\in[0,T].\]
 By definition, $M$ is a right-continuous martingale with respect to the filtration $\mathbb{F}$. Taking limits from below and applying \cite[Theorem 4.6.10]{Durrett.19},
 \begin{align*}
     M_{t-} &=\mathbb{E}\left[\lambda\int_{0}^Te^{-\beta s}d(X_s^{*,i}-\overline{X}_s)\bigg|\mathcal{F}_{t-}\right], \ \ \ t\in(0,T].
 \end{align*}
Using \eqref{eqn:block.cost.case.3.v2} and the predictability of the inventory processes,
\[M_{t-}=\lambda\int_{0}^te^{-\beta s}d(X_s^{*,i}-\overline{X}_s), \ \ \ t\in(0,T].\]
From this, we can infer that $M$ is continuous on $[0,T)$ and takes the form
\[M_{t}=\lambda\int_{0}^te^{-\beta s}d(X_s^{*,i}-\overline{X}_s).\]
The definition of $M_T$ reveals that the continuity holds on the entire interval $[0,T]$. Thus $M$ is a bounded variation and continuous martingale, implying that it must be constant. Hence,
 \[0=dM_t = \lambda e^{-\beta t} d(X_t^{*,i}-\overline{X}_t) \]
 from which we conclude that $dX^{*,i}_t = d\overline{X}_t$.
 Since $X^{*,i}_T=\overline{X}_T=0$ (recall that inventory liquidation is enforced by $T$) we have
 \[0=X_T^{*,i}-\overline{X}_T = x^i-\overline{x} + \int_0^Td(X_t^{*,i}-\overline{X}_t)=x^i-\overline{x}. \]
 From this we conclude that $x^i = \overline{x}$ for all $i$. Hence, $x^i=x^j$ for all $i,j$. 

\paragraph*{Case (4):}
Suppose that $\vartheta_0\not=\frac{\lambda(N-1)}{2}$ and 
$\vartheta_T\not=\frac{\lambda}{2}$. If a Nash equilibrium exists, then 
the argument for case~(2) yields $\overline{x}=0$, while the argument for 
case~(3) yields $x^i=x^j$ for all $i,j$. Hence $x^i=0$ for all $i$. 
This proves case~(4) and completes the proof of Theorem \ref{thm:equil.block.cost}. \qed

\subsection{Corollary \ref{cor:eq.cost.block.cost}}
From Theorem \ref{thm:equil.block.cost}, we can differentiate $X^{*,i}$ on $(0,T)$ to get
\begin{equation*}
		\dot{X}^{*,i}_t= -\frac{\beta}{\beta T+1} (x^{i}-\overline{x})-\frac{\beta N  (N+1) e^{\beta \frac{N+1}{N-1} T}+ \frac{2\beta N(N+1)}{N-1}e^{\beta \frac{N+1}{N-1} t}}{N((\beta T+1)(N+1)+2)e^{\beta \frac{N+1}{N-1} T}-(N-1)} \overline{x}.
\end{equation*}
Moreover, we have 
\[\Delta X_0^{*,i} = -\frac{(N+1)\left(1+Ne^{\beta \frac{N+1}{N-1}T}\right)}{N((\beta T+1)(N+1)+2)e^{\beta \frac{N+1}{N-1} T}-(N-1)}\overline{x}, \ \ \ \Delta X_T^{*,i} = -\frac{x^i-\overline{x}}{1+\beta T}.\]
From this we can obtain the equilibrium impact process,
\[I_t=-\frac{ \lambda  N (N +1)\left(e^{\beta \frac{N +1}{N -1}t}+N e^{\beta \frac{N +1}{N -1}T}\right) }{N \left(( \beta T +1) (N+1) +2\right) e^{\beta \frac{\left(N +1\right) T}{N -1}}-(N -1)}\overline{x}, \ \ \ t\in[0,T),
\]
where
\[\Delta I_0 =  \lambda \sum_{i=1}^N \Delta X_0^{*,i}=-\frac{\lambda N(N+1)\left(1+Ne^{\beta \frac{N+1}{N-1}T}\right)}{N((\beta T+1)(N+1)+2)e^{\beta \frac{N+1}{N-1} T}-(N-1)}\overline{x}\]
and
\[\Delta I_T=\lambda\sum_{i=1}^N\Delta X_T^{*,i} = 0.\]
Now, a direct computation (omitted for the sake of brevity) gives Corollary \ref{cor:eq.cost.block.cost}. \qed

\section{Proofs for Section \ref{sec:id.limit}}\label{app:proofs.for.sec.id.limit}

\subsection{Theorem \ref{thm:small.eps.lim}}

We divide the proof into two steps, corresponding to the two claims in the theorem.

\subsubsection*{Step 1: Strategy convergence.}

Note that $z_1\to \frac{N+1}{N-1}\beta$ as $\varepsilon\downarrow 0$. Writing $\tau_1=  \frac{N+1}{N-1}\beta$ we conclude
\begin{equation}\label{eqn:small.eps.1}
\frac{e^{z_1 T}-1}{z_1}\to \frac{e^{\tau_1 T}-1}{ \tau_1}.
\end{equation}
On the other hand, we see that $z_2\to -\infty$ as $\varepsilon\downarrow0$. So,
\[\gamma_1=\frac{1}{z_1+\beta}+\frac{1}{z_1-\beta}e^{z_1 T}\to\frac{1}{\tau_1+\beta}+\frac{1}{\tau_1-\beta}e^{\tau_1 T}, \]
\[\gamma_2=\frac{1}{z_2+\beta}+\frac{1}{z_2-\beta}e^{z_2 T}\to 0, \ \ \ \text{and} \ \ \ \frac{e^{z_2 t} -1}{ z_2}\to 0,  \ \ \ t\in(0,T].\]
The limit of 
$\gamma_1(e^{z_2 t} -1)/(\gamma_2z_2)$
therefore depends on the ratio of these last two terms. We have
\[z_2\gamma_2=\frac{z_2}{z_2+\beta}+\frac{z_2}{z_2-\beta}e^{z_2 T}\to 1, \ \ \ \text{and} \ \ \ e^{z_2 t} -1\to -1, \ \ \ t\in(0,T].\]
So,
\begin{equation*}%
\frac{\gamma_1}{\gamma_2}\frac{e^{z_2 t} -1}{ z_2} \to -\left(\frac{1}{\tau_1+\beta}+\frac{1}{\tau_1-\beta}e^{\tau_1 T}\right), \ \ \ t\in(0,T],
\end{equation*}
and the limit is $0$ if $t=0$.
Arguing similarly,
\begin{equation*}%
    \rho_{-}=\frac{1}{z_1-\beta}e^{z_1T}-\frac{\gamma_1}{\gamma_2} \frac{1}{z_2-\beta}e^{z_2T}\to \frac{1}{\tau_1-\beta} e^{\tau_1 T}.
\end{equation*}
Finally, we see that $z_3\uparrow \infty$ as $\varepsilon\downarrow 0$.
Moreover,
\[\varepsilon z_3 = \beta\varepsilon+\lambda \to \lambda, \ \ \ \text{and} \ \ \ \frac{e^{z_{3} t}-1}{e^{z_3 T}}=e^{-z_3(T-t)}-e^{-z_3 T}\to 0,\ \ \ t\in[0,T).\]
When $t=T$, the latter converges to $1$. Therefore,
\begin{equation}\label{eqn:small.eps.4}
    \frac{\lambda(e^{z_3 t}-1)}{\varepsilon z_3e^{z_3 T}}\to 0, \ \ \ t\in[0,T), \ \ \ \text{and} \ \ \ \frac{\lambda(e^{z_3 T}-1)}{\varepsilon z_3e^{z_3 T}}\to 1.
\end{equation}
Applying \eqref{eqn:small.eps.1}--\eqref{eqn:small.eps.4} yields that as $\varepsilon\downarrow0$,
\begin{align*}
		\mathfrak{f}_t &\to 1-\frac{\beta t}{1+\beta T} , \ \ \ t\in[0,T),\\
		\mathfrak{f}_T &\to 0,  \\
    \mathfrak{g}_t &\to 1-\frac{\frac{\beta}{\tau_1-\beta} e^{\tau_1 T}t+\frac{e^{\tau_1 t}-1}{\tau_1}+\frac{1}{\tau_1+\beta}+\frac{1}{\tau_1-\beta}e^{\tau_1 T}}{\frac{\beta}{\tau_1-\beta} e^{\tau_1 T}T+\frac{e^{\tau_1 T}-1}{\tau_1} +\frac{1}{\tau_1+\beta}+\frac{1}{\tau_1-\beta}e^{\tau_1 T}}, \ \ \ t\in(0,T],\\
    \mathfrak{g}_0 &\to 1.
\end{align*}
Rearranging these expressions shows agreement with Theorem \ref{thm:equil.block.cost} on $(0,T)$. 

It remains to upgrade the pointwise convergence to locally uniform on $(0,T)$. To this end, we first claim that $\mathfrak{f}$ and $\mathfrak{g}$ are monotone decreasing. As $z_3\geq0$ for all $\varepsilon$ we have
\begin{equation*}
		\dot{\mathfrak{f}}_t=-\frac{\beta+\frac{\lambda e^{z_{3} t}}{\varepsilon e^{z_3T}}}{\beta T+ \frac{\lambda(e^{z_3 T}-1)}{\varepsilon z_3e^{z_3 T}}}\leq 0.
\end{equation*}
Turning to $\dot{\mathfrak{g}}$ we claim that
\begin{equation}\label{eqn:dg.negative}
		\dot{\mathfrak{g}}_t=-\frac{\beta \rho_{-}+e^{z_1 t}-\frac{\gamma_1}{\gamma_2}e^{z_2 t}}{\beta \rho_{-}T+\frac{e^{z_1 T}-1}{z_1} -\frac{\gamma_1}{\gamma_2}\frac{e^{z_2 T}-1}{ z_2}}\leq 0.
\end{equation}
Indeed, from the form of $z_1$ and $z_2$ we can see that $z_1>0$, $z_2<0$, and $z_1$ increases and $z_2$ decreases as $\varepsilon\downarrow0$. Taking $\varepsilon\uparrow\infty$ we deduce
\[z_1 >\lim_{\varepsilon\uparrow\infty} z_1=\beta, \ \ \ z_2< \lim_{\varepsilon\uparrow\infty}z_2=-\beta.\]
These observations allow us to conclude $\gamma_1\geq0$ and $\gamma_2\leq 0$. Rearranging the expression for~$\rho_{-}$, we get
\[\rho_{-}=e^{z_{1} T}\frac{\left(z_{2}+\beta \right) \left(z_{1}-\beta \right) e^{(z_{2}-z_1) T}- \left(z_{2}-\beta \right) \left(z_{1}+\beta \right)}{\left(z_{1}+\beta \right) \left(-\left(z_{2}+\beta \right) e^{z_{2} T}-(z_2-\beta)\right) \left(z_{1}-\beta \right)}.
\]
The denominator is positive by the bounds on $z_1$ and $z_2$. As $(z_2+\beta)(z_1-\beta)\leq 0$ and $z_2-z_1\leq 0$ 
we have
\begin{align*}\rho_{-}&\geq e^{z_1 T}\frac{\left(z_{2}+\beta \right) \left(z_{1}-\beta \right)- \left(z_{2}-\beta \right) \left(z_{1}+\beta \right)}{\left(z_{1}+\beta \right) \left(-\left(z_{2}+\beta \right) e^{z_{2} T}-(z_2-\beta)\right) \left(z_{1}-\beta \right)}\\
&=e^{z_1 T}\frac{2 \beta (z_1-z_2)}{\left(z_{1}+\beta \right) \left(-\left(z_{2}+\beta \right) e^{z_{2} T}-(z_2-\beta)\right) \left(z_{1}-\beta \right)}\geq0.
\end{align*}
Taken together we see that each term in the numerator and denominator of $\dot{\mathfrak{g}}$ is positive and so \eqref{eqn:dg.negative} follows.

In summary, the monotone functions $\mathfrak{f}$ and $\mathfrak{g}$ converge pointwise to $\mathbbm{f}$ and $\mathbbm{g}$ which are continuous on $(0,T)$. By Dini's second theorem, it follows that the convergence is locally uniform.

\subsubsection*{Step 2: Cost convergence.}

We first consider the limit of the equilibrium impact cost from Corollary \ref{cor:eq.cost.liq.constr}. Along the lines of Step~1, it can be checked that as $\varepsilon\downarrow 0$,
\begin{equation}\label{eqn:xi.limit}
    \Xi \to 1+\beta T,
\end{equation}
\begin{equation}\label{eqn:psi.limit}
\Psi \to \frac{\beta}{\tau_1-\beta} e^{\tau_1 T}T+\frac{e^{\tau_1 T}-1}{\tau_1} +\frac{1}{\tau_1+\beta}+\frac{1}{\tau_1-\beta}e^{\tau_1 T},
\end{equation}
\begin{align}\mathfrak{h}_1&= \rho_{-}(\Psi+\beta\varrho_1) +\mathfrak{r}_1\label{eqn:h1.limit}\\
&\to \frac{1}{\tau_1-\beta} e^{\tau_1 T}\left(\frac{\beta}{\tau_1-\beta} e^{\tau_1 T}T+\frac{e^{\tau_1 T}-1}{\tau_1} +\frac{1}{\tau_1+\beta}+\frac{1}{\tau_1-\beta}e^{\tau_1 T}+\beta \frac{e^{\tau_1T} -1}{\tau_1(\tau_1 +\beta)}\right)\nonumber\\
&\quad +\frac{e^{2\tau_1 T}-1}{2\tau_1(\tau_1+\beta)}-\frac{1}{2}\left(\frac{1}{\tau_1+\beta}+\frac{1}{\tau_1-\beta}e^{\tau_1 T}\right)^2+\frac{1}{(\tau_1+\beta)^2}+\frac{1}{(\tau_1-\beta)(\tau_1+\beta)}e^{\tau_1 T},\nonumber
\end{align}
and
\begin{equation}\label{eqn:h2.limit}\mathfrak{h}_2=\rho_{-}\Xi+\beta\varrho_1 + \lambda\varepsilon^{-1}\mathfrak{m}_1\to \frac{1}{\tau_1-\beta} e^{\tau_1 T}(1+\beta T)+\beta \frac{e^{\tau_1T} -1}{\tau_1(\tau_1 +\beta)}+\frac{e^{\tau_1 T}}{\tau_1 +\beta}.
\end{equation}
Using \eqref{eqn:psi.limit} and \eqref{eqn:h2.limit}, we get (after simplification) that $\mathfrak{h}_2\Psi^{-1}\to 1$, so
\begin{equation}\label{eqn:h2.ratio.limit}
    \frac{\mathfrak{h}_2}{\Psi\Xi}\to\frac{1}{1+\beta T}.
\end{equation}
A protracted algebraic manipulation of the quantities in \eqref{eqn:psi.limit} and \eqref{eqn:h1.limit} also yields
\begin{equation}\label{eqn:h1.ratio.limit}
    \frac{\mathfrak{h}_1}{\Psi^2}\to \frac{ N^2 (N+1)\left( \left(\left(\beta T +\frac{1}{2}\right) (N+1) +3\right) e^{\frac{2 \left(N +1\right) \beta  T}{N -1}}-\frac{2 \left(N -1\right)}{N^2} \left(N e^{\frac{\left(N +1\right) \beta  T}{N -1}}+\frac{1}{4}\right) \right)}{\left(N \left(\left(\beta T +1\right)(N+1) + 2\right) e^{\frac{\left(N +1\right) \beta  T}{N -1}}-(N -1)\right)^{2}}.
\end{equation}
Passing to the limit in the expression of the impact cost in Corollary \ref{cor:eq.cost.liq.constr} and applying  \eqref{eqn:h2.ratio.limit} and  \eqref{eqn:h1.ratio.limit} gives the desired convergence to the corresponding term in Corollary \ref{cor:eq.cost.block.cost}.

We now turn to the instantaneous cost in Corollary \ref{cor:eq.cost.liq.constr}. Once again, a careful accounting allows us to extract the limits
\begin{equation}\label{eqn:h3.limit}\varepsilon \mathfrak{h}_3=\varepsilon\beta\rho_{-}(\Psi+\varrho_0)+\varepsilon\mathfrak{r}_0\to \frac{\lambda(N-1)}{2} \left(\frac{1}{\tau_1+\beta}+\frac{1}{\tau_1-\beta}e^{\tau_1 T}\right)^2,
\end{equation}
\begin{equation}\label{eqn:h4.limit}
\varepsilon\mathfrak{h}_4=\varepsilon\beta(2\Xi -\beta T)+\lambda^2 \varepsilon^{-1}e^{-z_3 T}\mathfrak{q}_3\to \frac{\lambda}{2},
\end{equation}
and
\begin{equation}\label{eqn:h5.limit}
\varepsilon\mathfrak{h}_5=\varepsilon\beta\rho_{-}\Xi+\varepsilon\beta\varrho_0+\lambda\mathfrak{m}_0\to0.
\end{equation}
Using \eqref{eqn:psi.limit} and \eqref{eqn:h3.limit}, we obtain after simplification,
\begin{equation}\label{eqn:h3.ratio.limit}
    \frac{\varepsilon \mathfrak{h}_3}{2\Psi^2}\to \frac{\lambda(N-1) (N+1)^2 (1+ N e^{\beta \frac{N+1}{N-1} T})^2}{4 \left(N((\beta T+1) (N+1) + 2)e^{\beta \frac{N+1}{N-1} T} - (N-1)\right)^2}.
\end{equation}
Appealing to \eqref{eqn:xi.limit}, \eqref{eqn:psi.limit}, \eqref{eqn:h4.limit}, and \eqref{eqn:h5.limit} we similarly find
\begin{equation}\label{eqn:h4.and.h5.ratio.limit}
    \frac{\varepsilon \mathfrak{h}_4}{2\Xi^2}\to \frac{\lambda}{4 (\beta T+1)^2} \ \ \ \text{and} \ \ \ \frac{\varepsilon\mathfrak{h}_5}{\Xi\Psi}\to 0.
\end{equation}
Applying the limits \eqref{eqn:h3.ratio.limit} and \eqref{eqn:h4.and.h5.ratio.limit} to the instantaneous cost expression in Corollary~\ref{cor:eq.cost.liq.constr} shows that this cost converges to the block cost in Corollary \ref{cor:eq.cost.block.cost} for the specific choice $\vartheta_0 =\frac{\lambda (N-1)}{2}$ and $\vartheta_T = \frac{\lambda}{2}$. 

It remains to sharpen this conclusion in order to show that for any $\delta\in(0,T)$,
\begin{align*}
     \varepsilon \int_0^{\delta}(\dot{X}^{*,\eps,i}_t)^2 dt \to   \vartheta_0(\Delta X^{*,0,i}_0)^2 \quad\mbox{and}\quad \varepsilon \int_{\delta}^{T}(\dot{X}^{*,\eps,i}_t)^2 dt \to   \vartheta_T(\Delta X^{*,0,i}_T)^2.
\end{align*} 
Observe from \cref{thm:equil.liq.constr} that
\[(\dot{X}_t^{*,\varepsilon,i})^2=\frac{h^3_t}{\Psi^2}\overline{x}^2+\frac{h^4_t}{\Xi^2}(x^{i}-\overline{x})^2+\frac{2h^5_t}{\Xi\Psi}\overline{x}(x^{i}-\overline{x})\]
for the same functions
\[h^3_t=\left[\beta\rho_{-}+e^{z_1 t}- \frac{\gamma_1}{\gamma_2} e^{z_2 t}\right]^2, \ \ \ h^4_t=\left[\beta+\frac{\lambda e^{z_{3} t}}{\varepsilon e^{z_3T}}\right]^2,\]
\[h^5_t=\left[\beta\rho_{-}+e^{z_1 t}- \frac{\gamma_1}{\gamma_2} e^{z_2 t}\right]\left[\beta+\frac{\lambda e^{z_{3} t}}{\varepsilon e^{z_3T}}\right]\]
that arise in the proof of Corollary \ref{cor:eq.cost.term.pen}. Letting $\mathfrak{h}_i^{a,b}:=\int_a^b h_t^idt$ for $i=3,4,5$ we see that
\begin{equation}\label{eqn:tail.est.L2.der} \varepsilon \int_\delta^{T}(\dot{X}^{*,\eps,i}_t)^2 dt=\frac{\eps\mathfrak{h}^{\delta,T}_3}{\Psi^2}\overline{x}^2+\frac{\eps \mathfrak{h}^{\delta,T}_4}{\Xi^2}(x^{i}-\overline{x})^2+\frac{2\eps \mathfrak{h}^{\delta,T}_5}{\Xi\Psi}\overline{x}(x^{i}-\overline{x}),
\end{equation}
and the analogous expression holds for $\int_{0}^{\delta}(\dot{X}^{*,\eps,i}_t)^2 dt$. 

To prove the claim it suffices to show the limit
\[\varepsilon\int_{\delta}^{T}(\dot{X}^{*,\eps,i}_t)^2 dt \to   \vartheta_T(\Delta X^{*,0,i}_T)^2=\frac{\vartheta_T (x^i-\overline{x})^2}{(\beta T+1)^2}\]
as, by splitting $\int_{0}^{T}(\dot{X}^{*,\eps,i}_t)^2 dt$, the remaining limit immediately follows from the convergence of the full instantaneous cost to the full block cost that we have already shown. To this end we must evaluate the $\mathfrak{h}_i^{\delta,T}$. 
By expanding the product form of the $h_t^i$ and integrating we get
\begin{align*}
    \mathfrak{h}_3^{\delta,T} &=\beta^2\rho_{-}^2(T-\delta)+2\beta\rho_{-}\left[\frac{e^{z_1 T}-e^{z_1\delta}}{z_1}- \frac{\gamma_1}{\gamma_2} \frac{e^{z_2 T}-e^{z_2\delta}}{z_2}\right]\\
    &\quad +\frac{e^{2z_1T}-e^{2z_1\delta}}{2z_1}+\frac{\gamma_1^2}{\gamma_2^2}\frac{e^{2z_2T}-e^{2z_2\delta}}{2z_2}-2\frac{\gamma_1}{\gamma_2}\frac{e^{(z_1+z_2)T}-e^{(z_1+z_2)\delta}}{z_1+z_2},
\end{align*}
\begin{align*}
    \mathfrak{h}_4^{\delta,T} &=\beta^2(T-\delta)+\frac{2\beta\lambda (e^{z_3 T}-e^{z_3\delta})}{\varepsilon z_3 e^{z_3T}}+\frac{\lambda^2 (e^{2z_3 T}-e^{2z_3\delta})}{2\varepsilon^2 z_3 e^{2 z_3 T}}, \ \ \ \text{and}
\end{align*}
\begin{align*}
    \mathfrak{h}_5^{\delta,T}&=\beta^2 \rho_{-}(T-\delta)+\beta\left[\frac{e^{z_1 T}-e^{z_1\delta}}{z_1}- \frac{\gamma_1}{\gamma_2} \frac{e^{z_2 T}-e^{z_2\delta}}{z_2}\right]\\
    &\quad +\frac{\beta\lambda\rho_{-}(e^{z_3T}-e^{z_3\delta})}{\varepsilon z_3e^{z_3T}}+\frac{\lambda}{\varepsilon e^{z_3 T}}\left[\frac{e^{(z_1+z_3)T}-e^{(z_1+z_3)\delta}}{z_1+z_3}-\frac{\gamma_1}{\gamma_2}\frac{e^{(z_2+z_3)T}-e^{(z_2+z_3)\delta}}{z_2+z_3}\right].
\end{align*}
Once again, carefully passing to the limit in a fashion analogous to \eqref{eqn:h3.limit}--\eqref{eqn:h5.limit} yields
\begin{equation}\label{eqn:tail.constants} \varepsilon \mathfrak{h}_3^{\delta,T} \to 0, \ \ \ \varepsilon \mathfrak{h}_4^{\delta,T} \to \frac{\lambda}{2}, \ \ \ \text{and} \ \ \ \varepsilon \mathfrak{h}_5^{\delta,T} \to 0.
\end{equation}
Passing to the limit across \eqref{eqn:tail.est.L2.der} and combining \eqref{eqn:tail.constants} with \eqref{eqn:xi.limit} and \eqref{eqn:psi.limit} completes the proof.\qed

\newpage
\section{Table of Constants}\label{app:constants}

\begin{table}[!h]
    \centering
    \begin{tabular}{|c|l|}
        \hline
        Constant & Definition \\
        \hline
        $z_1$ & $\frac{-\lambda(N-1)+\sqrt{(N-1)^2\lambda^2+4\beta \varepsilon (N+1) \lambda +4\beta^2\varepsilon^2}}{2\varepsilon}$\\
        \hline
        $z_2$ & $\frac{-\lambda(N-1)-\sqrt{(N-1)^2\lambda^2+4\beta \varepsilon (N+1) \lambda +4\beta^2\varepsilon^2}}{2\varepsilon}$\\
        \hline
        $z_3$ & $\beta + \varepsilon^{-1} \lambda$ \\
        \hline
        $\gamma_i$ &  $\frac{1}{z_i+\beta}+\frac{1}{z_i-\beta}e^{z_i T}, \ \ \ i=1,2.$\\
        \hline
        $\mathfrak{b}_{i}$ & $\frac{e^{z_iT}-1}{z_i}$, \ \ \ $i=1,2,3$\\
        \hline
        $\mathfrak{q}_0$ & $\mathfrak{b}_3e^{-z_3T}$\\
        \hline
        $\mathfrak{q}_i$ & $\mathfrak{b}_{i3}e^{-z_3T}$, \ \ \ $i=1,2,3$\\
        \hline
        $\mathfrak{b}_{ij}$ & $\frac{e^{(z_i+z_j)T}-1}{z_i+z_j}$, \ \ \ $i,j=1,2,3$\\
        \hline

        $\varrho_0$ & $\mathfrak{b}_1-\frac{\gamma_1}{\gamma_2}\mathfrak{b}_2$\\ %
        \hline
        $\varrho_1$ & $\frac{\mathfrak{b}_1}{(z_1+\beta)}-\frac{\gamma_1}{\gamma_2}\frac{\mathfrak{b}_2}{(z_2+\beta)}$\\
        \hline
        $\rho_{0}$ & $e^{z_1T}-\frac{\gamma_1}{\gamma_2}e^{z_2T}$ \\
        \hline
        $\rho_{\pm}$ & $\frac{1}{z_1\pm\beta}e^{z_1T}-\frac{\gamma_1}{\gamma_2} \frac{1}{z_2\pm\beta}e^{z_2T}$ \\

        \hline
        $\mathfrak{m}_0$ &  $\mathfrak{q}_1-\frac{\gamma_1}{\gamma_2}\mathfrak{q}_2$\\
        \hline
        $\mathfrak{m}_1$  &  $\frac{\mathfrak{q}_1}{(z_1+\beta)}-\frac{\gamma_1}{\gamma_2}\frac{\mathfrak{q}_2}{(z_2+\beta)}$\\
        \hline
        $\mathfrak{r}_0$ & $\mathfrak{b}_{11}+\frac{\gamma_1^2}{\gamma_2^2}\mathfrak{b}_{22}-2\frac{\gamma_1}{\gamma_2}\mathfrak{b}_{12}$\\
        \hline
        $\mathfrak{r}_1$  & $\frac{\mathfrak{b}_{11}}{(z_1+\beta)}+\frac{\gamma_1^2}{\gamma_2^2}\frac{\mathfrak{b}_{22}}{(z_2+\beta)}-\frac{\gamma_1}{\gamma_2}\left[\frac{1}{z_1+\beta}+\frac{1}{z_2+\beta}\right]\mathfrak{b}_{12}$\\
        \hline
        $\mathfrak{p}$ & $\rho_0+\beta\rho_{-}+\varepsilon^{-1}\lambda N(\rho_{+}+\rho_{-})$\\
        \hline
        $\Psi$ & $\varrho_0 + \beta\rho_{-}T$\\
        \hline
        $\Xi$ & $\beta T+\lambda\varepsilon^{-1}\mathfrak{q}_0$\\ %
        \hline
        $\psi$ & $\mathfrak{p}+\varepsilon^{-1}\varphi \Psi$ \\ %
        \hline
        $\xi$ &  $z_3 + \varepsilon^{-1}\varphi\Xi$ \\ %
        \hline
        $\mathfrak{h}_1$ & $\rho_{-}(\Psi+\beta\varrho_1) +\mathfrak{r}_1$\\
        \hline
        $\mathfrak{h}_2$ & $\rho_{-}\Xi+\beta\varrho_1 + \lambda\varepsilon^{-1}\mathfrak{m}_1$\\
        \hline
        $\mathfrak{h}_3$ & $\beta\rho_{-}(\Psi+\varrho_0)+\mathfrak{r}_0$\\
        \hline
        $\mathfrak{h}_4$ & $\beta(2\Xi -\beta T)+\lambda^2 \varepsilon^{-2}e^{-z_3 T}\mathfrak{q}_3$ \\
        \hline
        $\mathfrak{h}_5$ & $\beta\rho_{-}\Xi+\beta\varrho_0+\lambda\varepsilon^{-1}\mathfrak{m}_0$\\
        \hline
    \end{tabular}
    \caption{Main Constants}
    \label{tab:constants}
\end{table}

\bibliographystyle{abbrv}
\bibliography{stochfin}

\end{document}